\documentclass[a4paper,draft=true,DIV=classic,fontsize=10pt]{scrartcl}
\usepackage[english]{babel}
\usepackage[left=3.5cm,top=3cm,right=3.5cm,bottom=4.8cm]{geometry}
\usepackage{amsmath,amssymb,amsthm}
\usepackage{enumerate}
\usepackage{graphicx}
\usepackage{multirow}
\usepackage{comment}

\newtheorem{remark}{Remark}
\newtheorem{theorem}{Theorem}
\newtheorem{corollary}{Corollary}
\newtheorem{proposition}{Proposition}
\newtheorem{lemma}{Lemma}

\newcommand{\R}{\mathbb{R}}
\newcommand{\eins}{\text{\ensuremath{1\hspace*{-0.9ex}1}}}

\usepackage{apacite}

\begin{document}
\title{Robust risk aggregation with neural networks}
\author{Stephan Eckstein\thanks{Department of Mathematics, University of Konstanz, Universit\"{a}tsstraße 10, 78464 Konstanz, Germany, stephan.eckstein@uni-konstanz.de} \and Michael Kupper\thanks{Department of Mathematics, University of Konstanz, Universit\"{a}tsstraße 10, 78464 Konstanz, Germany, kupper@uni-konstanz.de} \and Mathias Pohl\thanks{Faculty of Business, Economics \& Statistics, University of Vienna, Oskar-Morgenstern-Platz 1, 1090 Vienna, Austria, mathias.pohl@univie.ac.at}
}
\date{\today}

\maketitle

\begin{abstract}
  \noindent 
{\textbf{Abstract}.}
We consider settings in which the distribution of a multivariate random variable is partly ambiguous. We assume the ambiguity lies on the level of the dependence structure, and that the marginal distributions are known. Furthermore, a current best guess for the distribution, called reference measure, is available. We work with the set of distributions that are both close to the given reference measure in a transportation distance (e.g.~the Wasserstein distance), and additionally have the correct marginal structure. The goal is to find upper and lower bounds for integrals of interest with respect to distributions in this set.
 
The described problem appears naturally in the context of risk aggregation. When aggregating different risks, the marginal distributions of these risks are known and the task is to quantify their joint effect on a given system. This is typically done by applying a meaningful risk measure to the sum of the individual risks. For this purpose, the stochastic interdependencies between the risks need to be specified. In practice the models of this dependence structure are however subject to relatively high model ambiguity. 

The contribution of this paper is twofold: Firstly, we derive a dual representation of the considered problem and prove that strong duality holds. Secondly, we propose a generally applicable and computationally feasible method, which relies on neural networks, in order to numerically solve the derived dual problem. The latter method is tested on a number of toy examples, before it is finally applied to perform robust risk aggregation in a real world instance.  

\end{abstract}

\section{Introduction}
\subsection{Motivation}
\label{ssec:motivation}

Risk aggregation is the process of combining multiple types of risk within a firm. The aim is to obtain 
meaningful measures for the overall risk the firm is exposed to. The stochastic interdependencies between the different risk types are crucial in this respect. There is a variety of different approaches to model these interdependencies. One generally observes that these models for the dependence structure between the risk types are significantly less accurate than the models for the individual types of risk.

We take the following approach to address this issue: We assume that the distributions of the marginal risks are known and fixed. This assumption is justified in many cases of practical interest. Moreover, risk aggregation is per definition not concerned with the computation of the marginal risks' distributions. Additionally, we take a probabilistic model for the dependence structure linking the marginals risks as given. Note that there are at least two different approaches in the literature to specify this \emph{reference dependence structure}: The construction of copulas and factor models. The particular form of this reference model is not relevant for our approach as long as it allows us to generate random samples. Independently of the employed method, the choice of a reference dependence structure is typically subject to high uncertainty. Our contribution is to model the ambiguity with respect to the specified reference model, while fixing the marginal distributions.
We address the following question in this paper: 
\begin{center}
\noindent How can we account for model ambiguity with respect to a specific dependence structure when aggregating different risks?
\end{center}
We propose an intuitive approach to this problem: We compute the aggregated 
risk with respect to the worst case dependence structure in a neighborhood 
around the specified reference dependence structure.
For the construction of this neighborhood we use transportation distances. These distance measures between probability distributions are flexible enough to capture different kinds of \emph{model ambiguity}. At the same time, they allow us to generally derive numerical methods, which solve the corresponding problem of robust risk aggregation in reasonable time. 
To highlight some of the further merits of our approach, we are able to determine the worst case dependence structure for a problem at hand. Hence, our method for robust risk \emph{measurement} is arguably a useful tool also for risk \emph{management} as it provides insights about which scenarios stress a given system the most. Moreover, it should be emphasized that our approach is restricted neither to a particular risk measure nor a particular aggregation function.\footnote{Note also that our methods can be applied to solve completely unrelated problems, such as the portfolio selection problem under dependence uncertainty introduced in \citeA{pflug2017review}.}

In summary, the approach presented provides a flexible way to include model ambiguity in situations where a reference dependence structure is given and the marginals are fixed. It is generally applicable and computationally feasible.
In the subsequent subsection we outline our approach in some more details before discussing the related literature.

\subsection{Overview}
\label{ssec:Overview}
We aim to evaluate 
$$ \int_{\R^d} f d\bar{\mu},$$
for some $f:\R^d \to \R$ in the presence of ambiguity with respect to the probability measure $\bar{\mu} \in \mathcal{P}(\R^d)$, where $\mathcal{P}(\R^d)$ denotes the set of all Borel probability measures on $\R^d$. In particular, we assume that the marginals $\bar{\mu}_1, \dots, \bar{\mu}_d$ of $\bar{\mu}$ are known and the ambiguity lies solely on the level of the dependence structure. Moreover, we assume a reference dependence structure, namely the one implied by the reference measure $\bar{\mu}$, is given and that the degree of ambiguity with respect to the reference measure $\bar{\mu}$ can be modeled by the transportation distance $d_c$, which is defined in \eqref{eq:DefTransDistanceOverview} below. Hence, we consider the following problem 
\begin{align} \label{eq:ProblemPHI}
\phi(f):= \max_{\substack{\mu \in \Pi(\bar{\mu}_1,\dots,\bar{\mu}_d)\\  d_c(\bar{\mu},\mu)\le\rho}} \int_{\R^d} f\, d\mu,
\end{align}
where the set $\Pi(\bar{\mu}_1,\dots,\bar{\mu}_d)$ consists of all $\mu\in\mathcal{P}(\R^d)$
satisfying $\mu_i=\bar{\mu}_i$ for all $i=1,\dots,d$, where $\mu_i \in \mathcal{P}(\R)$ and $\bar{\mu}_i \in \mathcal{P}(\R)$ denote the $i$-th marginal distributions of $\mu$ and $\bar{\mu}$.
We fix a continuous function $c:\R^d \times \R^d \to[0,\infty)$ such that $c(x,x)=0$ for all $x\in \R$. The cost of transportation between $\bar{\mu}$ and $\mu$ in $\mathcal{P}(\R^d)$ with respect to the cost function $c$ is defined as 
\begin{align} \label{eq:DefTransDistanceOverview}
d_c(\bar{\mu},\mu) := \inf_{\pi \in \Pi(\bar{\mu},\mu)} \int_{\R^d \times \R^d} c\left(x,y\right) \pi(dx,dy),
\end{align}
where $\Pi(\bar{\mu},\mu)$ denotes the set of all couplings of the marginals $\bar{\mu}$ and $\mu$. For the cost function $c(x,y)=||x-y||^p$ with $p\ge 1$, the mapping $d_c^{1/p}$ corresponds to the Wasserstein distance of order $p$. 

The numerical methods to solve problem \eqref{eq:ProblemPHI}, which are developed in this paper, build on the following dual formulation of problem \eqref{eq:ProblemPHI}:
\begin{align}
\inf_{\lambda \geq 0,\, h_i \in C_b(\R)} \Big\{ \rho\lambda + \sum_{i=1}^d \int_{\R} h_i\, d\bar{\mu}_i  + \int_{\R^d} \sup_{y \in \R^d} \Big[ f(y) - \sum_{i=1}^d h_i(y_i) - \lambda c(x,y) \Big]\, \bar{\mu}(dx) \Big\}, \label{eq:DualFormulationIntro}
\end{align}
where $C_b(\R)$ the set of all continuous and bounded functions $h : \R \rightarrow \R$.
This dual formulation was initially derived by \citeA{gao2017data}. These authors show that strong duality holds, i.e.~problem \eqref{eq:ProblemPHI} and \eqref{eq:DualFormulationIntro} coincide, for upper semicontinuous functions $f: X \rightarrow \R$ satisfying the growth condition $\sup_{x \in X} \frac{f(x)}{c(x,y_0)} < \infty$ for some $y_0 \in X$, where $X = X_1 \times ... \times X_d$ for possibly non-compact subsets $X_1, ..., X_d$ of $\mathbb{R}$.

Theorem~\ref{thm:main} in Section~\ref{sec:results} extends the duality in the following aspects: Firstly, the functions $f: X \rightarrow \R$ need not satisfy a growth condition that depends on the cost $c$. Our results allow for upper semicontinuous functions of bounded growth. Secondly, we can consider a space $X = X_1 \times \dots \times X_d$, where $X_i$ can be arbitrary polish spaces. We emphasize that the problem setting can therefore include an information structure where multivariate marginals are known and fixed.
Lastly, Theorem \ref{thm:main} extends the constraint $d_c(\bar{\mu},\mu)\le\rho$ to a more general way of penalizing with respect to $d_c(\bar{\mu},\mu)$.

We now turn to the question how the dual formulation \eqref{eq:DualFormulationIntro} can be utilized to solve problem \eqref{eq:ProblemPHI}. One approach is to assume that the reference distribution $\bar{\mu}$ is a discrete distribution. In this context,
\citeA{gao2017data} show that the dual problem \eqref{eq:DualFormulationIntro} can be reformulated as a linear program under the following assumptions: First, the function $f$ can be written as the maximum of affine functions. Second, the reference distribution $\bar{\mu}$ is given by an empirical distribution on $n$ points $x^1,\dots, x^n$ in $\mathbb{R}^d$.
Third, the cost function $c$ has to be additively separable, i.e. $c(x, y) = \sum_{i=1}^d c_i(x_i, y_i)$.
For further details we refer to Corollary~\ref{coro:LPreformulation} in Section~\ref{ssec:duality}.

This linear programming approach is especially useful when 
only few observations are available to construct the reference distribution $\bar{\mu}$  - a case where accounting for ambiguity with respect to the dependence structure is often required. Nevertheless, the assumptions under which problem \eqref{eq:DualFormulationIntro} can be solved by means of linear programming exclude many cases of practical interest. Even in cases that linear programming is applicable, the resulting size of the linear program quickly becomes intractable in higher dimensions. 
Hence, this paper presents a generally applicable and computationally feasible method to numerically solve problem \eqref{eq:DualFormulationIntro} which uses neural networks.

The basic idea is to use neural networks to parametrize the functions $h_i \in C_b(\R)$ and then solve the resulting finite dimensional problem. Theoretically, such an approach is justified by the universal approximation properties of neural networks, see for example \citeA{hornik1991approximation}.

To utilize neural networks, we first dualize the point-wise supremum inside the integral of \eqref{eq:DualFormulationIntro}. Under mild assumptions, this leads to
\[
 \inf_{\substack{\lambda \geq 0,\\ h_i \in C_b(\R),\, g \in C_b(\R^d):\\ g(x) \geq f(y) - \sum_{i=1}^d h_i(y_i) - \lambda c(x,y)}} \Big\{ \lambda \rho + \sum_{i=1}^d \int_{\R} h_i \,d\bar{\mu}_i + \int_{\R^d} g \,d\bar{\mu} \Big\}.
\]
As the pointwise inequality constraint prevents a direct implementation with neural networks, the constraint is penalized. This is done by introducing a measure $\theta \in \mathcal{P}(\mathbb{R}^{2d})$, which we refer to as the \textsl{sampling measure}. Further, we are given a family of \textsl{penalty functions} $(\beta_\gamma)_{\gamma > 0}$ which increases the accuracy of the penalization for increasing $\gamma$, e.g.~$\beta_\gamma(x) = \gamma \max \lbrace 0,x \rbrace^2$. The resulting optimization problem reads 
\begin{align} \label{eq:PhiThetaGammaOverview}
\phi_{\theta,\gamma}(f) := \inf_{\substack{\lambda \geq 0,\\ h_i \in C_b(\R),\, g \in C_b(\R^d)}} \Big\{ \lambda \rho &+ \sum_{i=1}^d \int_{\R} h_i \,d\bar{\mu}_i + \int_{\R^d} g \,d\bar{\mu} \\[-10pt] 
+& \int_{\R^{2d}} \beta_{\gamma}\big(f(y) - \sum_{i=1}^d h_i(y_i) - \lambda c(x,y) - g(x)\big) \,\theta(dx,dy)\Big\}. \notag
\end{align}
Before we develop numerical methods to evaluate $\phi_{\theta,\gamma}(f)$ and thereby approximate $\phi(f)$, we need to study the convergence
\begin{align}\phi_{\theta,\gamma}(f) \rightarrow \phi(f) \quad \text{ for } \gamma \rightarrow \infty.\label{convergencegamma}\end{align}
A sufficient condition for this convergences is derived in Proposition \ref{prop:penal}. We additionally give a general instance where this derived condition is satisfied. It states that \eqref{convergencegamma} holds whenever the cost function $c$ satisfies a mild growth condition and the sampling measure $\theta$ is the product measure between the reference measure and the respective marginals, i.e.~$\theta = \bar{\mu} \otimes (\bar{\mu}_1 \otimes ... \otimes \bar{\mu}_d)$.

Besides the optimal value of problem \eqref{eq:ProblemPHI} also the corresponding optimizer is of interest.
To this end, we develop duality for problem \eqref{eq:PhiThetaGammaOverview}. This duality leads to a simple formula to obtain approximate  optimizers of the initial problem \eqref{eq:ProblemPHI} once the dual formulation \eqref{eq:PhiThetaGammaOverview} is solved. It shows that any optimizer $(\lambda^\star, (h^\star_i)_{i=1, ..., d}, g^\star)$ of \eqref{eq:PhiThetaGammaOverview} gives an approximate optimizer $\mu^\star$ of \eqref{eq:ProblemPHI} by setting $\mu^\star$ equal to the second marginal of $\pi^\star$, where $\pi^\star$ is defined by the Radon-Nikodym derivative
\begin{align}
		\frac{d\pi^\star}{d\theta}(x,y) := \beta_{\gamma}^\prime \Big(f(y) - g^\star(x) - \sum_{i=1}^d h^\star_i(y_i) - \lambda^\star c(x,y)\Big).
\end{align}

Problem $\phi_{\theta,\gamma}(f)$ fits into the standard framework in which neural networks can be applied to parametrize the functions $h_i \in C_b(\R)$ and $g \in C_b(\R^d)$. We justify this parametrization theoretically by giving conditions under which the approximation error vanishes for a infinite-size neural network. In Section \ref{sec:implementation}, we give details concerning the numerical solution of $\phi_{\theta,\gamma}(f)$ using neural networks, which encompasses the choice of the neural network structure, hyperparameters and optimization method.

This approach based on neural networks is the main reason to derive and study the penalized problem \eqref{eq:PhiThetaGammaOverview}. Nonetheless, problem \eqref{eq:PhiThetaGammaOverview} is interesting in its own right and by no means limited to the  application of neural networks: it may be efficiently solved using advanced first-order methods, see e.g.~\citeA{nesterov2012efficiency}. We thank an anonymous referee for pointing this out to us.

\vspace{5mm}
The remainder of the paper is structured as follows. In the subsequent Section \ref{ssec:literature}, we provide a overview of the relevant literature. Our main results can be found in Section \ref{sec:results}, which consists of three parts:  First, we state and prove the general form of the duality between \eqref{eq:ProblemPHI} and \eqref{eq:DualFormulationIntro} and derive some implications thereof. In the second part of Section \ref{sec:results}, we study the penalization introduced in equation \eqref{eq:PhiThetaGammaOverview} above. Third, we give conditions under which $\phi_{\theta,\gamma}(f)$ can be approximated with neural networks.
Section \ref{sec:implementation} gives implementation details.
Section \ref{sec:Examples} is devoted to four toy examples, which aim to shed some light on the developed concepts. In the final Section \ref{sec:DNB}, the acquired techniques are applied to a real world example. We thereby demonstrate how to implement robust risk aggregation with neural networks in practice.

\subsection{Related literature}
\label{ssec:literature}
There are three different strings of literature, which are relevant in the present context: Firstly, literature on risk aggregation; secondly, literature on model ambiguity and particularly on ambiguity sets constructed using the Wasserstein distance; thirdly, recent application of neural networks in finance and related optimization problems.
\subsubsection*{Risk aggregation}
In Section \ref{sec:DNB}, we motivate from an applied point of view why there is interest in risk bounds for the sum of losses of which the marginal distributions are known. The theoretical interest in this topic started with the following questions: How can one compute bounds for the distribution function of a sum of two random variables when the marginal distributions are fixed? This problem was solved in 1982 by \citeA{makarov1982estimates} and \citeA{ruschendorf1982random}. Starting with the work of \citeA{embrechts2006bounds} more than 20 years later, the higher dimensional version of this problem was studied extensively due to its relevance for risk management. We refer to \citeA{embrechts2015aggregation} and \citeA{puccetti2015extremal}  for an overview of the developments concerning \emph{risk aggregation under dependence uncertainty}, as this problem was coined. Let us mention that \citeA{puccetti2012computation} introduced the so-called \emph{rearrangement algorithm}, which is a fast procedure to numerically compute the bounds of interest. Applying this algorithm to real-world examples demonstrates a conceptual drawback of the assumption that no information concerning the dependence of the marginal risk is available: The implied lower and upper bound for the aggregated risk are impractically far apart.

Hence, some authors recently tried to overcome this drawback and to come up with more realistic bounds by including partial information about the dependence structure. For instance,  \citeA{puccetti2012bounds} discuss how positive, negative or independence information influence the above risk bounds; \citeA{bernard2017value} derive risk bounds with constraints on the variance of the aggregated risk;
\citeA{bernard2017risk} consider partially specified factor models for the dependence structure. The interested reader is referred to \citeA{ruschendorf2017risk} for a recent review of these and related approaches. Finally, we want to point out the intriguing contribution by \citeA{lux2016model}. These authors provide a framework which allows them to derive VaR-bounds if (a) extreme value information is available, (b) the copula linking the marginals is known on a subset of its domain and (c) the latter copula lies in the neighborhood of a reference copula as measured by a statistical distance. 

Since our paper aims to contribute to this string of literature, let us point out that the latter mentioned type of partial information about the dependence structure used in \citeA{lux2016model} is similar in spirit to our approach. We emphasize that Lux and Papapantoleon use statistical distances which are different to the transportation distance $d_c$ defined in the previous subsection.

\subsubsection*{Model Ambiguity}
There is an obvious connection of problem \eqref{eq:ProblemPHI}, which is studied in this paper, with the following minimax stochastic optimization problem
\begin{align} \label{eq:StochOptFramework}
\min_{x \in \mathbb{X}} \max_{Q \in \mathcal{Q}} \mathbb{E}^Q \left[ f(x, \xi) \right],\end{align}
where $\mathbb{X} \subset \R^m$, $f: \R^m \times \Xi \rightarrow \R$, $\xi$ is a random vector whose distribution $Q$ is supported on $\Xi \subset \R^d$ and $\mathcal{Q}$ is a nonempty set of probability distributions, referred to as \emph{ambiguity set}. Problems of this form recently became known as distributionally robust stochastic optimization problems. As pointed out by \citeA{shapiro2017distributionally}, there are two natural and somewhat different approaches to constructing the ambiguity set $\mathcal{Q}$. On the one hand, ambiguity sets have been defined by moment constraints, see ~\citeA{delage2010distributionally} and references therein. An alternative approach is to assume a reference probability distribution $\bar{Q}$ is given and define the ambiguity set by all distributions which are in the neighborhood of $\bar{Q}$ as measured by a statistical distance. 
To the best of our knowledge two distinct choices of this statistical distance have been established in the literature: The $\phi$-divergence and the Wasserstein distance. Concerning ambiguity sets constructed using the $\phi$-divergence we refer to\citeA{bayraksan2015data} and references therein. In the following, we focus on approaches which rely on the Wasserstein distance to account for model ambiguity. \citeA{pflug2007ambiguity} were the first to study these particular ambiguity sets. \citeA{esfahani2017data} showed that distributionally robust stochastic optimization problems over Wasserstein balls centered around a discrete reference distribution possess a tractable reformulation: under mild assumptions these problems belong to the same complexity class as their non-robust counterparts. The duality result driving this insight was also proven by \citeA{blanchet2016quantifying},\citeA{gao2016distributionally} and \citeA{bartl2017computational} based on different techniques and assumptions. These contribution indicate that distributionally robust stochastic optimization using the Wasserstein distance developed into an active field of research in recent years. For instance, \citeA{zhao2018datadriven} and \citeA{hanasusanto2016conic} adapted similar ideas in the context of two-stage stochastic programming and \citeA{chen2018distributionally} and \citeA{yang2017convex} study distributionally robust Markov decision processes using the Wasserstein distance. \citeA{obloj2018statistical} analyze a robust estimation method for superhedging prices relying on a Wasserstein ball around the empirical measure.

Most relevant in the context of our paper are the following two references:
\citeA{gao2017distributionally} put two Wasserstein-type-constraints on the probability distribution $Q$ in \eqref{eq:StochOptFramework}: 
$Q$ has to be close in Wasserstein distance to a reference distribution $\bar{Q}$, while the dependence structure implied by $Q$ has to be close, again in Wasserstein distance, to a specified reference dependence structure.
In their follow-up paper, \citeA{gao2017data} consider problem \eqref{eq:ProblemPHI} in the context of stochastic optimization, i.e.~in the framework \eqref{eq:StochOptFramework}. The contribution of this paper, i.e.~their duality result and the LP-formulation, is already reviewed in the above overview. In addition, the authors provide numerical experiments in portfolio selection and nonparametric density estimation.

\subsubsection*{Neural networks in finance and optimization}
Applications of neural networks have vastly increased in recent years. Most of the popularity arose from successes of neural networks related to data representation tasks, e.g.~related to pattern recognition, image classification, or task-specific artificial intelligence. In contrast to such a utilization, neural networks have also been applied strictly as a tool to solve certain optimization problems. This is the way we use neural networks in this paper, and they have found similar uses in various areas related to finance. Among others, they were applied to solving high dimensional partial differential equations and stochastic differential equations \cite<see e.g.>{beck2018solving,berner2018analysis,weinan2017deep}
as well as backward stochastic differential equation \cite{henry2017deep}, in optimal stopping~\cite{becker2018deep}, optimal hedging with respect to a risk measure \cite{buehler2018deep}, and superhedging \cite{eckstein2018computation}.

For more classical learning tasks where neural networks are applied, ideas from 
optimal transport and distributional robustness are also used. While the 
settings are often quite different in nature to the one in this paper, the 
optimization problems which are eventually implemented are 
nevertheless similar. Most related to the current paper are 
settings in which optimal transport type of constraints are 
solved via a penalization or regularization method. Examples
include generative models for images \cite<see e.g.>{gulrajani2017improved,roth2017stabilizing}, optimal transport and 
calculation of barycenters for images \cite<see e.g.>{seguy2017large}, 
martingale optimal transport \cite<see e.g.>{henry2019martingale},
or 
distributional robustness methods applied to learning tasks \cite<see e.g.>{blanchet2016robust,gao2016distributional}.

\section{Results}
\label{sec:results}

\subsection{Duality}
\label{ssec:duality}
 Let $X = X_1 \times X_2 \times\cdots \times X_d$ be a Polish space,
 and denote by $\mathcal{P}(X)$ the set of all Borel probability measures on $X$. Throughout, we fix a reference probability measure  $\bar{\mu} \in \mathcal{P}(X)$. For $i = 1,...,d$, we denote by $\mu_i:=\mu\circ\text{pr}_i^{-1}$ the $i$-th marginal of $\mu\in\mathcal{P}(X)$, where $\text{pr}_i\colon X\to X_i$ is the projection $\text{pr}_i(x):=x_i$.
 Further, let  $\kappa: X \rightarrow [1,\infty)$ 
 be a growth function of the form $\kappa(x_1,...,x_d) = \sum_{i=1}^d \kappa_i(x_i)$, where each $\kappa_i : X_i \rightarrow [1,\infty)$ is continuous and satisfies $\int_{X_i} \kappa_i \,d\bar{\mu}_i < \infty$. We further assume one of the following: Either $\kappa$ has compact sublevel sets,\footnote{This is for example satisfied if all $\kappa_i$ have compact sublevel sets, since the sublevel sets are by continuity closed and further by positivity of $\kappa_i$ it holds $\{\kappa \leq c\} \subseteq \{\kappa_1 \leq c\} \times ... \times \{\kappa_d \leq c\}$.} or $X_i = \mathbb{R}^{d_i}$ for all $i=1, ..., d$.
 Denote by $C_\kappa(X)$ and $U_\kappa(X)$ the spaces of all
 continuous, respectively upper semicontinuous functions $f:X\to \mathbb{R}$ such that 
$f/\kappa$ is bounded. Recall that $C_b(X)$ denotes the set of all continuous and bounded functions on $X$.

In the following we fix a continuous function $c:X\times X\to[0,\infty)$ such that $c(x,x)=0$ for all $x\in X$. The cost of transportation between $\bar{\mu}$ and $\mu$ in $\mathcal{P}(X)$ with respect to the cost function $c$ is defined as 
\begin{align} \label{eq:DefTransDistance}
d_c(\bar{\mu},\mu) := \inf_{\pi \in \Pi(\bar{\mu},\mu)} \int_{X\times X} c\left(x,y\right) \pi(dx,dy),
\end{align}
where $\Pi(\bar{\mu}_1,\dots,\bar{\mu}_d)$ denotes the
set of all $\mu\in\mathcal{P}(X)$ such that $\mu_i = \bar{\mu}_i$ for all $i=1,\dots,d$. The elements in $\Pi(\bar{\mu}_1,\dots,\bar{\mu}_d)$ are referred to as  couplings of the marginals $\bar{\mu}_1,\dots,\bar{\mu}_d$. Although the computation of the convex conjugate in the following result relies on \citeA{bartl2017computational}, we do not need their growth condition on the cost function $c$.
The main reason we do not require this condition is that continuity from above of the functional \eqref{functional:main} - which corresponds to tightness of the considered set of measures - is already obtained by the imposed marginal constraints.

\begin{theorem}\label{thm:main}
For every 
convex and lower semicontinuous function
$\varphi: [0,\infty] \to [0,\infty]$ such that $\varphi(0) = 0$ and $\varphi(\infty) = \infty$, and all $f \in U_\kappa(X)$, it holds that
\begin{align}
	&\max_{\mu \in \Pi(\bar{\mu}_1,\dots,\bar{\mu}_d)} \Big\{\int_{X} f\, d\mu - \varphi( d_c(\bar{\mu},\mu)) \Big\}  \label{functional:main}\\
	&= \inf_{\lambda \geq 0,\, h_i \in C_{\kappa_i}(X_i)} \Big\{\varphi^\ast(\lambda) + \sum_{i=1}^d \int_{X_i} h_i d\bar{\mu}_i +\int_{X} \sup_{y \in X} \Big[ f(y) - \sum_{i=1}^d h_i(y_i) - \lambda c(x,y) \Big]\, \bar{\mu}( dx) \Big\},  \nonumber
	\end{align}
	where $\varphi^\ast$ denotes the convex conjugate of $\varphi$, i.e.~$\varphi^\ast(\lambda) = \sup_{x \geq 0} \lbrace \lambda x - \varphi(x) \rbrace$.
\end{theorem}

\begin{proof}
		1) Define the optimal transport functional
		$\psi_1 : C_\kappa(X) \to \mathbb{R}$ by
		\[ \psi_1(f) := \inf \Big\{ \sum_{i=1}^d \int_{X_i} h_i \,d\bar{\mu}_i : h_i \in C_{\kappa_i}(X_i) \text{ such that } \oplus_{i=1}^d h_i \geq f\Big\},
		\]
		where $\oplus_{i=1}^d h_i:X\to\mathbb{R}$ is defined as $\oplus_{i=1}^d h_i(x):=\sum_{i=1}^d h_i(x_i)$. 
	    We show that $\psi_1$ is continuous from above on $C_\kappa(X)$, i.e.~for every sequence $(f^n)$ in  $C_\kappa(X)$ such that $f^n \downarrow f\in C_\kappa(X)$ one has $\psi_1(f^n) \downarrow \psi_1(f)$.
	    Fix $\varepsilon>0$ and $h_i\in C_{\kappa_i}(X_i)$ such that $\oplus_{i=1}^d h_i\ge f$ and $\psi_1(f)+\frac{\varepsilon}{3}\ge \sum_{i=1}^d \int_{X_i} h_i\,d\bar{\mu}_i$. Since $f^1\in C_{\kappa}(X)$ and $h_i \in C_{\kappa_i}(X_i)$, there exists a constant $M>0$ such that $|f^1|\le M \kappa$ and $|h_i| \le M \kappa_i$. By assumption,  $\int_{X_i}\kappa_i\,d\bar{\mu}_i<+\infty$ for all $i=1,\dots,d$. We now show that $\psi_1(f^{n_0}) \leq \psi_1(f) + \varepsilon$, which we do separately depending on whether we assume $\kappa$ has compact sublevel sets, or $X_i = \mathbb{R}^{d_i}$.
	    \begin{itemize}
	    	\item Let $\kappa$ have compact sublevel sets. Choose $z > 0$ such that $\sum_{i=1}^d \int_{X_i} 4 M (\kappa_i - \frac{z}{d})^+ \,d\bar{\mu}_i \leq \frac{\varepsilon}{3}$. By Dini's lemma there exists $n_0 \in \mathbb{N}$ such that $f^{n_0} \leq \oplus_{i=1}^d h_i + \frac{\varepsilon}{3}$ on the compact $\{\kappa \leq 2z\}$. Since it further holds $\kappa \eins_{\{\kappa > 2z\}} \leq 2 (\kappa - z)^+ \leq 2 \oplus_{i=1}^d (\kappa_i - \frac{z}{d})^+$, one obtains 
	    	\begin{align*}
	    	f^{n_0} &= \eins_{\{\kappa \leq 2 z\}} f^{n_0} + \eins_{\{\kappa > 2 z\}} f^{n_0} \\
	    	&\leq \eins_{\{\kappa \leq 2 z\}} \oplus_{i=1}^d h_i + \eins_{\{\kappa > 2 z\}} f^{n_0} + \frac{\varepsilon}{3} \\
	    	&= \oplus_{i=1}^d h_i + \eins_{\{\kappa > 2 z\}} (f^{n_0} - \oplus_{i=1}^d h_i) + \frac{\varepsilon}{3} \\
	    	&\leq \oplus_{i=1}^d h_i + \eins_{\{\kappa > 2 z\}} 2 M \kappa + \frac{\varepsilon}{3} \\
	    	&\leq \oplus_{i=1}^d \Big(h_i + 4M\big(\kappa_i - \frac{z}{d}\big)^+\Big) + \frac{\varepsilon}{3}
	    	\end{align*}
	    	and hence $\psi_1(f^{n_0}) \leq \sum_{i=1}^d \int_{X_i} h_i + 4M(\kappa_i - \frac{z}{d})^+ \,d\bar{\mu}_i + \frac{\varepsilon}{3} \leq \psi_1(f) + \varepsilon$.
		    \item Let $X_i = \mathbb{R}^{d_i}$. Choose $z > 0$ such that $\sum_{i=1}^d \int_{X_i} 4 M (\kappa_i - \frac{z}{d})^+ \,d\bar{\mu}_i \leq \frac{\varepsilon}{6}$. Choose $R_i > 0$ such that $\sum_{i=1}^d \bar{\mu}_i(\overline{B(0, R_i)}^c) \cdot 4 M z < \frac{\varepsilon}{6}$, where $B(0, r)$ is the open euclidean ball around $0$ of radius $r$. By Dini's lemma there exists $n_0 \in \mathbb{N}$ such that $f^{n_0} \leq \oplus_{i=1}^d h_i + \frac{\varepsilon}{3}$ on the compact $K := K_1 \times ... \times K_d := \overline{B(0, R_1+2)} \times ... \times \overline{B(0, R_d+2)}$. Since $\eins_{B(0, R_i+1)^c}$ is upper semicontinuous, we can find continuous and bounded functions $g_i$ such that $\eins_{B(0, R_i+1)^c} \leq g_i$ and $\sum_{i=1}^d \int_{X_i} g_i \,d\bar{\mu}_i \cdot 4 M z < \frac{\varepsilon}{6}$ (since $g_i$ approximates $\eins_{B(0, R_i+1)^c}$ and $\eins_{B(0, R_i+1)^c} \leq \eins_{\overline{B(0, R_i)}^c}$). With some of the same steps as in the case where $\kappa$ has compact sublevel sets, one obtains
	    	\begin{align*}
	    	f^{n_0} &= \eins_{K} f^{n_0} + \eins_{K^c} f^{n_0} \\
	    	&\leq \oplus_{i=1}^d h_i + \eins_{K^c} (f^{n_0} - \oplus_{i=1}^d h_i) + \frac{\varepsilon}{3} \\
	    	&\leq \oplus_{i=1}^d h_i + \eins_{K^c} \eins_{\{\kappa > 2 z\}} 2 M \kappa + \eins_{K^c} \eins_{\{\kappa \leq 2 z\}} 2 M \kappa + \frac{\varepsilon}{3} \\
	    	&\leq \oplus_{i=1}^d \Big(h_i + 4M\big(\kappa_i - \frac{z}{d}\big)^+\Big) + \eins_{K^c} 4Mz + \frac{\varepsilon}{3} \\
	    	&\leq \oplus_{i=1}^d \Big(h_i + 4M\big(\kappa_i - \frac{z}{d}\big)^+\Big) + (\oplus_{i=1}^d \eins_{K_i^c}) 4Mz + \frac{\varepsilon}{3} \\
	    	&\leq \oplus_{i=1}^d \Big(h_i + 4M\big(\kappa_i - \frac{z}{d}\big)^+ +4Mz \cdot g_i \Big) + \frac{\varepsilon}{3}
	    	\end{align*}
		    and hence $\psi_1(f^{n_0}) \leq \sum_{i=1}^d \int_{X_i} h_i + 4M(\kappa_i - \frac{z}{d})^+ + 4Mz \cdot g_i \,d\bar{\mu}_i + \frac{\varepsilon}{3} \leq \psi_1(f) + \varepsilon$.
	    \end{itemize}
	    This shows that $\psi_1$ is continuous from above on $C_\kappa(X)$. Moreover, its convex conjugate is given by 
		\begin{align}
		\psi^\ast_{1,C_\kappa}(\mu)&=\sup_{f\in C_\kappa(X)}\Big(\int_{X} f\,d\mu-\inf_{\substack{h_i\in C_{\kappa_i}(X_i)\nonumber\\\oplus_{i=1}^d h_i\ge f }} \sum_{i=1}^d \int_{X_i} h_i\,d\bar{\mu}_i\Big)\nonumber\\
		&=\sup_{h_i\in C_{\kappa_i}(X_i)}\sup_{\substack{f\in C_\kappa(X)\\\oplus_{i=1}^d h_i\ge f }}\Big(\int_{X} f\,d\mu- \sum_{i=1}^d \int_{X_i} h_i\,d\bar{\mu}_i\Big) \nonumber\\
		&=\sup_{h_i\in C_{\kappa_i}(X_i)}\sum_{i=1}^d \Big(\int_{X} h_i \,d\mu- \int_{X_i} h_i\,d\bar{\mu}_i\Big)
		=\begin{cases} 0&\mbox{if }\mu\in\Pi(\bar{\mu}_1,\dots,\bar{\mu}_d)\\+\infty&\mbox{else}\end{cases}\label{conj1}
		\end{align}
		for all $\mu\in\mathcal{P}_\kappa(X)$, where $\mathcal{P}_\kappa(X)$ denotes the set of all $\mu\in\mathcal{P}(X)$ such that $\kappa\in L^1(\mu)$. Notice that 
		$\Pi(\bar{\mu}_1,\dots,\bar{\mu}_d)\subset \mathcal{P}_\kappa(X)$.

	2) Define $\psi_2 : C_\kappa(X) \to \mathbb{R}\cup\{+\infty\}$ by 
	\[ \psi_2(f):= \inf_{\lambda \geq 0} \Big( \varphi^\ast(\lambda) + \int_X \sup_{y \in X} \big[ f(y) - \lambda c(x,y) \big] \,\bar{\mu}( dx) \Big).
	\]
	By definition $\psi_2$ is convex and increasing.
	Further, since $\inf_{\lambda\ge 0}\varphi^\ast(\lambda)=\varphi^\ast(0)=0$ and $f^{\lambda c}(x):=\sup_{y \in X}\{  f(y) - \lambda c(x,y)\} \ge f(x)$ for all $\lambda\ge 0$, it follows that
	\[
	\psi_2(f)\ge\inf_{\lambda\ge 0}\Big(\varphi^\ast(\lambda)+\int_X f\,d\bar{\mu}\Big)>-\infty
	\]
	for all $f\in C_\kappa(X)$, where we use that $f\in L^1(\bar{\mu})$. For the convex conjugates one has
	\begin{align}
	\psi^\ast_{2,C_\kappa}(\mu)&:=\sup_{f\in C_\kappa(X)}\Big(\int_X f\,d\mu-\psi_2(f)\Big)\nonumber\\&=\sup_{f\in U_\kappa(X)}\Big(\int_X f\,d\mu-\psi_2(f)\Big)=:	\psi^\ast_{2,U_\kappa}(\mu)=\varphi( d_c(\bar{\mu},\mu)) \label{conj2}
	\end{align}
	for all $\mu\in\mathcal{P}_\kappa(X)$. Indeed, for every $\mu \in \mathcal{P}_\kappa(X)$ one has
	\[
	\psi_{2,U_\kappa}^\ast(\mu) \geq \psi_{2,C_\kappa}^\ast(\mu) \geq \psi_{2,C_b}^\ast(\mu) = \varphi( d_c(\bar{\mu},\mu)),
	\]
	where the last equality is shown in \citeA[Proof of Thm.~2.4, Step 4]{bartl2017computational}, notably without using the growth condition for $c$ imposed in \citeA{bartl2017computational}. It remains to show that $\psi_{2,U_\kappa}^\ast(\mu) \leq \varphi( d_c(\bar{\mu},\mu))$. Since $\varphi(\infty) = \infty$, the case $ d_c(\bar{\mu},\mu) = \infty$ is obvious. Suppose $ d_c(\bar{\mu},\mu) < +\infty$.
	Note that $\int_X f^{\lambda c}\, d\bar{\mu}$ is well-defined since $f^{\lambda c} \geq f \in L^1(\bar{\mu})$, so that the negative part of the integral is finite. 
	Further, by eliminating redundant choices in supremum and infimum of the convex conjugate, one obtains
	\[
	\psi_{2,U_\kappa}^\ast(\mu) = \sup_{\substack{f\in U_{\kappa}(X)\\\psi_2(f) < \infty}} \Big\{ \int_X f d\mu - \inf_{\substack{\lambda \geq 0,\,\varphi^{\ast}(\lambda) < \infty,\\\int_X f^{\lambda c} d\bar{\mu} < \infty}}\Big(\varphi^{\ast}(\lambda) + \int_X f^{\lambda c} d\bar{\mu}\Big)\Big\}.
	\]
	For every $\varepsilon>0$, $f \in U_\kappa(X)$ and $\lambda \geq 0$ such that $\psi_2(f) < +\infty$, $\varphi^\ast(\lambda) < +\infty, \int_X f^{\lambda c} d\bar{\mu} < +\infty$, it follows that $\int_X f \,d\mu$, $\varphi^{\ast}(\lambda)$ and $\int_X f^{\lambda c} \,d\bar{\mu}$ are real numbers, so that
	\begin{align*}
	&\int_X f \,d\mu - \varphi^\ast(\lambda) - \int_X f^{\lambda c} \,d\bar{\mu} -\varepsilon\\
	 & \leq \int_X f \, d\mu - \lambda  d_c(\bar{\mu},\mu) + \varphi( d_c(\bar{\mu},\mu)) - \int_X f^{\lambda c} d\bar{\mu}- \varepsilon\\
	&\le \int_{X\times X} f(y) \,\pi(dx,dy) - \int_{X\times X} \lambda c(x,y) \,\pi(dx,dy) -\int_{X\times X} f^{\lambda c}(x)\, \pi(dx,dy) + \varphi( d_c(\bar{\mu},\mu))\\
	&\leq  \int_{X\times X} \big[\lambda c(x,y) + f^{\lambda c}(x) - \lambda c(x,y) - f^{\lambda c}(x)\big] \,\pi(dx,dy) + \varphi( d_c(\bar{\mu},\mu)) \\
	&= \varphi( d_c(\bar{\mu},\mu)),
	\end{align*}
	where $\pi \in \Pi(\bar{\mu},\mu)$ is such that $\lambda  d_c(\bar{\mu},\mu)+\varepsilon \ge \int_{X\times X}\lambda  c\, d\,\pi$, and where we used that $\varphi^\ast(\lambda) \geq \lambda  d_c(\bar{\mu},\mu) - \varphi( d_c(\bar{\mu},\mu))$ and $f(y) \leq \lambda c(x,y) + f^{\lambda c}(x)$. Taking the supremum over all such $f$ and $\lambda$ implies $\psi_{2,U_\kappa}^{\ast}(\mu) \leq \varphi( d_c(\bar{\mu},\mu))$.
	
	3) For $f\in U_\kappa(X)$ define the convolution
	\begin{align*}
	&\psi(f) := \inf_{g \in C_\kappa(X)} \left\{ \psi_1(g) + \psi_2(f-g)\right\} \\
	&= \inf_{\lambda \geq 0,\, h_i \in C_b(X_i)} \left\lbrace\varphi^\ast(\lambda) + \sum_{i=1}^d \int_{X_i} h_i d\bar{\mu}_i +\int_{X} \sup_{y \in X} \Big[ f(y) - \sum_{i=1}^d h_i(y_i) - \lambda c(x,y) \Big]\, \bar{\mu}( dx) \right\rbrace.
	\end{align*}
	For the associated convex conjugates it follows from
	\eqref{conj1} and \eqref{conj2} that
	\begin{align*}
	\psi^\ast_{C_\kappa}(\mu)&=\sup_{f\in C_\kappa(X)} \sup_{g\in C_\kappa(X)}\Big(\int_X f\,d\mu-\psi_1(g)-\psi_2(f-g)  \Big)\\
	&=\sup_{g\in C_\kappa(X)} \Big(\int_X g\,d\mu-\psi_1(g)\Big)+
	\sup_{f\in C_\kappa(X)}\Big(\int_X f\,d\mu-\psi_2(f)  \Big)=\psi^\ast_{1,C_\kappa}(\mu)+\psi^\ast_{2,C_\kappa}(\mu)\\
	&=\psi^\ast_{1,C_\kappa}(\mu)+\psi^\ast_{2,U_\kappa}(\mu)
	=\sup_{g\in C_\kappa(X)} \Big(\int_X g\,d\mu-\psi_1(g)\Big)+
	\sup_{f\in U_\kappa(X)}\Big(\int_X f\,d\mu-\psi_2(f)  \Big)\\
	&=\sup_{f\in U_\kappa(X)} \sup_{g\in C_\kappa(X)}\Big(\int_X f\,d\mu-\psi_1(g)-\psi_2(f-g)  \Big)\\&=\psi^\ast_{U_\kappa}(\mu)=\begin{cases}\varphi\big( d_c(\bar{\mu},\mu)\big)&\mbox{if }\mu\in\Pi(\bar{\mu}_1,\dots,\bar{\mu}_d)\\+\infty&\mbox{else}\end{cases}
	\end{align*}
	for all $\mu\in\mathcal{P}_\kappa(X)$.
	
	4) For every $f\in U_\kappa(X)$ one has
	\[
	\psi(f)\ge \int_X f\,d\bar{\mu}- \psi^\ast_{U_\kappa}(\bar{\mu})=
	\int_X f\,d\bar{\mu}>-\infty
	\]
	since $\psi^\ast_{U_\kappa}(\bar{\mu})=\varphi( d_c(\bar{\mu},\bar{\mu}))=\varphi(0)=0$ and $f\in L^1(\mu)$. This shows that $\psi:U_\kappa(X)\to\mathbb{R}$. By definition, $\psi$ is convex and increasing. Moreover, $\psi$ is continuous from above on $C_\kappa(X)$, since for every sequence $(f^n)$ in $C_\kappa(X)$ such that $f^n\downarrow 0$ one has
	\begin{align*}
	\inf_{n\in\mathbb{N}}\psi(f^n)&=\inf_{n\in\mathbb{N}}\inf_{g\in C_\kappa(X)}\big(\psi_1(g)+\psi_2(f^n-g)\big) \\
	&=\inf_{g\in C_\kappa(X)}\inf_{n\in\mathbb{N}}\big(\psi_1(f^n-g)+\psi_2(g)\big) \\
	&=\inf_{g\in C_\kappa(X)}\big(\psi_1(-g)+\psi_2(g)\big) = \psi(0),
	\end{align*}
	where we use that $\psi_1$ is continuous from above on $C_\kappa(X)$ by the first step. Since also $\psi^\ast_{C_\kappa}=\psi^\ast_{U_\kappa}$ on $\mathcal{P}_\kappa(X)$ by the third step, it follows from \cite[Theorem 2.2.]{bartl2017robust} and \cite[Proposition 2.3.]{bartl2017robust} that $\psi$ has the  dual representation
	\begin{align*}
	\psi(f)&=\max_{\mu \in \mathcal{P}_\kappa(X)} \Big\{\int_{X} f\, d\mu - \psi^\ast_{C_\kappa}(\mu)\Big\} =\max_{\mu \in \Pi(\bar{\mu}_1,\dots,\bar{\mu}_d)} \Big\{\int_{X} f\, d\mu - \varphi( d_c(\bar{\mu},\mu)) \Big\}
	\end{align*}
	for all $f\in U_\kappa(X)$.
\end{proof}

\begin{corollary}
\label{coro:DualityWassersteinball}
For every $f\in U_\kappa(X)$ one has
\begin{align}
&\max_{\substack{\mu \in \Pi(\bar{\mu}_1,\dots,\bar{\mu}_d)\\  d_c(\bar{\mu},\mu)\le\rho}} \int_{X} f\, d\mu \label{eq:primal1} \\
&= \inf_{\lambda \geq 0,\, h_i \in C_{\kappa_i}(X_i)} \Big\{ \rho\lambda + \sum_{i=1}^d \int_{X_i} h_i\, d\bar{\mu}_i 
+ \int_{X} \sup_{y \in X} \Big[ f(y) - \sum_{i=1}^d h_i(y_i) - \lambda c(x,y) \Big]\, \bar{\mu}( dx) \Big\} \label{eq:dual1}
\end{align}
for each radius $\rho\ge 0$.
\end{corollary}
\begin{proof}
This follows directly from Theorem \ref{thm:main} for $\varphi$ given by $\varphi(x)=0$ if $x\le\rho$ and $\varphi(x)=+\infty$ if $x>\rho$. In that case the conjugate is given by $\varphi^\ast(\lambda)=\rho\lambda$.
\end{proof}

\begin{remark}
Let us comment on the interpretation of the dual problem \eqref{eq:dual1}:
Roughly speaking, in case $\rho = \infty$, the above result collapses to the duality of multi-marginal optimal transport. On the other hand, if $\rho = 0$, both the primal problem \eqref{eq:primal1} and the dual problem \eqref{eq:dual1} reduce to $\int f \,d\bar{\mu}$. Finally, if one drops the constraint $\mu \in \Pi(\bar{\mu}_1,\dots,\bar{\mu}_d)$ in the primal formulation \eqref{eq:primal1}, the functions $h_1 = h_2 = \dots = 0$.
\end{remark}

From a computational point of view the penalty function $\varphi(x)=x$ is of particular interest since the optimization in  Theorem \ref{thm:main} over the Lagrange multiplier $\lambda$ disappears.

\begin{corollary}
For every $f\in U_\kappa(X)$ one has
\begin{align*}
&\max_{\mu \in \Pi(\bar{\mu}_1,\dots,\bar{\mu}_d)} \Big\{\int_{X} f\, d\mu -  d_c(\bar{\mu},\mu) \Big\} \\
&= \inf_{h_i \in C_{\kappa_i}(X_i)} \Big\{ \sum_{i=1}^d \int_{X_i} h_i d\bar{\mu}_i 
+ \int_{X} \sup_{y \in X} \Big[ f(y) - \sum_{i=1}^d h_i(y_i) - c(x,y) \Big]\, \bar{\mu}( dx) \Big\}.
\end{align*}
\end{corollary}
\begin{proof}
This follows from Theorem \ref{thm:main} for $\varphi(y)=y$. Indeed, as the convex conjugate is given by $\varphi^\ast(\lambda)=0$ for $0\le\lambda\le 1$ and $\varphi^\ast(\lambda)=+\infty$ for $\lambda>1$,
the infimum in Theorem \ref{thm:main} is attained at $\lambda=1$.
\end{proof}

\begin{corollary}[\citeA{gao2017data}]\label{coro:LPreformulation}
Let $f(x) = \max_{1\leq m \leq M} (a^m)^\top x + b^m$ for $x \in \mathbb{R}^d$, $a^m \in \mathbb{R}^d$, and $b^m \in \mathbb{R}$. Let $\bar{\mu} = \frac{1}{n} \sum_{j=1}^n \delta_{x^j}$ for given points $x^1,\dots, x^n$ in $\mathbb{R}^d$.\footnote{Note that $\delta_{x}(A) = 1$ if $x\in A$, and $\delta_{x}(A) = 0$ otherwise.} Let the same points $x^1,\dots, x^n$ define the sets $X_i$, i.e.~$X_i = \lbrace x_i^1,\dots,x_i^n \rbrace$ and $X= X_1\times \dots \times X_d$. Let the cost function $c$ be additively separable, i.e.~$c(x,y) = \sum_{i=1}^d c_i(x_i,y_i)$. 
Then, the dual problem \eqref{eq:dual1} is equivalent to the linear program
\begin{align} \label{eq:dualLP}
&\min_{\lambda,\, h_i(j),\, g(j),\, u_i(j,m)} \, \Big\{\lambda \rho  + \frac{1}{n} \sum_{i=1}^d \sum_{j=1}^n h_i(j) + \frac{1}{n} \sum_{j=1}^n g(j)\Big\} \\
&\text{s.t.: } g(j) \geq b^m + \sum_{i=1}^d u_i(j,m) \hspace{21.5mm} j=1,\dots,n;\, m=1,\dots,M \\
&  \, \,  u_i(j,m) \geq a_i^m x_i^k  - h_i(k) - \lambda c_i(x_i^j,x_i^k) \qquad i=1,\dots,d;\, m=1,\dots,M;\, j,k = 1,\dots,n \label{eq:LPconstraint2} \\
&\hspace{12mm} \lambda \geq 0.
\end{align}
\end{corollary}
The proof can be found in \citeA{gao2017data}. For the convenience of the reader, we also present a direct proof of Corollary \ref{coro:LPreformulation}.

\begin{proof}
Due to the assumptions that $X_i = \lbrace x_i^1,\dots,x_i^n \rbrace$ and $\bar{\mu} = \frac{1}{n} \sum_{j=1}^n \delta_{x^j}$, the term $\int_{X_i} h_i \, d\bar{\mu}_i $ in \eqref{eq:dual1} can be written as $\frac{1}{n} \sum_{j=1}^n h_i(x^j)$ and we shall use that $h_i(x^j) = h_i(x_i^j)$.
Combing these facts with the assumption $c(x,y) = \sum_{i=1}^d c_i(x_i,y_i)$, the dual problem \eqref{eq:dual1} can be reformulated as 
\begin{align*}
&\min_{\lambda \geq 0,\, h_i} \Bigg\{\lambda \rho + \frac{1}{n} \sum_{i=1}^d \sum_{j=1}^n h_i(x^j) \\
& \qquad \qquad \quad + \frac{1}{n} \sum_{j=1}^n \max_{y \in X} \Big\{ \max_{1 \leq m \leq M} \Big( \sum_{i=1}^d a_i^m y_i + b^m \Big) - \sum_{i=1}^d h_i(y) - \lambda c( x^j, y) \Big\}\Bigg\} \\
&= \min_{\lambda \geq 0,\, h_i} \Bigg\{\lambda \rho + \frac{1}{n} \sum_{i=1}^d \sum_{j=1}^n h_i(x^j) \\
& \qquad \qquad \qquad + \frac{1}{n} \sum_{j=1}^n \max_{1 \leq m \leq M} \Big\{ \max_{y \in X}  b^m + \sum_{i=1}^d \Big( a_i^m y_i  - h_i(y_i) - \lambda c_i( x_i^j, y_i) \Big) \Big\}\Bigg\}
\end{align*} 
The assumption $X= X_1\times \dots \times X_d$ implies that for any $y \in X$ we can find indices $k_1,\dots,k_d$ with $1 \leq k_i \leq n$ for $i=1,\dots,d$
such that $y = (x_1^{k_1},\dots,x_d^{k_d})$.
We introduce the auxiliary variables $g(j) \in \mathbb{R}$ for $j=1,\dots,n$ and write the above problem as 
\begin{align*}
&\min_{\lambda\geq 0,\, h_i,\, g(j)} \, \Big\{\lambda \rho  + \frac{1}{n} \sum_{i=1}^d \sum_{j=1}^n h_i(x^j) + \frac{1}{n} \sum_{j=1}^n g(j): \\
& \qquad \quad g(j) \geq   \max_{k_1,\dots, k_d} b^m + \sum_{i=1}^d \Big( a_i^m x_i^{k_i} - h_i(x_i^{k_i}) - \lambda c_i( x_i^j, x_i^{k_i}) \Big), 1\leq j \leq n, 1 \leq m \leq M   \Big\}\\
&= \min_{\lambda\geq 0,\, h_i,\, g(j)} \, \Big\{\lambda \rho  + \frac{1}{n} \sum_{i=1}^d \sum_{j=1}^n h_i(x^j) + \frac{1}{n} \sum_{j=1}^n g(j): \\
&\qquad \quad g(j) \geq  b^m + \sum_{i=1}^d \max_{1\leq k \leq n} \Big( a_i^m x_i^{k_i} - h_i(x_i^{k_i}) - \lambda c_i( x_i^j, x_i^{k_i}) \Big), 1\leq j \leq n, 1 \leq m \leq M \Big\},
\end{align*} 
where we use that
$$\max_{k_1,\dots, k_d} \sum_{i=1}^d a_i^m x_i^{k_i} - h_i(x^{k_i}) - \lambda c_i(x_i^j,x_i^{k_i}) =  \sum_{i=1}^d \max_{1\leq k \leq n} \Big( a_i^m x_i^k - h_i(x_i^k) - \lambda c_i(x_i^j,x_i^k) \Big).$$
Introducing the auxiliary variables $u_i(j,m) \in \mathbb{R}$, where $i=1,\dots,d, j=1,\dots,n$ and $m=1,\dots,M$, in order to remove the remaining max function, together with the notation $h_i(j) := h_i(x^j) \in \mathbb{R}$ yields the assertion. 
\end{proof}

\subsection{Penalization}
\label{subsec:penalization}
The aim of this section is to modify the functional
\eqref{functional:main}, so that it allows for a numerical solution by neural networks.

To focus on the main ideas, we assume that $\kappa$ is bounded, i.e.~we restrict to continuous bounded functions, as well as $\varphi = \infty \text{\ensuremath{1\hspace*{-0.9ex}1}}_{(\rho, \infty)}$ as in the overview in Section \ref{ssec:Overview}.
Hence, in line with Corollary \ref{coro:DualityWassersteinball}
we consider the functional 
\begin{align}
\phi(f)&:= \max_{\substack{\mu \in \Pi(\bar{\mu}_1,\dots,\bar{\mu}_d)\\  d_c(\bar{\mu},\mu)\le\rho}} \int_{X} f\, d\mu  \label{def:phi} \\
&= \inf_{\lambda \geq 0,\, h_i \in C_{\kappa_i}(X_i)} \Big\{ \rho\lambda + \sum_{i=1}^d \int_{X_i} h_i\, d\bar{\mu}_i 
+ \int_{X} \sup_{y \in X} \Big[ f(y) - \sum_{i=1}^d h_i(y_i) - \lambda c(x,y) \Big]\, \bar{\mu}( dx) \Big\}
\nonumber
\end{align}
for all $f\in C_b(X)$ and a fixed radius $\rho> 0$. For simplicity, we assume that the function
$f^{\lambda c}(x)=\sup_{y \in X}\{f(y) - \lambda c(x,y)\}$ is continuous for all $\lambda\ge 0$ and $f\in C_b(X)$.\footnote{By definition, $f^{\lambda c}$ is lower semicontinuous. Moreover, if $c(x,y)=\bar c(x-y)$ for a continuous function $\bar{c}\colon X\to[0,\infty)$ with compact sublevel sets, then $f^{\lambda c}$ is upper semicontinuous and therefore continuous. This for instance holds for
$\bar{c}(x)=\sum_{i=1}^d |x_i|$ or $\bar{c}(x)=\sum_{i=1}^d |x_i|^2$ corresponding to the first and second order Wasserstein distance on $\mathbb{R}^d$.} In that case, the functional
$\phi_1\colon C_b(X^2)\to\mathbb{R}$ defined as
\begin{equation}\label{def:phi1}
\phi_1(f):=\inf_{\substack{\lambda \geq 0,\, h_i \in C_{b}(X_i),\, g\in C_b(X):\\g(x) \geq f(x,y) - \sum_{i=1}^d h_i(y_i) - \lambda c(x,y)}} \Big\{\lambda \rho + \sum_{i=1}^d \int_{X_i} h_i\, d\bar{\mu}_i +\int_{X} g\, d\bar\mu\Big\}
\end{equation}
satisfies $\phi(\tilde{f})=\phi_1(\tilde{f}\circ\text{pr}_2)$ for all $\tilde{f}\in C_b(X)$, i.e.~$\phi_1$ is an extension of $\phi$ from $C_b(X)$ to $C_b(X^2)$. The functional $\phi_1$ can be regularized by penalizing the inequality constraint. To do so, we consider the functional
\begin{align}
\phi_{\theta, \gamma}(f) &:= \inf_{\substack{\lambda \geq 0,\, h_i \in C_b(X_i),\\ g\in C_b(X)}} \Big\{\lambda \rho + \sum_{i=1}^d \int_{X_i} h_i\, d\bar\mu_i + \int_X g\, d\bar\mu \nonumber \\&\qquad\qquad + \int_{X^2} \beta_{\gamma}\big(f(x, y) - g(x) - \sum_{i=1}^d h_i(y_i) - \lambda c(x,y)\big)\, \theta(dx, dy)\Big\} \label{eq:PhiThetaGamma}
\end{align}
for a sampling measure $\theta\in\mathcal{P}(X^2)$,
and a penalty function $\beta_{\gamma}(x) := \frac{1}{\gamma}\beta(\gamma x)$, $\gamma>0$, where
 $\beta : \mathbb{R} \rightarrow [0,\infty)$ is 
convex, nondecreasing, differentiable, and satisfies $\frac{\beta(x)}{x}\rightarrow \infty$ for $x\rightarrow \infty$. 
Let $\beta_\gamma^\ast(y) := \sup_{x\in \mathbb{R}} \{xy - \beta_\gamma(x)\}$ for $y \in \mathbb{R}_+$, and notice that 
$\beta_{\gamma}^\ast(y) = \frac{1}{\gamma} \beta^\ast(y)$.

Notice that the introduced penalization method is in no way specific to the penalized constraint and hence rather general. It includes as a special case the well-studied entropic penalization related to the Sinkhorn algorithm, which is often applied to optimal transport problems. The penalization can also be seen as a regularization since it introduces a slight smoothness bias for the probability measures in the optimization problem. 
On the one hand, this leads to an approximation error, which can be made arbitrarily small theoretically, see Propositions \ref{prop:penal} and \ref{prop2} below.
On the other hand, the resulting smoothness is also seen as a feature which produces good empirical results \cite<see e.g.>{cuturi2013sinkhorn, genevay2017learning}.

The following lemma sets the stage for Proposition~\ref{prop:penal}, in which we provide a duality result for $\phi_{\theta, \gamma}(f)$, study the respective relation of primal and dual optimizers, and outline convergence $\phi_{\theta,\gamma}(f) \rightarrow \phi(f)$ for $\gamma \rightarrow \infty$.

\begin{lemma}\label{lem:phi_theta_gamma}
For every $f\in C_b(X^2)$ one has
\begin{equation}\label{eq:convolution}
\phi_{\theta, \gamma}(f) = \inf_{\tilde{f} \in C_b(X^2)} \big\{\phi_1(\tilde{f}) + \phi_2(f - \tilde{f})\big\},
\end{equation}
where $\phi_2(f) := \int_{X^2} \beta_\gamma(f)\, d\theta$. Moreover, the convex conjugate of $\phi_{\theta,\gamma}$ is given by 
\[
\phi_{\theta, \gamma}^\ast(\pi) = \left\{ \begin{array}{l l}   \int_{X^2} \beta_{\gamma}^\ast\big(\tfrac{d\pi}{d\theta}\big)\, d\theta &
\text{ if } \pi_1 = \bar{\mu},\,\pi_2 \in \Pi(\bar{\mu}_1,...,\bar{\mu}_d) \mbox{ and } \int_{X^2} c \,d\pi \le  \rho \\ \infty & \text{ else}  \end{array} \right.
\]
for all $\pi\in\mathcal{P}(X^2)$
with the convention $\tfrac{d\pi}{d\theta} = +\infty$ if $\pi$ is not absolutely continuous with respect to $\theta$.	
\end{lemma}
\begin{proof}
Observe that for every $f\in C_b(X^2)$ one has
\begin{align*}
&\inf_{\tilde{f} \in C_b(X^2)} \big\{\phi_1(\tilde{f}) + \phi_2(f - \tilde{f})\big\} \\=& \inf_{\substack{\lambda \geq 0,\, h_i \in C_b(X_i),\, g \in C_b(X),\,\tilde {f} \in C_b(X^2):\\ \tilde{f}(x,y) \leq g(x) + \sum_{i=1}^d h_i(y_i) + \lambda c(x,y)}} \Big\{\lambda \rho + \sum_{i=1}^d \int_{X_i} h_i \,d\bar{\mu}_i + \int_{X} g\, d\bar{\mu} + \int_{X^2} \beta_\gamma(f - \tilde{f})\, d\theta\Big\}
\end{align*}
where the right hand side is equal to $\phi_{\theta, \gamma}(f)$. This follows from the dominated convergence theorem applied on the sequence
$\tilde{f}_n(x,y) = \min\{n,~ g(x) + \sum_{i=1}^d h_i(y_i) + \lambda c(x,y)\}$.

As for the calculation of the convex conjugate,
we first show that 
$\phi_{\theta, \gamma}^\ast(\pi) = \infty$ whenever $\pi_1 \neq \bar{\mu}$ or $\pi_2 \not\in\Pi(\bar{\mu}_1, ...,\bar{\mu}_d)$. Indeed, since 
\begin{align*}
\phi_{\theta, \gamma}(f) &\le \inf_{h_i \in C_b(X_i),\,g\in C_b(X)} \Big\{ \sum_{i=1}^d \int_{X_i} h_i\, d\bar\mu_i + \int_X g\, d\bar\mu \nonumber \\&\qquad\qquad + \int_{X^2} \beta_{\gamma}\big(f(x, y) - g(x) - \sum_{i=1}^d h_i(y_i) \big)\, \theta(dx, dy)\Big\} \\
 &\le \inf_{\substack{h_i \in C_b(X_i),\, g\in C_b(X): \\ g(x)+\sum_i h_i(y_i)\ge f(x,y)}} \Big\{ \sum_{i=1}^d \int_{X_i} h_i\, d\bar\mu_i + \int_X g\, d\bar\mu \Big\}+\beta_\gamma(0),
\end{align*}
it follows that $\phi_{\theta,\gamma}$ is bounded above by a multi-marginal transport problem. As the respective convex conjugate is $+\infty$, it follows that $\phi_{\theta, \gamma}^\ast(\pi) = \infty$ for all $\pi\in\mathcal{P}(X^2)$ such that $\pi_1 \neq \bar{\mu}$ or $\pi_2 \not\in\Pi(\bar{\mu}_1, ...,\bar{\mu}_d)$. 
Conversely, if $\pi_1 = \bar{\mu}$ and $\pi_2 \in \Pi(\bar{\mu}_1, ...,\bar{\mu}_d)$ one has
\begin{align*}
\phi_{\theta, \gamma}^\ast(\pi) &= \sup_{f \in C_b(X^2)}\Big\{ \int_{X^2}  f \,d\pi - \phi_{\theta,\gamma}(f) \Big\} \\
&= \sup_{\lambda \geq 0}  \sup_{\tilde{f}\in C_b(X^2)} \Big\{-\lambda \rho + \int_{X^2} \tilde{f}\,d\pi - \int_{X^2} \beta_\gamma(\tilde{f} - \lambda c)\, d\theta\Big\} \\
&= \sup_{\lambda \geq 0} \sup_{\bar{f}\in C_b(X^2)} \Big\{-\lambda \rho + \lambda \int_{X^2} c\, d\pi + \int_{X^2} \bar{f} d\pi - \int_{X^2} \beta_\gamma(\bar{f})\, d\theta\Big\}\\
&= \sup_{\lambda \geq 0} \lambda \Big( \int_{X^2} c\, d\pi - \rho \Big) + \int_{X^2} \beta_\gamma^{\ast}\big(\tfrac{d\pi}{d\theta}\big) \,d\theta.\\
&=\begin{cases}  
\int_{X^2} \beta_\gamma^{\ast}\big(\tfrac{d\pi}{d\theta}\big) d\theta & \mbox{if } \int_{X^2} c\, d\pi \le \rho\\
+\infty & \mbox{else} 
\end{cases}.
\end{align*}
Here, the second equality follows by substituting $\tilde{f}(x,y) = f(x,y) - \sum_{i=1}^d h_i(y_i) - g(x)$ and using the structure of the marginals of $\pi$. The third equality  follows by setting $\bar{f}^n = \tilde{f} + \min\{n,\lambda c\}$  and using the dominated convergence theorem. Finally, the fourth equality follows by a standard selection argument, see e.g.~the proof of \cite[Lemma 3.5]{bartl2017robust}.
\end{proof}

\begin{proposition}
	\label{prop:penal}
		Suppose there exists $\pi \in \mathcal{P}(X^2)$ such that $\phi_{\theta,\gamma}^\ast(\pi) < \infty$. 
		Then it holds:
	\begin{itemize}
	    \item[(a)] For for every $f\in C_b(X^2)$ one has 
 	    \begin{equation}\label{eq:PhiThetaGamma:dual}
 	    \phi_{\theta, \gamma}(f) = \max_{\substack{\pi \in \Pi(\bar{\mu}, \bar{\mu}_1, ..., \bar{\mu}_d):\\ \int c\,d\pi \leq \rho}} \int_{X^2} f d\pi - \int_{X^2} \beta_\gamma^{\ast}\big(\tfrac{d\pi}{d\theta}\big) \,d\theta.\end{equation}
	\item[(b)] Let $f\in C_b(X^2)$. If $g^\star \in C_b(X)$, $h^\star_i \in C_b(X_i)$, $i=1, ..., d$, and $\lambda^\star \geq 0$ are optimizers of \eqref{eq:PhiThetaGamma},
		then the probability measure $\pi^\star$ defined by
		\[
		\frac{d\pi^\star}{d\theta}(x,y) := \beta_{\gamma}^\prime \Big(f(x,y) - g^\star(x) - \sum_{i=1}^d h^\star_i(y_i) - \lambda^\star c(x,y)\Big)
		\]
		is a maximizer of \eqref{eq:PhiThetaGamma:dual}. Hence $\mu^\star := \pi^\star \circ \textrm{pr}_2^{-1}$ is a feasible solution to \eqref{def:phi}.
		\item[(c)] Fix $f \in C_b(X)$ and $\varepsilon > 0$. Suppose that $\mu_\varepsilon \in \mathcal{P}(X)$ is an $\varepsilon$-optimizer of \eqref{def:phi}, and 
		$\pi_\varepsilon \in \Pi(\bar{\mu}, \mu_\varepsilon)$ satisfies
		$\alpha := \int_{X^2} \beta^\ast\big(\tfrac{d\pi_\varepsilon}{d\theta}\big) \,d\theta < \infty$, and $\int_{X^2} c \,d\pi_\varepsilon \leq \rho$.
	    Then one has
		\[
		\phi_{\theta, \gamma}(f\circ \operatorname{pr}_2) - \frac{\beta(0)}{\gamma}\le \phi(f) \leq \phi_{\theta, \gamma}(f\circ\text{pr}_2) + \varepsilon + \frac{\alpha}{\gamma}.
		\]
	\end{itemize}
\end{proposition}	
	\begin{proof}
		(a) To show duality, we check condition (R1) from \cite[Theorem 2.2]{bartl2017robust}, i.e.~we have to show that $\phi_{\theta, \gamma}$ is real-valued and continuous from above.
		That $\phi_{\theta,\gamma}$ is real-valued follows from the assumption that there exists $\pi \in \mathcal{P}(X^2)$ such that $\phi_{\theta,\gamma}^\ast(\pi) < \infty$. Indeed, it holds $\infty > \phi_{\theta, \gamma}^\ast(\pi) \geq \int f d\pi - \phi_{\theta, \gamma}(f)$ and hence $\phi_{\theta, \gamma}(f) > -\infty$ (while $\phi_{\theta, \gamma}(f) < \infty$ is clear) for all $f \in C_b(X^2)$.

		To show continuity from above, let $(f_n)$ be a sequence in $C_b(X^2)$ such that $f_n \downarrow 0$. In view of \eqref{eq:convolution}, one has  
		\[
		\inf_{n\in\mathbb{N}} \phi_{\theta, \gamma}(f_n) = \inf_{\tilde{f} \in C_b(X^2)} \inf_{n\in\mathbb{N}} \big\{\phi_1(\tilde{f}) + \phi_2(f_n - \tilde{f})\big\}=\inf_{\tilde{f} \in C_b(X^2)} \big\{\phi_1(\tilde{f}) + \phi_2( - \tilde{f})\big\}=\phi_{\theta, \gamma}(0),
		\]
		since $\inf_{n\in\mathbb{N}} \phi_2(f_n - \tilde{f})=  \phi_2(- \tilde{f})$ by dominated convergence. By Lemma \ref{lem:phi_theta_gamma}, the claim follows.

		(b) That $\pi^\star$ is a feasible solution in the sense that  $\pi^\star_1=\bar\mu$, $\pi^\star_2\in \Pi(\bar\mu_1,\dots,\bar\mu_d)$, and $\int_{X^2} c\,d\pi^\star=\rho$ whenever $\lambda^\star>0$, follows from the first order conditions. For instance,
		since the derivative of \eqref{eq:PhiThetaGamma} in direction
		$g^\star + t g$ vanishes at $t=0$, it follows  $\int_X g\,d\bar\mu-\int_{X^2} g\circ\text{pr}_1\,d\pi^\star=0$ for all $g\in C_b(X)$, which shows that
		$\pi^\star_1=\bar\mu$. This also implies that $\pi^\star$ is a probability measure. Similarly, $\pi^\star_2 \in \Pi(\bar{\mu}_1, ..., \bar{\mu}_d)$ follows by considering the derivative in direction $h^\star_i + t h_i$, and $\int_{X^2} \lambda^\star c\, d\pi^\star = \lambda^\star\rho$ from the first order condition for $\lambda$. Hence, as $\pi^\star$ is feasible it follows from Lemma \ref{lem:phi_theta_gamma} that
		\begin{align*}
		\phi_{\theta,\gamma}(f)&\ge\int_{X^2} f\,d\pi^\star-\phi^\ast_{\theta,\gamma}(\pi^\star) \\
		&= \int_{X^2} f \beta^\prime_\gamma\big(f-g^\star-\sum_ih^\star_i-\lambda^\star c\big) - \beta^\ast_\gamma\Big(\beta^\prime_\gamma\big(f-g^\star-\sum_ih^\star_i-\lambda^\star c \big)\Big)\,d\theta	\\
		&=\int_{X^2} g^\star+\sum_i h^\star_i +\lambda^\star c\,d\pi^\star
		+\int_{X^2} \beta_\gamma\big(f-\hat g-\sum_ih^\star_i-\lambda^\star c \big)\,d\theta \\
		&= \lambda^\star\rho +\sum_i\int_{X_i} h^\star_i\,d\bar{\mu}_i+\int_Xg^\star\,d\bar{\mu}+\int_{X^2} \beta_\gamma\big(f-\hat g-\sum_ih^\star_i-\lambda^\star c \big)\,d\theta\\
		&=\phi_{\theta,\gamma}(f)
		\end{align*}		
		where we use that
		$\beta^\ast_\gamma\big(\beta^\prime_\gamma(x)\big)=\beta^\prime_\gamma(x)x-\beta_\gamma(x)$ for all $x\in\mathbb{R}$. This shows that $\pi^\star$ is an optimizer.

		(c) By restricting the infimum in \eqref{eq:PhiThetaGamma} to those $\lambda\ge 0$, $h_i\in C_b(X_i)$, $g\in C_b(X)$ such that $g(x)\ge f(y)-\sum_i h_i(y_i)-\lambda c(x,y)$, it follows that
		\begin{align*}
		\phi_{\theta,\gamma}(f\circ\text{pr}_2)&\le 
		\inf_{\substack{\lambda \geq 0,\, h_i \in C_{b}(X_i),\, g\in C_b(X):\\g(x) \geq f(y) - \sum_{i=1}^d h_i(y_i) - \lambda c(x,y)}} \Big\{\lambda \rho + \sum_{i=1}^d \int_{X_i} h_i\, d\bar{\mu}_i +\int_{X} g\, d\bar\mu\Big\} + \beta_\gamma(0)\\
		&=\phi(f)+\frac{\beta(0)}{\gamma},
		\end{align*}
		where the last equality follows from 
		\eqref{def:phi1}. As for the second inequality, since $\mu_\varepsilon \in \mathcal{P}(X)$ is an $\varepsilon$-optimizer of \eqref{def:phi}, and 
		$\pi_\varepsilon \in \Pi(\bar{\mu}, \mu_\varepsilon)$ one has
\begin{align*}
	\phi(f)&\le \int_{X} f\, d\mu_\varepsilon +\varepsilon = \int_{X^2} f\circ \text{pr}_2\, d\pi_\varepsilon-\phi^\ast_{\theta,\gamma}(\pi_\varepsilon)+\phi^\ast_{\theta,\gamma}(\pi_\varepsilon)+\varepsilon 
	\le \phi_{\theta,\gamma}(f\circ\text{pr}_2)+\frac{\alpha}{\gamma}+\varepsilon.
\end{align*}
The proof is complete.		
\end{proof}

The following Proposition shows that the convergence result from Proposition \ref{prop:penal} (c) can be applied whenever the sampling measure is chosen as $\theta = \bar{\mu}\otimes \bar{\mu}_1 \otimes ... \otimes\bar{\mu}_d$, and a minimal growth condition on the cost function $c$ is imposed, see Proposition \ref{prop2} (b)(ii) below. In this case, existence of $\pi \in \mathcal{P}(X^2)$ such that $\phi^{\ast}_{\theta, \gamma}(\pi) < \infty$ holds as well, so that Proposition \ref{prop:penal} applies in full. It is worth pointing out that this result below trivially transfers to all reference measures $\theta$ for which the Radon-Nikodym derivative $\frac{d \bar{\mu}\otimes \bar{\mu}_1 \otimes ... \bar{\mu}_d}{d\theta}$ is bounded. As pointed out by a referee, it is especially desirable to have the values $\alpha_\varepsilon := \int_{X^2} \beta^\ast\big(\tfrac{d\pi_\varepsilon}{d\theta}\big) \,d\theta$ uniformly bounded in $\varepsilon$ (respectively growing in a certain order depending on $\varepsilon$), so that a linear convergence $\phi_{\theta, \gamma}(f) \rightarrow \phi(f)$ for $\gamma \rightarrow \infty$ (respectively a slower order of convergence) is implied. The result below does not achieve this, and we believe this to be a non-trivial task left open for future work.
\begin{proposition}
	\label{prop2}
	\begin{itemize}
		\item[(a)] Let $\mu_i \in \mathcal{P}(X_i)$ for $i=1, ..., 
		d$. Let $\nu \in \Pi(\mu_1, ..., \mu_d)$ and let $\mu := \mu_1 \otimes 
		\mu_2 
		\otimes ... \otimes \mu_d$.
		Then there exist $\nu^{n} \in \Pi(\mu_1, ..., \mu_d)$ for $n\in\mathbb{N}$ such 
		that $\nu^n \stackrel{w}{\rightarrow} \nu$ for $n \rightarrow \infty$, 
		$\nu^{n} \ll \mu$ and there exist constants $0 < C_n < \infty$ such 
		that 
		$\frac{d \nu^n}{d\mu} \leq C_n$ $\mu$-a.s..
		\item[(b)] Let $\eta_i: X_i \rightarrow [0, \infty)$ be Borel measurable with $\int_{X_i} \eta_i \,d\bar{\mu}_i < \infty$. Let $\eta(x) = \sum_{i=1}^d \eta_i(x_i)$. Assume there is a constant $C > 0$ such that for all $x, y\in X$ it holds $c(x, y) \leq C (\eta(x) + \eta(y))$. Let $\theta = \bar{\mu} \otimes \bar{\mu}_1 \otimes ... \otimes \bar{\mu}_d$. Then it holds:
		\begin{itemize}
			\item[(i)] For $\pi^* \in \Pi(\bar{\mu}, \bar{\mu}_1, ..., \bar{\mu}_d)$ with 
			$\int c \,d\pi^* \leq \rho$, there exist $\pi_\varepsilon \in 
			\Pi(\bar{\mu}, \bar{\mu}_1, ..., \bar{\mu}_d)$ for $\varepsilon > 0$ such that $\pi_\varepsilon 
			\stackrel{w}{\rightarrow} \pi^*$ for $\varepsilon \rightarrow 0$, $\pi_\varepsilon \ll \theta$, 
			$\frac{d\pi_\varepsilon}{d\theta}$ is $\theta$-a.s. bounded and $\int c 
			\,d\pi_\varepsilon \leq \rho$.
			\item[(ii)] The condition for Proposition \ref{prop:penal} (c) is satisfied, i.e.~for every $\varepsilon > 0$ there exists $\mu_\varepsilon \in \mathcal{P}(X)$ which is an $\varepsilon$-optimizer of \eqref{def:phi}, and 
			$\pi_\varepsilon \in \Pi(\bar{\mu}, \mu_\varepsilon)$ satisfying
			$\alpha_\varepsilon := \int_{X^2} \beta^\ast\big(\tfrac{d\pi_\varepsilon}{d\theta}\big) \,d\theta < \infty$, and $\int_{X^2} c \,d\pi_\varepsilon \leq \rho$.
		\end{itemize} 
	\end{itemize}
	\begin{proof} \textsl{Proof of (a):}
		We endow each $X_i$ by a compatible metric $m_i$ and without loss of generality we specify the metric on $X = X_1 \times X_2 \times ... \times X_d$ as $m(x, y) = \sum_{i=1}^d m_i(x_i, y_i)$.
		
		\textsl{Step 1) Construction of $\nu^n$:} 
		Let $K^n_i \subseteq X_i$ be compact for $i=1, ..., d$ such that 
		$\mu_i(K_i^n) \rightarrow 1$ for $n\rightarrow \infty$, and $K^n = 
		K_1^n \times ... \times K_d^n$. Notably $\nu(K^n) \rightarrow 1$ 
		follows.
		
		Further, since each $K_i^n$ is compact
		we can choose a Borel partition
		$\mathcal{A}_n$ of $K^n$, i.e.\[\stackrel{.}{\cup}_{A\in \mathcal{A}_n} 
		A = K^n,\] where each $A \in \mathcal{A}_n$ is Borel measurable and satisfies 
		$A=A_1 \times ... \times A_d$ as well as \[\max_{A \in \mathcal{A}_n} \sup_{x, y\in 
			A} m(x, y) \leq \frac{1}{n}.\] A simple way of obtaining such 
		a partition is to first cover each $K^n_i$ by countably many open balls 
		of radius $1/(2dn)$, choosing a finite subcover, and building a 
		partition of $K^n_i$ out of that subcover. For the partition of $K^n$ simply 
		choose all product 
		sets that can be formed from the partitions of the $K^n_i$.
		
		Note that $(K^n)^c$ is the disjoint union of $2^d-1$ many 
		product sets, namely $K_{n, 1}^c \times K_{n, 2} \times ... \times 
		K_{n, d}$, ..., $K_{n, 1}^c \times ... \times K_{n, d}^c$. We denote 
		the union of $\mathcal{A}_n$ with the family of these $2^d-1$ many 
		product sets by $\overline{\mathcal{A}}_n$, which is a partition 
		of 
		$X$. Define
		\[
		\nu^n := \sum_{A \in 
			\overline{\mathcal{A}}_n} \nu(A) \, \cdot 
		\, 
		(\nu_{|A})_1 \otimes ... \otimes (\nu_{|A})_d
		\]
		where implicitly the sum is understood to only include those terms 
		where $\nu(A) > 0$ and $\nu_{|A}$ is then defined as $\nu_{|A}(B) = 
		\nu(A \cap B) / 
		\nu(A)$. We won't make this explicit, but every time we divide by 
		$\nu(A)$ or $\mu_i(A_i)$ we will assume it is one of the relevant terms 
		with $\nu(A) > 0$, 
		where of course $\nu(A) > 0$ implies $\mu_i(A_i) > 0$ for all $i = 1, 
		..., d$ and $A \in \overline{\mathcal{A}}_n$.
		
		\textsl{Step 2) Verifying marginals of $\nu^n$:} We only show that
		$\nu^n_1 = \mu_1$, while the other marginals follow in the same way by symmetry. Let $B_1 \subseteq X_1$ be Borel. It holds
		\begin{align*}
		\nu^n(B_1 \times X_2 \times ... \times X_d) &= \sum_{A\in 
			\overline{\mathcal{A}}_n} \nu(A) \, \cdot \,
		(\nu_{|A})_1(B_1) 
		\\
		&=
		\sum_{A\in\overline{\mathcal{A}}_n} \nu(A) \, \cdot \, \nu_{|A}(B_1 
		\times X_2 
		\times 
		... 
		\times 
		X_d)\\
		&= \sum_{A\in 
			\overline{\mathcal{A}}_n} \nu\left(A \cap (B_1 \times X_2 \times ... \times 
		X_d) \right)\\
		&= \nu(B_1 \times X_2 \times ... \times X_d) = \nu_1(B_1).
		\end{align*}
		
		\textsl{Step 3) Convergence $\nu^n \stackrel{w}{\rightarrow} \nu$ for 
			$n\rightarrow \infty$:} Let $f: X \rightarrow \mathbb{R}$ be bounded 
		and Lipschitz continuous with constant $L$. We have to show $\int f \,
		d\nu^n \rightarrow 
		\int f\, d\nu$ for $n\rightarrow \infty$. Since $\nu(A) = \nu^n(A)$ for all  $A \in 
		\overline{\mathcal{A}}_n$ (and in particular $\nu^n((K^n)^c) = 
		\nu((K^n)^c)$), 
		it holds
		\begin{align*}
		\Big| \int f \,d\nu^n - \int f\,d\nu \Big| &\leq \|f\|_\infty 2 
		\nu((K^n)^c) + 
		\sum_{A \in \mathcal{A}_n} \Big| \int f \eins_A \,d\nu^n - \int f 
		\eins_A 
		\,d\nu \Big|\\
		&\leq \|f\|_\infty 2 \nu((K^n)^c) + \sum_{A \in \mathcal{A}_n} \sup_{x, 
			y 
			\in A} |f(x)-f(y)| \nu(A) \\
		&\leq \|f\|_\infty 2 \nu((K^n)^c) + \sum_{A \in \mathcal{A}_n} \nu(A) L 
		\frac{1}{n}\\
		&\stackrel{n\rightarrow \infty}{\longrightarrow} 0.
		\end{align*}
		
		\textsl{Step 4) Absolute continuity and boundedness of 
			$\frac{d\nu^n}{d\mu}$:}
		Let \[C_n = \max_{\substack{A \in \overline{\mathcal{A}}_n:\\\nu(A) > 
				0}}
		\frac{1}{\nu(A)^{d-1}}.\]
		Given arbitrary Borel sets $B_i \subseteq X_i$ for $i=1, ..., d$, we 
		show 
		that 
		$\nu^n(B_1 \times ... \times B_d) \leq C_n \, \mu(B_1 \times ... \times 
		B_d)$. Once this is shown, $\nu^n(S) \leq C_n \, \mu(S)$ follows for 
		all Borel sets $S \subseteq X$ by the monotone class theorem, which 
		will immediately 
		yield both absolute continuity $\nu^n \ll \mu$ and $\frac{d\nu^n}{d\mu} 
		\leq C_n$.
		
		For $A \in \overline{\mathcal{A}}_n$ it holds 
		\begin{align*}
		(\nu_{|A})_i(B_i) &= 
		\nu(A)^{-1} \nu(A_1 \times ... \times (A_i \cap B_i) \times ... \times 
		A_d) \\
		&\leq \nu(A)^{-1}\,\cdot\, \nu_i(A_i \cap B_i) = 
		\nu(A)^{-1}\,\cdot\,\mu_i(A_i \cap 
		B_i).
		\end{align*}
		It follows
		\begin{align*}
		\nu^n(B_1 \times ... \times B_d) &= \sum_{A\in\overline{\mathcal{A}}_n} 
		\nu(A) \, \cdot \, (\nu_{|A})_1(B_1) \cdot ... \cdot (\nu_{|A})_d(B_d) 
		\\
		&\leq \sum_{A\in\overline{\mathcal{A}}_n} \frac{1}{\nu(A)^{d-1}} 
		\mu_1(B_1 \cap A_1) \cdot ... \cdot \mu_d(B_d \cap A_d)\\
		&\leq C_n \sum_{A\in\overline{\mathcal{A}}_n} \mu((B_1 \times ... 
		\times B_d) \cap A) = C_n \, \mu(B_1 \times ... \times B_d)
		\end{align*}
		where we note that for the last equality to hold, the second to last 
		sum over $A \in \overline{\mathcal{A}}_n$ includes all terms, not just 
		the ones where $\nu(A) > 0$ (this only makes the sum larger).
		The proof of part (a) is complete. 
		
		\textsl{Proof of (b), (i):}
		We first show the following: If $\Pi(\bar{\mu}, \bar{\mu}_1, ..., \bar{\mu}_d) \ni \pi_\varepsilon \stackrel{w}{\rightarrow} \pi \in \Pi(\bar{\mu}, \bar{\mu}_1, ..., \bar{\mu}_d)$ for $\varepsilon \rightarrow 0$, then $\int c \,d\pi_\varepsilon \rightarrow \int c \,d\pi$ for $\varepsilon \rightarrow 0$.
		
		To prove it, note that the growth condition implies that for every $\delta > 0$, we can choose a compact set $K \in X\times X$ such that $\sup_{\varepsilon > 0} \int_{K^c} c \,d\pi_\varepsilon \leq \delta$ and $\int_{K^c} c \,d\pi \leq \delta$. Restricted to $K$, $c$ is bounded from above, say by a constant $M > 0$ (note $c$ is non-negative). Hence for all $\varepsilon > 0$, it holds $|\int c \,d\pi_\varepsilon - \int \min\{c, M\} \,d\pi_\varepsilon| \leq 2 \delta$, and the same for $\pi$ instead of $\pi_\varepsilon$. Since $\min\{c, M\}$ is continuous and bounded, we get $|\int c \,d\pi_\varepsilon - \int c \,d\pi| \leq 4\delta + |\int \min\{c, M\} \,d\pi_\varepsilon - \int \min\{c, M\} \,d\pi| \rightarrow 4 \delta$ for $\varepsilon \rightarrow 0$. Letting $\delta$ go to zero yields the claim.
		
		For the statement of the part (b)(i), consider for $\lambda \in (0, 1)$ the coupling $\pi_\lambda 
		:= \lambda \pi^* + (1-\lambda) \big(\bar{\mu}\otimes (x \mapsto 
		\delta_x)\big)$. Then it holds $\int c \,d\pi_\lambda < \rho$ since 
		$c(x, 
		x) = 0$. Further $\pi_\lambda \stackrel{w}{\rightarrow} \pi^*$ for 
		$\lambda \rightarrow 1$. By approximating every $\pi_\lambda$ by a 
		$\pi_{\lambda, \varepsilon}$ via part (a), $\int c \,d 
		\pi_{\lambda, \varepsilon} \leq \rho$ follows automatically for 
		$\varepsilon$ small enough, which follows by $\int c \,d\pi_{\lambda, \varepsilon} \rightarrow \int c\,d\pi_\lambda$ for $\varepsilon \rightarrow 0$ as shown above. The claim hence follows by a 
		diagonal argument.
		
		\textsl{Proof of (b), (ii):} Any optimizer $\mu^\star \in \Pi(\bar{\mu}_1, ..., \bar{\mu}_d)$ of \eqref{def:phi} and a corresponding coupling $\pi^* \in \Pi(\bar{\mu}, \bar{\mu}_1, ..., \bar{\mu}_d)$ with $\int c \,d\pi^* \leq \rho$ can be approximated via part (b) by $(\pi_\varepsilon)_{\varepsilon > 0}$ which satisfies all required properties. Taking $\mu_\varepsilon$ as the projection of $\pi_\varepsilon$ onto the second component of $X \times X$, i.e. $\mu_\varepsilon = \pi_\varepsilon \circ ((x, y) \mapsto y)^{-1}$, we get that $\mu_\varepsilon \stackrel{w}{\rightarrow} \mu^\star$ and hence $\int f \,d\mu_\varepsilon \rightarrow \int f \,d\mu^\star$ which means after a possible change of indices, $\mu_\varepsilon$ is an $\varepsilon$-optimizer of \eqref{def:phi}.
	\end{proof}
\end{proposition}

\subsection{Approximation with neural networks}
Let us shortly recap. In Section~\ref{ssec:duality}, we show that our original problem 
$$\phi(f):= \max_{\substack{\mu \in \Pi(\bar{\mu}_1,\dots,\bar{\mu}_d)\\  d_c(\bar{\mu},\mu)\le\rho}} \int_{\R^d} f\, d\mu,$$
can be written as 
$$\inf_{\lambda \geq 0,\, h_i \in C_b(\R)} \Big\{ \rho\lambda + \sum_{i=1}^d \int_{\R} h_i\, d\bar{\mu}_i  + \int_{\R^d} \sup_{y \in \R^d} \Big[ f(y) - \sum_{i=1}^d h_i(y_i) - \lambda c(x,y) \Big]\, \bar{\mu}(dx) \Big\}, $$
for all continuous and bounded functions $f \in C_b(X)$. We then proceed in Section~\ref{subsec:penalization} with considering the penalized version of the latter problem 
\begin{align*} 
\phi_{\theta,\gamma}(f) := \inf_{\substack{\lambda \geq 0,\\ h_i \in C_b(\R),\, g \in C_b(\R^d)}} \Big\{ \lambda \rho &+ \sum_{i=1}^d \int_{\R} h_i \,d\bar{\mu}_i + \int_{\R^d} g \,d\bar{\mu} \\[-10pt] 
+& \int_{\R^{2d}} \beta_{\gamma}\big(f(y) - \sum_{i=1}^d h_i(y_i) - \lambda c(x,y) - g(x)\big) \,\theta(dx,dy)\Big\}. \notag
\end{align*}
We provide sufficient conditions for the convergence $\phi_{\theta,\gamma}(f) \rightarrow \phi(f)$ for $\gamma \rightarrow \infty$.
The subsequent and final step is to theoretically justify that neural networks can indeed be used to approximate $\phi_{\theta,\gamma}(f)$ and thereby $\phi(f)$.

To do so, let us introduce the following notation:  We denote by $A_0,\dots,A_l$ affine transformations with $A_0$ mapping form $\mathbb{R}^{d_0}$ to $\mathbb{R}^m$, $A_1,\dots,A_{l-1}$ mapping form $\mathbb{R}^m$ to $\mathbb{R}^m$ and $A_l$ mapping form $\mathbb{R}^m$ to $\mathbb{R}$. We further fix a non-constant, continuous and bounded \emph{activation function} $\varphi:\mathbb{R} \rightarrow \mathbb{R}$. The evaluation of $\varphi$ at a vector $y \in \mathbb{R}^m$ is understood pointwise, i.e.~$\varphi(y) = \left(  \varphi(y_1),\dots,\varphi(y_m) \right)$. Then,
$$\mathfrak{N}(m, d_0) := \left\lbrace g: \mathbb{R}^{d_0} \rightarrow \mathbb{R}: x\mapsto A_l \circ \varphi \circ A_{l-1} \circ \dots \circ \varphi \circ A_0(x) \right\rbrace$$
defines the set of neural network functions mapping to the real numbers $\R$ with a fixed number of layers $l\geq 2$ (at least one hidden layer), input dimension $d_0$ and hidden dimension $m$.
The following is a classical \textsl{universal approximation theorem} for neural networks.
\begin{theorem}[\citeA{hornik1991approximation}]	
	\label{thm:Hornik}
	Let $h \in C_b(\mathbb{R}^N)$. For any finite measure $\nu \in \mathcal{P}(\mathbb{R}^N)$ and $\varepsilon > 0$ there exists $m\in\mathbb{N}$ and $h^m \in \mathfrak{N}_{m, N}$ 
such that $\|h - h^m\|_{L^p(\nu)} \leq \varepsilon$.
\end{theorem}

For Proposition \ref{prop:conv_network} below, let $X_i := \mathbb{R}^{N_i}$ for $i=1, ..., d$ and thus $X = \mathbb{R}^N$ with $N = \sum_{i=1}^d N_i$.
Define the function $F(\lambda, h_1, ..., h_d, g)$ by
	\[
	\phi_{\theta, \gamma}(f) = \inf_{\substack{\lambda \geq 0,\\ h_i \in C_b(\mathbb{R}^{N_i}),\, g \in C_b(\mathbb{R}^N)}} F(\lambda, h_1, ..., h_d, g).
	\]
	We define the neural network approximation of $\phi_{\theta, \gamma}(f)$ by 
	\[
	\phi_{\theta, \gamma}^m(f) = \inf_{\substack{\lambda \geq 0,\\ h_i^m \in \mathfrak{N}(m, N_i),\, g^m \in \mathfrak{N}(m, N)}} F(\lambda, h_1^m, ..., h_d^m, g^m).
	\]
	The following result showcases a simple, yet general setting in which the neural network approximation is asymptotically precise:
\begin{proposition}
	\label{prop:conv_network}
		Fix $f\in C_b(\mathbb{R}^{N})$. Let $p > 1$, $\beta_{\gamma}(x) := \frac{1}{\gamma} (\gamma x)_+^p$ and assume $c \in L_p(\theta)$. Then
		\[
		\phi_{\theta, \gamma}^m(f) \rightarrow \phi_{\theta, \gamma}(f) ~~~ \text{ for } m \rightarrow \infty.
		\]
\end{proposition}
		\begin{proof}
		By the choice of the activation function, all network functions are continuous and bounded, and hence $\phi_{\theta, \gamma}^m(f) \geq \phi_{\theta, \gamma}(f)$.

It therefore suffices to show that for any $\varepsilon>0$, there exists an $m \in \mathbb{N}$ such that $$\phi_{\theta,\gamma}(f) \geq 	\phi_{\theta,\gamma}^m(f) - \varepsilon.$$
Choose any feasible $(\lambda, h_1, ..., h_d, g)$ for $\phi_{\theta, \gamma}(f)$. 
By Theorem 2 we can find a sequence $(\lambda^m, h_1^m, ..., h_d^m, g^m)$ with $h_i^m \in \mathfrak{N}(m, N_i)$ for $i=1, ..., d$ and $g^m \in \mathfrak{N}(m, N)$ such that for $m\rightarrow \infty$ it holds
		\begin{align*}
		\lambda^m &\rightarrow \lambda,  &&\\
		h_i^m &\rightarrow h_i ~ &&\text{ in } L_p(\bar{\mu}_i) \text{ for } i=1, ..., d,\\
		((x, y) \mapsto h_i^m(y_i)) &\rightarrow ((x, y) \mapsto h_i(y_i)) ~ &&\text{ in } L_p(\theta) \text{ for } i=1, ..., d,\\
		g^m &\rightarrow g ~ &&\text{ in } L_p(\bar{\mu}), \\
		((x, y) \mapsto g^m(x)) &\rightarrow ((x, y) \mapsto g(x)) &&\text{ in } L_p(\theta).
		\end{align*}
		As $c \in L_p(\theta)$, it also holds $\lambda^m c \rightarrow \lambda c$ in $L_p(\theta)$ and hence
		\begin{align*}
		((x, y) \mapsto f(x, y) - g^m(x) - \sum_{i=1}^d h_i^m(y_i) - \lambda^m c(x, y)) \\ \rightarrow ((x, y) \mapsto f(x, y) - g(x) - \sum_{i=1}^d h_i(y_i) - \lambda c(x, y))
		\end{align*}
		in $L_p(\theta)$ as $m\rightarrow \infty$. As the mapping $x \mapsto x^+$ is Lipschitz-1, taking only the positive parts lets the above convergence remain valid. As convergence in $L_p(\theta)$ implies convergence of the $p$-th moment, we obtain $F(\lambda^m, h_1^m, ..., h_d^m, g^m) \rightarrow F(\lambda, h_1, ..., h_d, g)$ as $m\rightarrow \infty$. 
		
For a given $\varepsilon>0$, choose a feasible $(\lambda, h_1, ..., h_d, g)$ for $\phi_{\theta, \gamma}(f)$, such that $\phi_{\theta, \gamma}(f) \geq F(\lambda, h_1, ..., h_d, g) - \frac{\varepsilon}{2}$. Due to the above proven convergence, we can find $(\lambda^m, h_1^m, ..., h_d^m, g^m)$ with $h_i^m \in \mathfrak{N}(m, N_i)$ for $i=1, ..., d$ and $g^m \in \mathfrak{N}(m, N)$ such that
$$\phi_{\theta, \gamma}(f) \geq F(\lambda, h_1, ..., h_d, g) - \frac{\varepsilon}{2} \geq \left( F(\lambda^m, h_1^m, ..., h_d^m, g^m) - \frac{\varepsilon}{2} \right) - \frac{\varepsilon}{2} \geq \phi_{\theta, \gamma}(f)^m - \varepsilon.$$
\end{proof}
\begin{remark}
	\label{rem:errors}
While the previous result obtains, for a fixed $f \in C_b(\mathbb{R}^N)$, the convergence $\phi^m_{\theta, \gamma}(f) \rightarrow \phi_{\theta, \gamma}(f)$ for $m\rightarrow \infty$, approximation errors for finite values of $m$ are also of interest. In general, there is hope to achieve this.
In the setting of Proposition \ref{prop:conv_network} with $p\in \mathbb{N}_{\geq 2}$: Assume $\lambda, h_1, ..., h_d, g$ are optimizers of $\phi_{\theta, \gamma}(f)$ which have sufficient moments and $h_1^m, ..., h_d^m, g^m$ are network functions which approximate $h_1, ..., h_d, g$ up to $\varepsilon$ accuracy for the respective $L^p$-norms.
Then it holds 
\[
\left| \phi_{\theta, \gamma}(f) - \phi_{\theta, \gamma}^m(f) \right| \leq C \cdot \varepsilon,
\]
where $C$ is a constant only depending on $f, \lambda, c, h_1, ..., h_d, g$.
	\begin{proof}
		First, define $T(x, y) := f(y) - \sum_{i=1}^d h_i(y_i) - \lambda c(x,y) - g(x)$ and $T^m(x, y) := f(y) - \sum_{i=1}^d h_i^m(y_i) - \lambda c(x,y) - g^m(x)$.
		Using the function $F$ introduced above, we have that $\phi_{\theta, \gamma}^m(f) \leq F(\lambda, h_1^m, ..., h_d^m, g^m)$ and hence
		\begin{align*}
		\phi_{\theta, \gamma}^m(f) - \phi_{\theta,\gamma}(f) &\leq |F(\lambda, h_1^m, ..., h_d^m, g^m) - F(\lambda, h_1^m, ..., h_d^m, g)| \\ 
		&\leq \sum_{i=1}^d \|h_i - h_i^m\|_{L^1(\bar{\mu}_i)} + \|g-g^m\|_{L^1(\bar{\mu})}  + \|T_+^p - (T^m)_+^p\|_{L^1(\theta)}
		\end{align*} 
		and by the inequality quoted in Lemma \ref{lem:Inequality} (see Appendix~\ref{subsec:Inequality}) it holds \[\|T_+^p - (T^m)_+^p\|_{L^1(\theta)} \leq \tilde{C} \|T_+ - T_+^m\|_{L_p(\theta)} \leq \tilde{C} \|T - T^m\|_{L^p(\theta)} \leq \tilde{C} (d+1) \varepsilon\] with $\tilde{C} = \sum_{k=0}^{p-1}\|T_+\|_{L^p(\theta)}^k \|T_+^m\|_{L^p(\theta)}^{p-1-k}$.
		Since $\phi_{\theta,\gamma}(f) \leq \phi_{\theta, \gamma}^m(f)$, we obtain
		$$\left| \phi_{\theta, \gamma}(f) - \phi_{\theta, \gamma}^m(f) \right| = \phi_{\theta, \gamma}^m(f) - \phi_{\theta,\gamma}(f) \leq (\tilde{C} (d+1) + (d+1)) \cdot \varepsilon.$$ While the constant $\tilde{C}$ formally depends on the $p$-th moments of both $T$ and $T^m$, to eliminate the dependence on $T^m$ one can use $\|T^m\|_{L^p(\theta)} \leq \|T\|_{L^p(\theta)} + \|T - T^m\|_{L^p(\theta)} \leq  \|T\|_{L^p(\theta)} + 1$ for $\varepsilon$ small enough.
	\end{proof}
\end{remark}

\section{Implementation}
	\label{sec:implementation}
	This section aims to give specifics regarding the implementation of problem \eqref{eq:PhiThetaGammaOverview} as an approximation of problem \eqref{eq:ProblemPHI}. In particular, the following points are discussed:
	\begin{enumerate}
		\item The choice of $\theta$, $\beta_\gamma$ and neural network structure.
		\item The optimization method for the parameters of the neural network.
		\item How to evaluate the quality of the obtained solution.
		\item The typical runtime.
	\end{enumerate}
	
	\subsection{Choice of $\theta$, $\beta_\gamma$ and neural network parameters}
	\label{subsec:parchoice}
	The neural network structure to approximate the space $C_b(\mathbb{R}^d)$ is chosen as a feedforward neural network with 5 layers (input, output, 3 hidden layers) with hidden dimension $64\cdot d$. The basic idea behind this was to increase the size of the neural networks until a further increase no longer changes the outcome of the optimization. As an activation function we use the ReLu function.
	
	To be precise, the neural network functions we work with are of the form
	\[
	x \mapsto \underbrace{A_4}_{\substack{\text{output}\\\text{layer}}} \circ \underbrace{\varphi \circ 
			A_{3}}_{4\text{th}~\text{layer}} 
		\circ \underbrace{\varphi \circ 
			A_{2}}_{3\text{rd}~\text{layer}} \circ \underbrace{\varphi \circ 
			A_{1}}_{2\text{nd}~\text{layer}} \circ  \underbrace{\varphi \circ A_0}_{\substack{\text{input}\\\text{layer}}}
		(x),
	\]
	where the activation function $\varphi$ is chosen as $\varphi(x) = \max\{0, x\}$. 
	The mappings $A_i$ are affine transformations, i.e.~$A_i(x) = M_i\,x + b_i$ for a matrix $M_i \in \mathbb{R}^{d_{i, 2} \times d_{i, 1}}$ and a vector $b_i \in \mathbb{R}^{d_{i, 2}}$. The matrices $M_0, ..., M_4$ and vectors $b_0, ..., b_4$ are the parameters of the network that one optimizes for. As described above, the dimensions of these parameters are chosen as follows: The input dimension $d_{0, 1} = d$ is given by the dimension of the input vector $x$, and $d_{i, 2} = d_{i+1, 1}$ has to hold for compatibility. The hidden dimension $d_{i, 1}$ for $i = 1, 2, 3, 4$ is set to $64 \cdot d$, while the output dimension $d_{4, 2}$ is always $1$.
	
	The penalization function $\beta_\gamma$ is set to $\beta_\gamma(x) = \gamma \max\{0, x\}^2$. On the one hand, this choice has shown to be stable across all examples. On the other hand, the theory in Proposition~\ref{prop:conv_network}  applies precisely to penalization functions of this kind.  is  Regarding the parameter $\gamma$, we usually first solve the problem with a low choice, like $\gamma=50$, which leads to stable performance. Then, we gradually increase $\gamma$ until a further increment no longer leads to a significant change in the objective value of \eqref{eq:PhiThetaGammaOverview}.
	Regarding instabilities when $\gamma$ is set too large, see Section \ref{subsec:qualityofsol}.
	
	Concerning the sampling measure $\theta$, the basic choice is to use $\theta^{\text{prod}} = \bar{\mu} \otimes \bar{\mu}_1 \otimes ... \otimes \bar{\mu}_d$. Particularly for low values of $\rho$, this is suboptimal: Indeed, for $\rho = 0$ in problem \eqref{eq:PhiThetaGamma:dual}, we know that the optimizer is always of the form $\pi^{\text{diag}} = \bar{\mu} \otimes K$, where $K$ is the stochastic kernel $K: \mathbb{R}^d \rightarrow \mathcal{P}(\mathbb{R}^d)$ given by $K(x) = \delta_x$. Since $\pi^{\text{diag}}$ is singular with respect to $\theta^{\text{prod}}$, using only $\theta^{\text{prod}}$ as sampling measure, one can expect high errors arising from penalization for small values of $\rho$. It hence makes sense to use (among other possibilities) $\theta^{\text{half}} := \frac{1}{2} \theta^{\text{prod}} + \frac{1}{2} \pi^{\text{diag}}$. This is very specific however, and most solutions will not put mass precisely where $\pi^{\text{diag}}$ puts mass. Hence, we add some noise to $\pi^{\text{diag}}$, e.g.~via a Gaussian measure with covariance matrix $\varepsilon^2$: $\theta^{\text{third}} := \frac{1}{2} \theta^{\text{prod}} + \frac{1}{4} \pi^{\text{diag}} + \frac{1}{4} (\pi^{\text{diag}} * \mathcal{N}(0, \varepsilon^2))$, where $*$ denotes convolution of measures. 
	The sampling measure $\theta^{\text{half}}$ is used in all four toy examples in Section~\ref{sec:Examples}, whereas we rely on $\theta^{\text{third}}$ in the final case study in Section~\ref{sec:DNB}.
	
	\subsection{Optimization method for the parameters of the neural network}
	\label{subsec:optmethod}
	This subsection may as well be called ``Training". However, as we do not employ neural networks in a training-testing kind of environment, this might be misleading. 
	
	Regarding this topic, trial and error is especially useful, as the simple goal is to obtain a stable convergence. For the parameters of the neural network, we use the Adam Optimizer with parameters $\beta_1 = 0.99$ and $\beta_2 = 0.995$. For the learning rate, we start with $\alpha=0.0001$ for the first $N_0$ iterations of training, and then decrease it by a factor of $0.98$ each $50$ iterations for a total of $N_{\text{fine}}$ further iterations. We use a batch size (the number of samples generated in each iteration for the measures involved) of around $2^7$ to $2^{16}$, see Section \ref{subsec:qualityofsol} for more details. $N_0$ and $N_{\text{fine}}$ are chosen problem specific: for simple problems in Section \ref{sec:Examples}, $N_0=15000$ and $N_{\text{fine}} = 5000$, while for the DNB case study in Section \ref{sec:DNB} they are chosen as $N_0=60000$ and $N_{\text{fine}} = 30000$. 
	
	The parameter $\lambda$ has to be optimized separately from the parameters of the neural network, since the value of $\lambda$ is clearly more important than any single parameter of the network. To be precise after a fixed number $N_\lambda$ of iterations, $\lambda$ is updated by
	$$ \lambda \mapsto \lambda - \alpha_\lambda  \frac{1}{N_\lambda} \sum_{i \in I} \nabla_\lambda^i,$$
	 where $I$ are the previous $N_\lambda$ many iterations, $\alpha_\lambda$ is the learning rate and $\nabla_\lambda^i$ is the sample derivative of the objective function with respect to $\lambda$ in iteration $i$. Concerning the choice of $\alpha_\lambda$ and $N_\lambda$, we usually first set $\alpha_\lambda$ to around $0.1$ (depending on the problem), and decrease it in the same fashion as $\alpha$, while $N_\lambda$ is set to 200.
Before we update $\lambda$ for the first time, we wait until the network parameters are in a sensible region, which typically takes around $1000$ to $10000$ iterations. 

If another parameter is involved in the optimization (such as $\tau$ in the examples that calculate the AVaR), we employ the same method as for $\lambda$, but we update this parameter even more rarely (once every 1000-2500 iterations), and wait longer at the start to update it the first time (between 5000 and 20000 iterations).
	\subsection{Evaluation of the solution quality}
	\label{subsec:qualityofsol}
	To evaluate the obtained solutions, we found that mostly three aspects have to be considered: 
	\begin{itemize}
		\item[(a)] Is the neural network structure rich enough?
		\item[(b)] How large is the effect of penalization?
		\item[(c)] Has the numerical optimization procedure converged to a (near) minimum?
	\end{itemize}
	Section \ref{subsec:convergenceanalysis} shows how this is put into practice for an exemplary case.
	
	Part (a) is the seemingly simplest, as we found the choice of network structure described in Subsection \ref{subsec:parchoice} to be sufficient for all problems, in the sense that further increasing the network size does not alter the obtained solution. 
	
	Regarding part (b), the most useful observation is the following: As described in Proposition \ref{prop:penal}, the numerical solution via neural networks can be used to obtain an approximate solution $\mu^\star$ of the primal problem. 
	If we evaluate the integral $\int f d\mu^\star$ and compare it to $\phi_{\theta, \gamma}(f)$, the difference is $\phi^\ast_{\theta, \gamma}(\pi^\star)$, which can be seen as the effect of penalization. If $\phi^\ast_{\theta, \gamma}(\pi^\star)$ has a small value, it indicates a small effect of penalization. The second observation is that $\phi_{\theta, \gamma}(f)$ is increasing in $\gamma$, and under the conditions studied in Proposition \ref{prop:penal} and \ref{prop2} converges to $\phi(f)$. Hence, starting with a low value of $\gamma$ and increasing it until no further change is observed is a good strategy. When doing so, values of $\gamma$ that are too large can of course be detrimental regarding part (c), and hence when increasing $\gamma$ a concurrent adaptation of training parameters (like learning rate or batch size) is often necessary.
	
	Regarding part (c), we found that most instabilities could be solved by increasing the batch size. This increase naturally comes with longer run times. Especially if $\gamma$ has to be increased a lot to allow for a small effect of penalization, very large batch sizes were required (e.g. in the DNB case study, we use a batch size of $2^{15}$). To obtain structured criteria for convergence (compared to just evaluating convergence visually), we can again use the dual relation arising from Proposition \ref{prop:penal}. Indeed, we can exploit the fact the numerically obtained $\mu^\star$ (as the second marginal of $\pi^\star$ from Proposition \ref{prop:penal} (b)) is an approximately feasible solution to problem \eqref{eq:ProblemPHI} if the algorithm has converged. Hence, as a necessary criteria for convergence, one can check whether $\mu^\star$ satisfies criteria for feasibility. To this end, one can compare the marginals of $\mu^\star$ to those of $\bar{\mu}$ (we did this mostly by visually evaluating empirical marginals of $\mu^\star$) as well as estimate $d_c(\bar{\mu}, \mu^\star)$.\footnote{Clearly, $d_c(\bar{\mu}, \mu^\star)$ should be be bounded by $\rho$ and in the optimum equal to $\rho$ (if one is not already in the edge case where $\rho$ is so large that it doesn't have an effect).}
	
	\subsection{Runtime}
	\label{subsec:hardware}
	Generally speaking, calculations using neural networks can benefit greatly from parallelization, e.g.~by employing GPUs. For most of our examples, this was not necessary however, and the respective calculations could be performed quickly (i.e.~in between one and five minutes) even with a regular CPU (intel i5-7200U; dual core with 2.5-3.1 GHz each). For the DNB case study however, a single run with stable learning parameters takes around 20 hours on a CPU. By utilizing a single GPU (Nvidia GeForce RTX 2080 Ti) this is reduced to around 30 minutes.  Notably, in the smaller examples there was less speed-up when using GPU compared to CPU, the reason being that the problems were too small to fully utilize parallel capabilities of a GPU.

\section{Examples}

\label{sec:Examples}
The aim of this section is to illustrate how the above introduced concepts can be used to numerically solve given problems. In particular we demonstrate that neural networks are able to (1) achieve a satisfactory empirical performance for all problems considered, (2) naturally determine the structure of the worst-case distribution via Proposition \ref{prop:penal}~(b) and (3) deal with problems, that cannot be reformulated as LP. Concerning the latter point, we consider both a function $f$, which cannot be written as the maximum of affine functions, as well as a cost function $c$, which is not additively separable. Additionally, we make a case for the generality of our duality result: we replace the distance constraint by a distance penalty and fix the distribution of bivariate, rather than univariate, marginals. Furthermore by considering unbounded functions $f$, we shed some light on the necessity of the growth functions $\kappa$ used in Theorem \ref{thm:main}.
To achieve all of these points, we consider three examples with increasing difficulty. 

Concerning the notation in this section, $c$ denotes the cost function 
$$c(x,y) = || x-y ||_1 = \sum_i \vert x_i - y_i \vert.$$ 
This notation implies that 
$$d_c(\bar{\mu},\mu) := \inf_{\pi \in \Pi(\bar{\mu},\mu)} \int_{\R^d \times \R^d} \sum_{i=1}^d \vert x_i - y_i \vert \, \pi(dx,dy),$$ 
is the first order Wasserstein distance with respect to the $L_1$-metric. On the other hand, we consider the first order Wasserstein distance with respect to the Euclidean metric
\begin{align}
d_{c_2}(\bar{\mu},\mu) := \left\lbrace \inf_{\pi \in \Pi(\bar{\mu},\mu)} \int_{\R^d \times \R^d} \left( \sum_{i=1}^d ( x_i - y_i )^2 \right)^{1/2}  \pi(dx,dy) \right\rbrace. \label{eq:WSdq1p2}
\end{align}
Notice that the cost function $c_2(x,y) := ||x-y||_2$ is not additively separable.\footnote{In the literature the Wasserstein distance with respect to the Euclidean metric is usually associated with order two, in which case the underlying cost function \emph{is}  additively separable.}

\subsection{Expected maximum of two comonotone standard Uniforms}
\label{ssec:MaxUniforms}
We start our exemplification with a toy example which is not connected to risk measurement.
Consider the following problem 
\begin{align}
\phi(f_1) :=& \sup_{\substack{\binom{V}{U} \sim \mu \in \Pi(\bar{\mu}_1,\bar{\mu}_2),\\ d_c(\bar{\mu},\mu)\leq \rho}}  \mathbb{E} \left[ \max(U,V) \right] = \sup_{\substack{\mu \in \Pi(\bar{\mu}_1,\bar{\mu}_2),\\ d_c(\bar{\mu},\mu)\leq \rho}}
\int_{[0,1]^2} \max(x_1,x_2) \, \mu(dx), \label{eq:ProblemExampleOne}
\end{align}
where $\bar{\mu}_1 = \bar{\mu}_2 = \mathcal{U}([0,1])$ are (univariate) standard uniformly distributed probability measures and $\bar{\mu}$ is the comonotone copula. In other words, $\bar{\mu}$ is a bivariate probability measure with standard uniformly distributed marginals which are perfectly dependent. In the notation of the Section \ref{sec:results}, we choose the function $f$ as $f_1(x) = \max(x_1,x_2)$  and $X = X_1 \times X_2 = [0,1] \times [0,1]$. Interpreting problem \eqref{eq:ProblemExampleOne}, we aim to compute the expected value of the maximum of two standard Uniforms under ambiguity with respect to the reference dependence structure, which is given by the comonotone coupling.
Problem \eqref{eq:ProblemExampleOne} possesses the following analytic solution
\begin{align*}
\phi(f_1) &= \frac{1+\min(\rho,0.5)}{2}.
\end{align*}
The derivation of this solution can be found in Appendix \ref{ssec:ProofEx1} and is based on the duality result in Corollary \ref{coro:DualityWassersteinball}. 
Hence, problem \eqref{eq:ProblemExampleOne} is well suited to benchmark the solution method based on neural networks. In comparison, we also solve the problem with linear programming. To be precise, we consider the following two methods:
\begin{enumerate}
\item We discretize the reference copula $\bar{\mu}$ (and thereby the marginal distributions $\bar{\mu}_1$ and $\bar{\mu}_2$) and solve the resulting dual problem by means of linear programming (see Corollary \ref{coro:LPreformulation}). There are two distinct ways to discretize $\bar{\mu}$:

\begin{enumerate}[a)]
\item We use Monte Carlo sampling. In the notation of Corollary \ref{coro:LPreformulation}, this means we sample $n$ points $x_1^1,\dots, x_1^n$ in $[0,1]$ from the standard Uniform distribution. Then, we set $x_2^j=x_1^j$ for $j=1,\dots,n$.
\item We set the points $x_1^j = x_2^j = \frac{2j-1}{2n}$ for $j=1,\dots,n$. As the comonotonic copula lives only on the main diagonal of the unit square, this deterministic discretization of $\bar{\mu}$ in some sense minimizes the discretization error. The simple geometrical argument used to find this discretization can be applied only due to the special structure of the reference distribution at hand.
\end{enumerate}
Let us emphasize that method 1.a) can be applied to any reference distribution $\bar{\mu}$. On the other hand, method 1.b) can only be used in this particular example as $\bar{\mu}$ is given by the comonotonic copula.
\item We solve the problem with the neural network approach described in the above Section \ref{sec:implementation}. As discussed, some hyperparameters need to be chosen problem specific. In particular we set: $N_0 = 15000$, $N_{\text{fine}} = 5000$, $\gamma = 1280$, batch size $=2^7$ and $\alpha_\lambda = 0.1$.\footnote{We wait for 2500 iterations until updating $\lambda$ for the first time where $\lambda$ is initialized to $\lambda = 0.75$} 
Concerning the sampling measure $\theta$, for this example we compare
\begin{enumerate}[a)]
\item the basic choice $\theta = \theta^{\text{prod}}$ and
\item the improved choice $\theta = \theta^{\text{half}}$.
\end{enumerate}
To better understand these parameter choices and our neural network approach in general, we provide a detailed convergence analysis for this example in Appendix~\ref{subsec:convergenceanalysis}.
\end{enumerate}

\begin{figure}
\begin{center}
\includegraphics[width=0.49\textwidth]{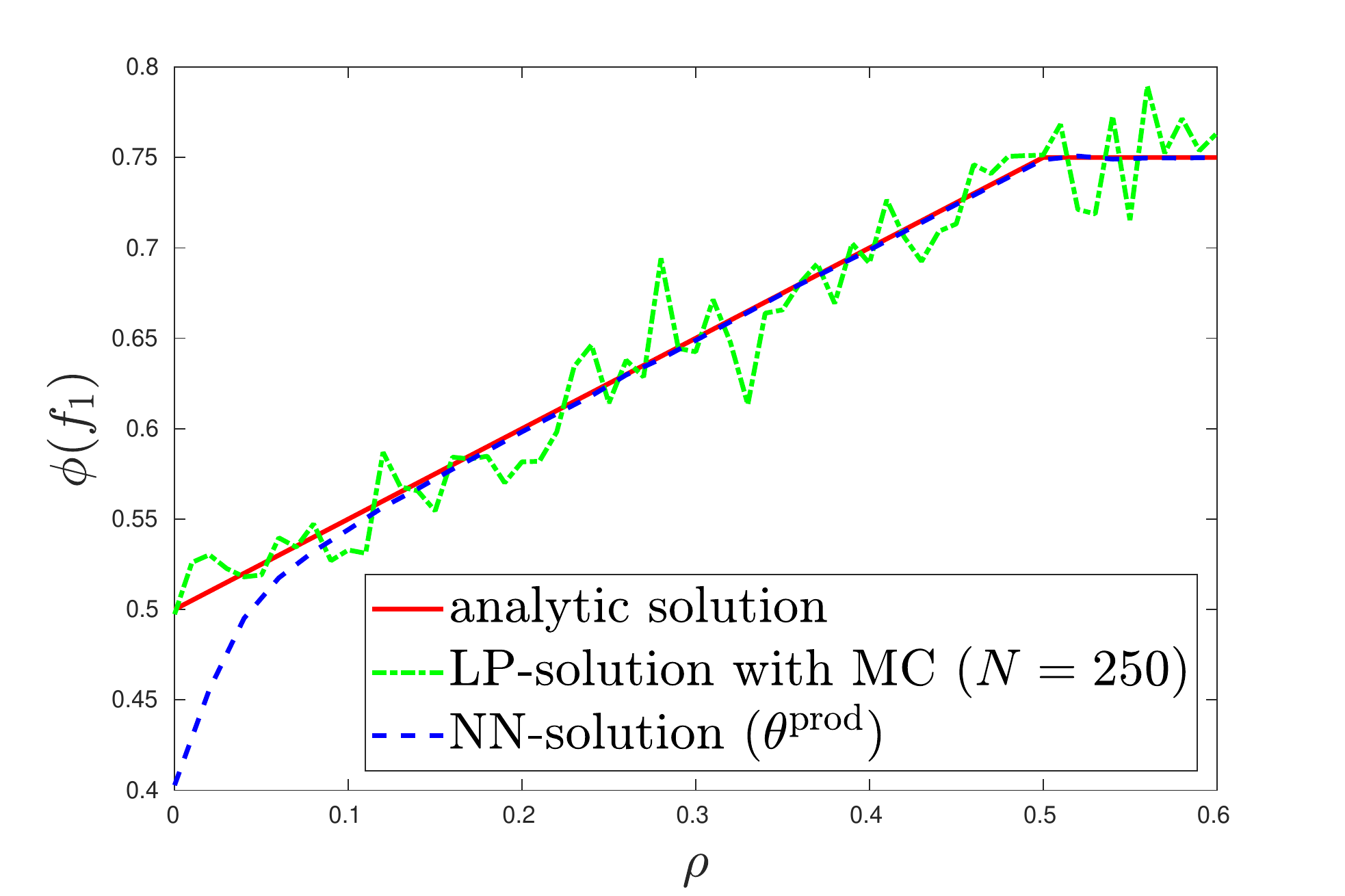}
\includegraphics[width=0.49\textwidth]{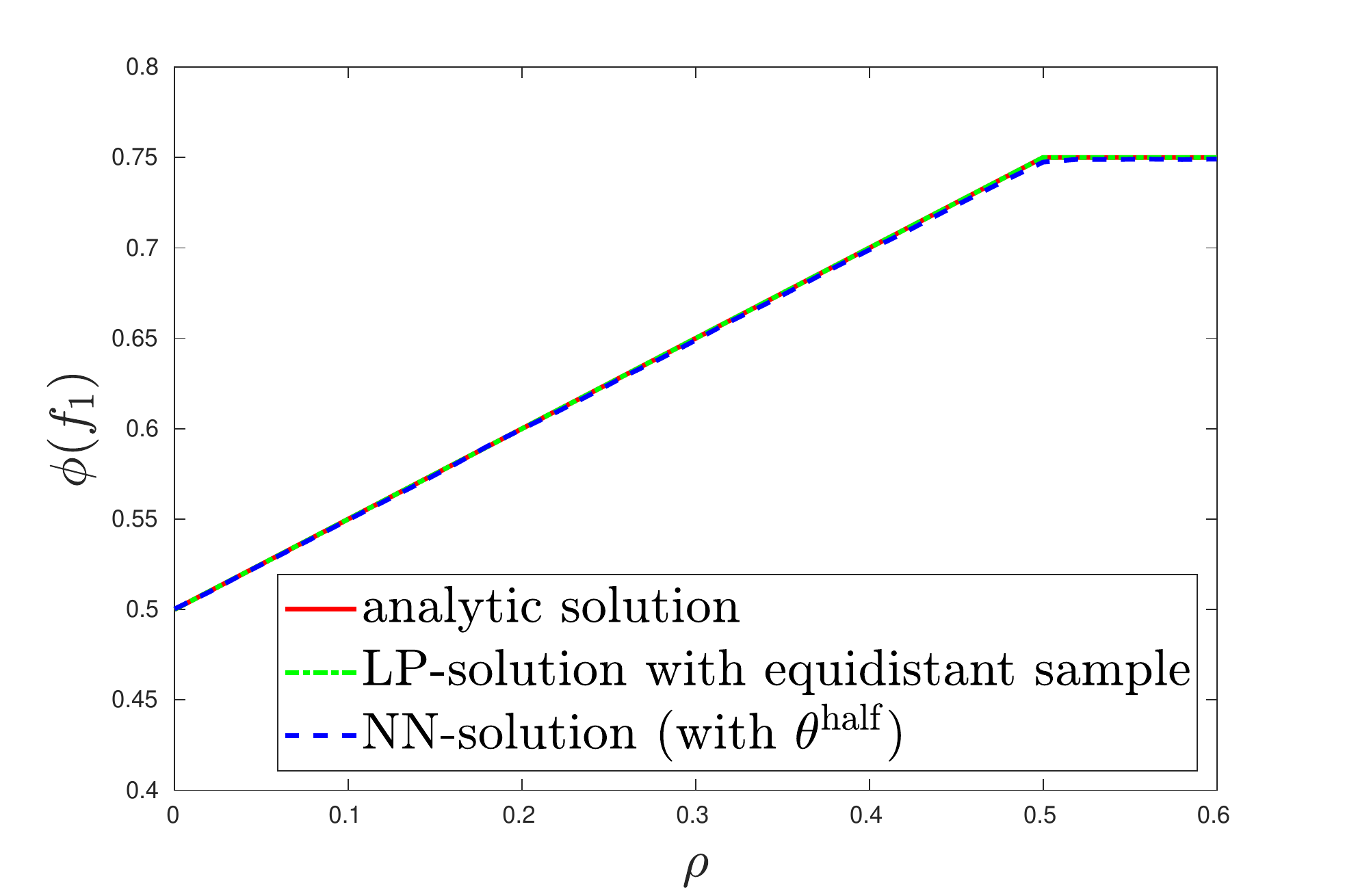}
\caption{In the left panel, the analytic solution $\phi(f_1)$ of problem \eqref{eq:ProblemExampleOne} is plotted as a function of $\rho$ and compared to corresponding numerical solutions obtained by method 1.a) and method 2.a), which are described in Section \ref{ssec:MaxUniforms}. The right panel shows the same for the improved methods 1.b) and 2.b).
}
\label{fig:Example1}
\end{center}
\end{figure}

Figure \ref{fig:Example1} compares the two above mentioned methods to solve problem \eqref{eq:ProblemExampleOne} for different values of $\rho$. In the left panel of Figure \ref{fig:Example1}, we observe that method 1.a) yields an unsatisfactory result even though $n=250$ is chosen as large as possible for the resulting LP to be solvable by a commercial computer. This issue arises due to the poor quality of the discretization resulting from Monte Carlo simulation. If one chooses the discretization as done in method 1.b), we recover the analytic solution of problem \eqref{eq:ProblemExampleOne} as can be seen in the right panel of Figure \ref{fig:Example1}. Moreover, Figure \ref{fig:Example1} indicates that method 2, i.e.~the approach presented in this paper, yields quite good and stable results. The left panel, however, shows that for small $\rho$ method 2.a) does not rediscover the true solution. The reason for this is that when drawing random samples from the chosen sampling measure $\theta^{\text{prod}}$, it is unlikely that we sample from the relevant region, namely the main diagonal of the unit square. As discussed in Section \ref{subsec:parchoice}, method 2.b) is designed to overcome precisely this weakness and the right panel of Figure \ref{fig:Example1} illustrates that it does.

\begin{figure}
\begin{center}
\includegraphics[width=0.49\textwidth]{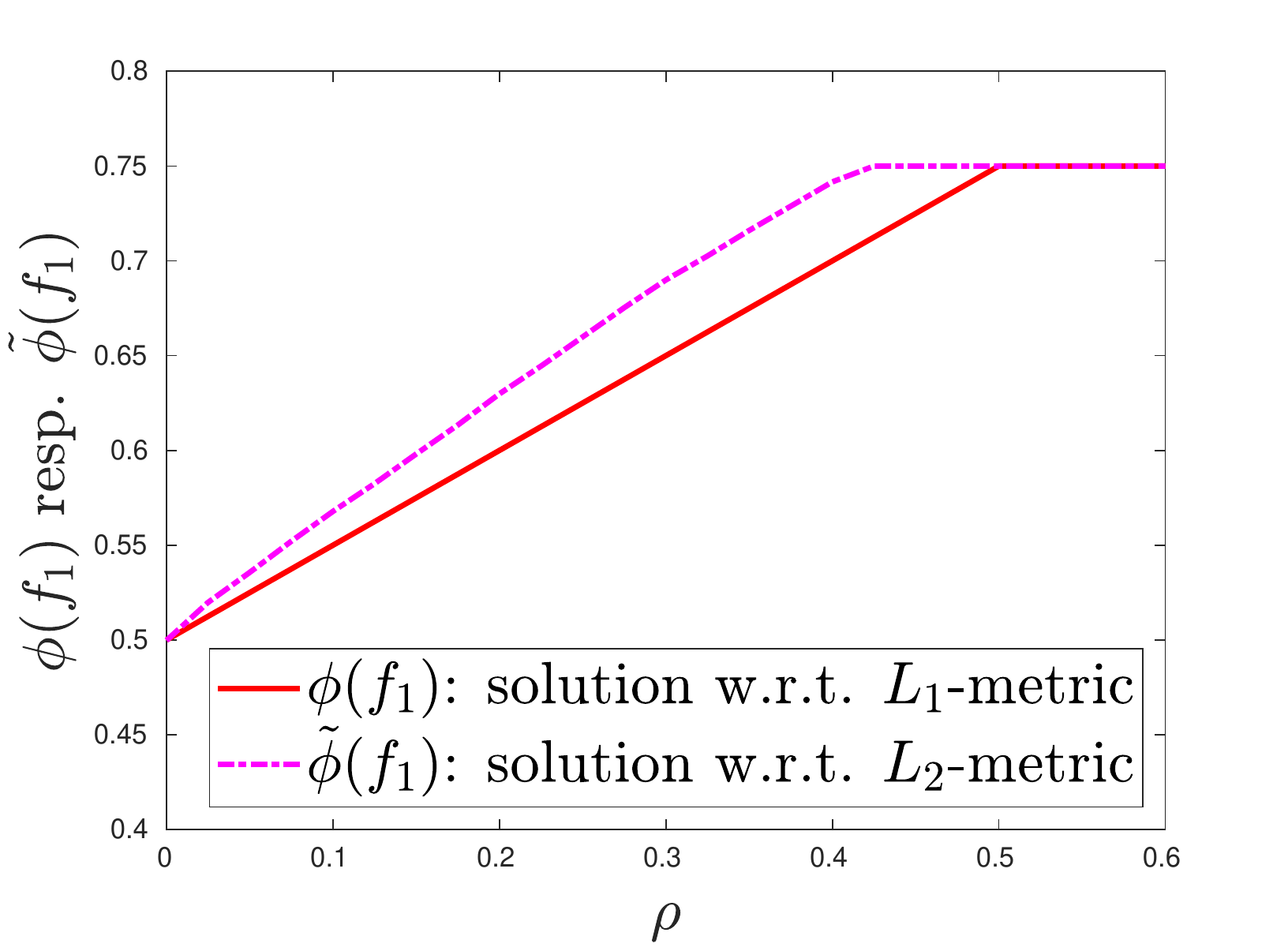}
\caption{The analytic solution $\phi(f_1)$ of problem \eqref{eq:ProblemExampleOne}, which uses the first order Wasserstein distance with respect to the $L_1$-metric, is compared to the numerical solution $\tilde{\phi}(f_1)$ of problem \eqref{eq:ProblemExampleOneOrder2}, which uses the first order Wasserstein distance with respect to the Euclidean metric, i.e.~the $L_2$-metric. }
\label{fig:Example1Order2}
\end{center}
\end{figure}

We finalize this example by considering the Wasserstein distance with respect to the Euclidean metric $d_{c_2}$, defined in equation \eqref{eq:WSdq1p2}, rather than the Wasserstein distance with respect to the $L_1$ metric $d_c$. 
Thus, we compare problem \eqref{eq:ProblemExampleOne} to 
\begin{align}
\tilde{\phi}(f_1) := \sup_{\substack{\mu \in \Pi(\bar{\mu}_1,\bar{\mu}_2),\\ d_{c_2}(\bar{\mu},\mu) \leq \rho}}
\int_{[0,1]^2} \max(x_1,x_2) \, \mu(dx).\label{eq:ProblemExampleOneOrder2}
\end{align}
Since the cost function $c_2$ is not additively separable, $\tilde{\phi}(f_1)$ - other than $\phi(f_1)$ - cannot be approximated based on Corollary \ref{coro:LPreformulation}, i.e.~linear programming. Nevertheless, we can approximate  $\tilde{\phi}(f_1)$ using neural networks, which demonstrates the flexibility of our approach.\footnote{We use the same choices for the hyperparameters as given in point 2. above, but increase $N_0$ considerably for small $\rho$ to guarantee convergence of $\lambda$, and use $\theta^{\text{half}}$ as a sampling measure.} 
Figure \ref{fig:Example1Order2} compares $\phi(f_1)$ and $\tilde{\phi}(f_1)$ for different $\rho$. Note that as $c(x,y) \geq c_2(x,y)$ for all $x,y$, $d_c(\bar{\mu},\mu) \geq d_{c^2}(\bar{\mu},\mu)^{1/2}$ for all $\bar{\mu},\mu \in \mathcal{P}(X)$. Hence, $\phi(f_1) \leq \tilde{\phi}(f_1)$ for fixed $\rho$. Figure \ref{fig:Example1Order2} is in line with this observation.

\subsection{Average Value at Risk of two independent standard Uniforms}
\label{ssec:AVaRUniforms}

We increase the level of complexity slightly compared to the previous example, as we now turn to robust risk aggregation. We aim to compute $\text{AVaR}_\alpha(U+V),$ where $U$ and $V$ are \emph{independent} standard Uniforms under ambiguity with respect to the independence assumption.  Note that the Average Value at Risk is defined by $$\text{AVaR}_\alpha(Y) := \min_{\tau \in \mathbb{R}} \big\lbrace \tau + \frac{1}{1-\alpha} \mathbb{E} \left[ \max(Y-\tau,0) \right] \big\rbrace,$$ see \citeA{rockafellar2000optimization}.
Using the first order Wasserstein distance to construct an ambiguity set around the reference dependence structure, we are led to the following problem
\begin{align}
\Phi_2 :=& \sup_{\substack{\binom{V}{U} \sim \mu \in \Pi(\bar{\mu}_1,\bar{\mu}_2),\\ d_c(\bar{\mu},\mu)\leq \rho}} \text{AVaR}_\alpha(U+V)   \label{eq:ExampleTwo} \\
=& \sup_{\substack{\mu \in \Pi(\bar{\mu}_1,\bar{\mu}_2),\\ d_c(\bar{\mu},\mu)\leq \rho}}
\inf_{\tau \in \mathbb{R}}  \left\lbrace \tau + \frac{1}{1-\alpha} \int_{[0,1]^2} \max(x_1+x_2-\tau,0) \mu(dx) \right\rbrace \label{eq:ExampleTwoLine2}  \\
=& \inf_{\tau \in \mathbb{R}} \phi(f_2^\tau), \label{eq:ExampleTwoLine3}
\end{align}
where $\bar{\mu}_1 = \bar{\mu}_2 = \mathcal{U}([0,1])$ are (univariate) standard uniformly distributed probability measures and
$\bar{\mu}$ is the independence copula. In other words, $\bar{\mu} = \mathcal{U}([0,1]^2)$ is a bivariate probability measure with independent, standard uniformly distributed marginals.
Moreover, we have that $f_2^\tau(x) = \tau + \frac{1}{1-\alpha} \max(x_1+x_2-\tau,0)$ and $\phi(\cdot)$ is defined as in equation \eqref{eq:ProblemPHI}. 

Notice that in the above formulation of the problem we can go from \eqref{eq:ExampleTwoLine2} to \eqref{eq:ExampleTwoLine3} since the problem is convex in $\tau$ and concave in $\mu$ and Wasserstein balls are weakly compact. Thus, we can apply Sion's Minimax Theorem to interchange the supremum and the infimum in \eqref{eq:ExampleTwoLine2}.
\vspace{5mm}

In Appendix \ref{ssec:ProofEx2}, we derive an analytical upper and lower bound for $\Phi_2$ in \eqref{eq:ExampleTwo}. These bounds are tight enough for the present purpose, which is to evaluate the performance of the two discussed numerical methods.

\begin{figure}
\begin{center}
\includegraphics[width=0.49\textwidth]{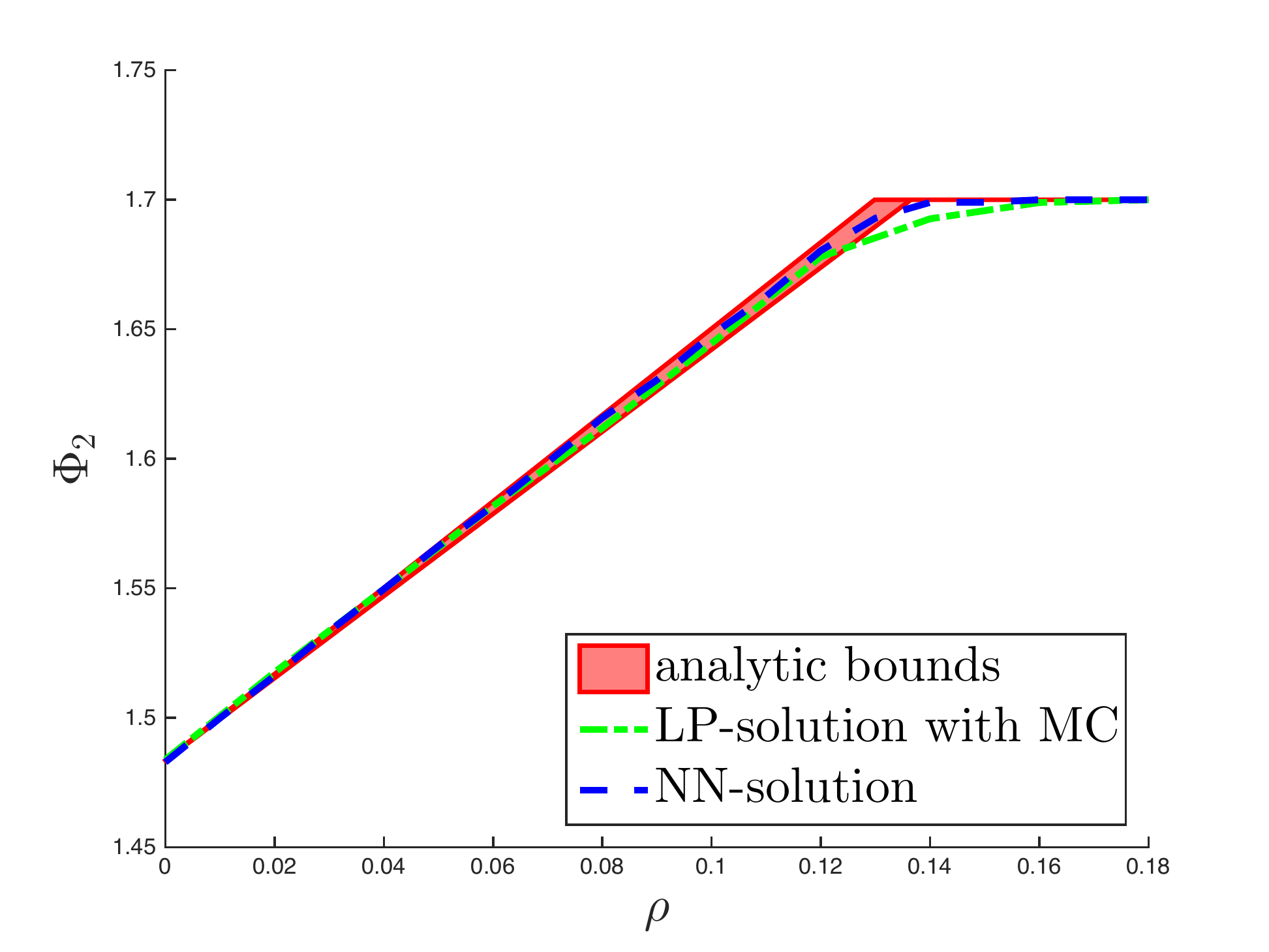}
\caption{The analytic upper and lower bounds of problem \eqref{eq:ExampleTwo} are compared to two distinct numerical solutions. The first numerical solution is obtained by Monte Carlo simulation with $n=100$ sample points as well as linear programming and averaged over 100 simulations for each fixed $\rho$. The second numerical solution is obtained by penalization and neural networks. The confidence level of the AVaR considered in problem \eqref{eq:ExampleTwo} is set to $\alpha = 0.7$.
}
\label{fig:Example2}
\end{center}
\end{figure}

Figure \ref{fig:Example2} supports the latter claim: The analytic bounds for $\Phi_2$ are rather tight when plotted as a function of $\rho$. The bounds are compared to the same two numerical methods as discussed in the previous example. With respect to the solution based on Monte Carlo simulation and linear programming, we now average over 100 simulations for each fixed $\rho$. Thus, the results in Figure \ref{fig:Example2} do not fluctuate as much as those we have seen in the left panel of Figure \ref{fig:Example1}. Nevertheless, Figure \ref{fig:Example2} shows that the solution obtained via MC and LP does not stay within the analytic bounds - other than the solution based on our neural networks approach. Arguably this is due to the lack of symmetry when discretizing the reference distribution $\mu$ using Monte Carlo.  Regarding runtime, both numerical methods take around the same time to calculate the values needed for Figure 3.

\setlength{\unitlength}{0.333\textwidth}
\begin{figure}[ht]
\begin{center}
\begin{picture}(3,2)(0,0)
\put(0,1){\includegraphics[width=0.3\textwidth]{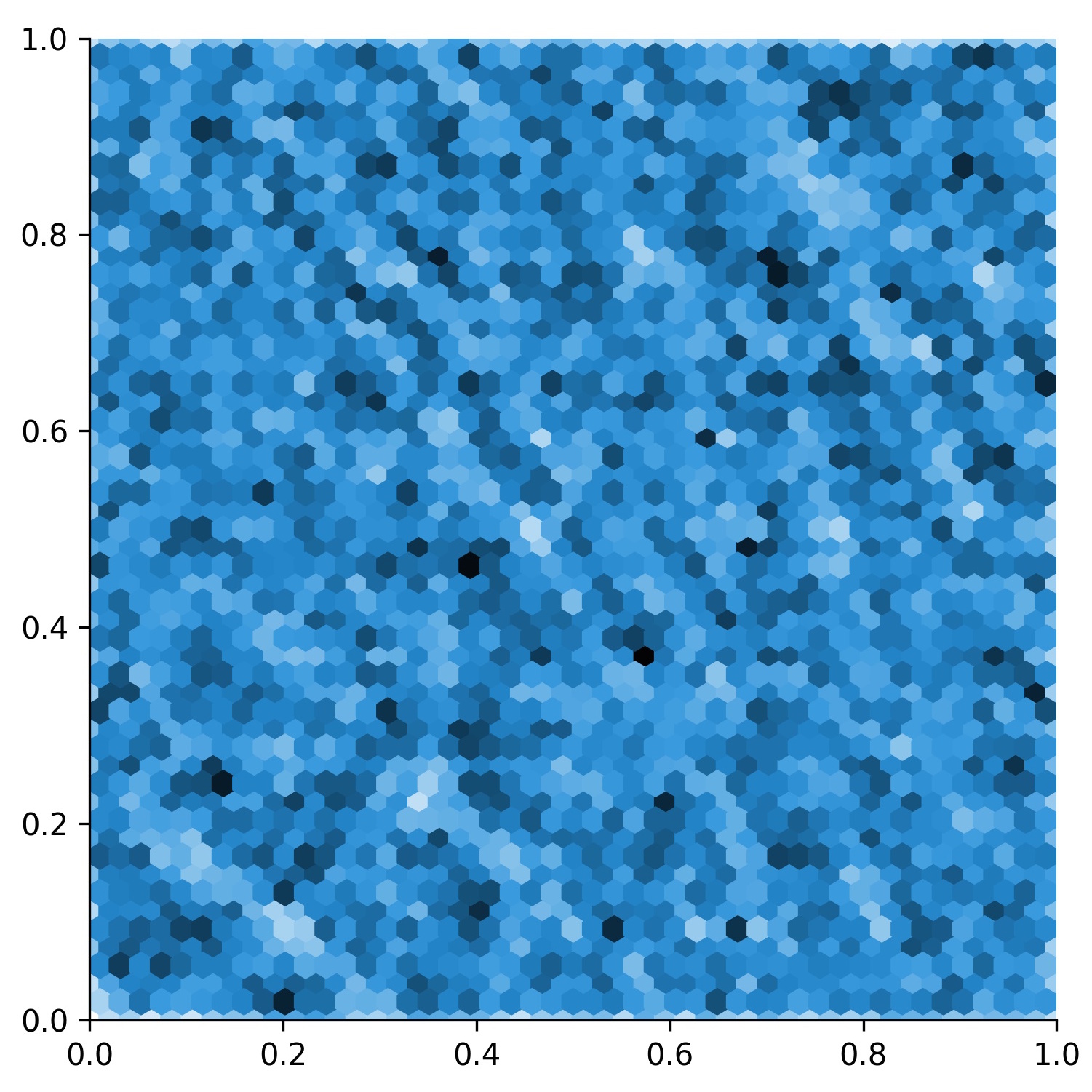}}
\put(1,1){\includegraphics[width=0.3\textwidth]{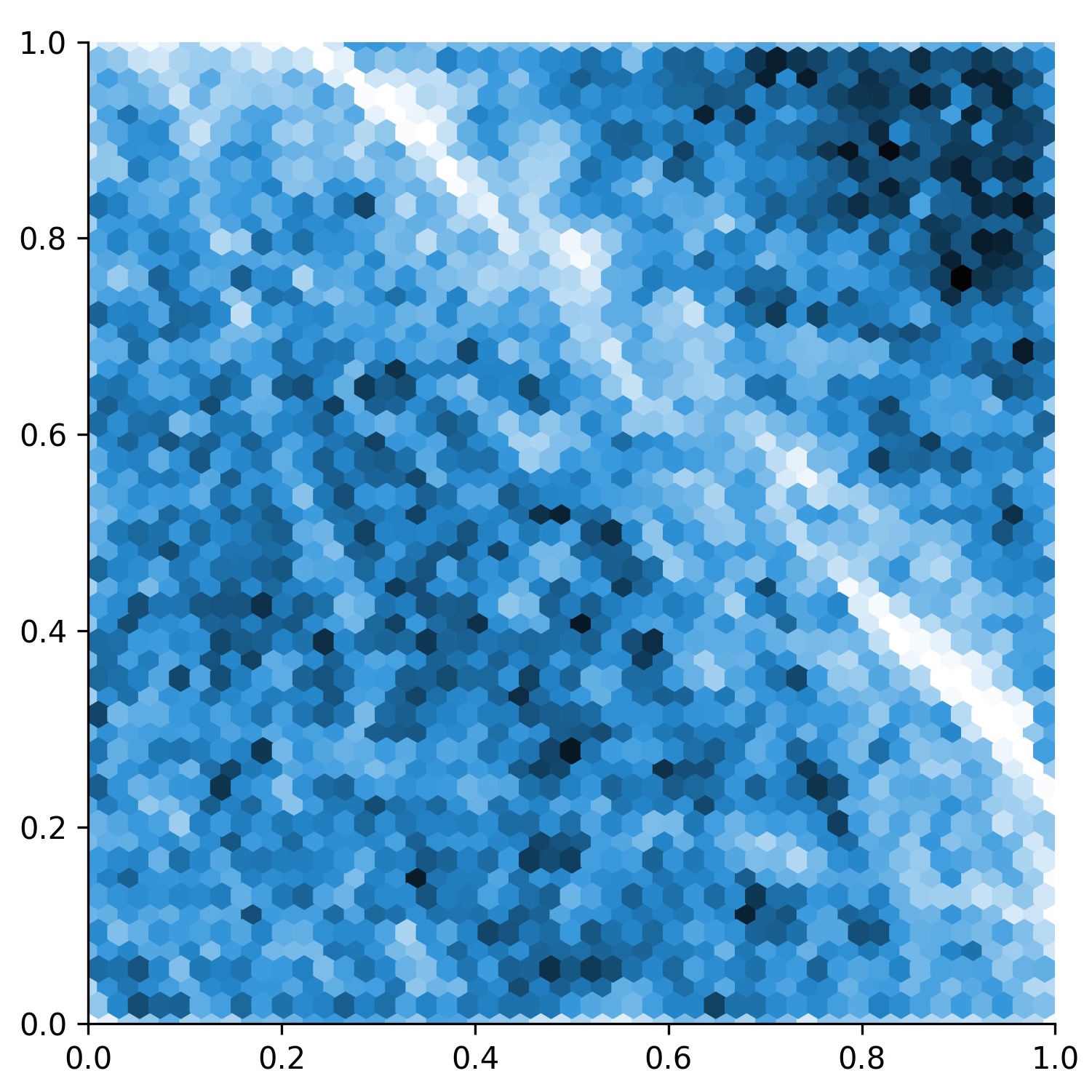}}
\put(2,1){\includegraphics[width=0.3\textwidth]{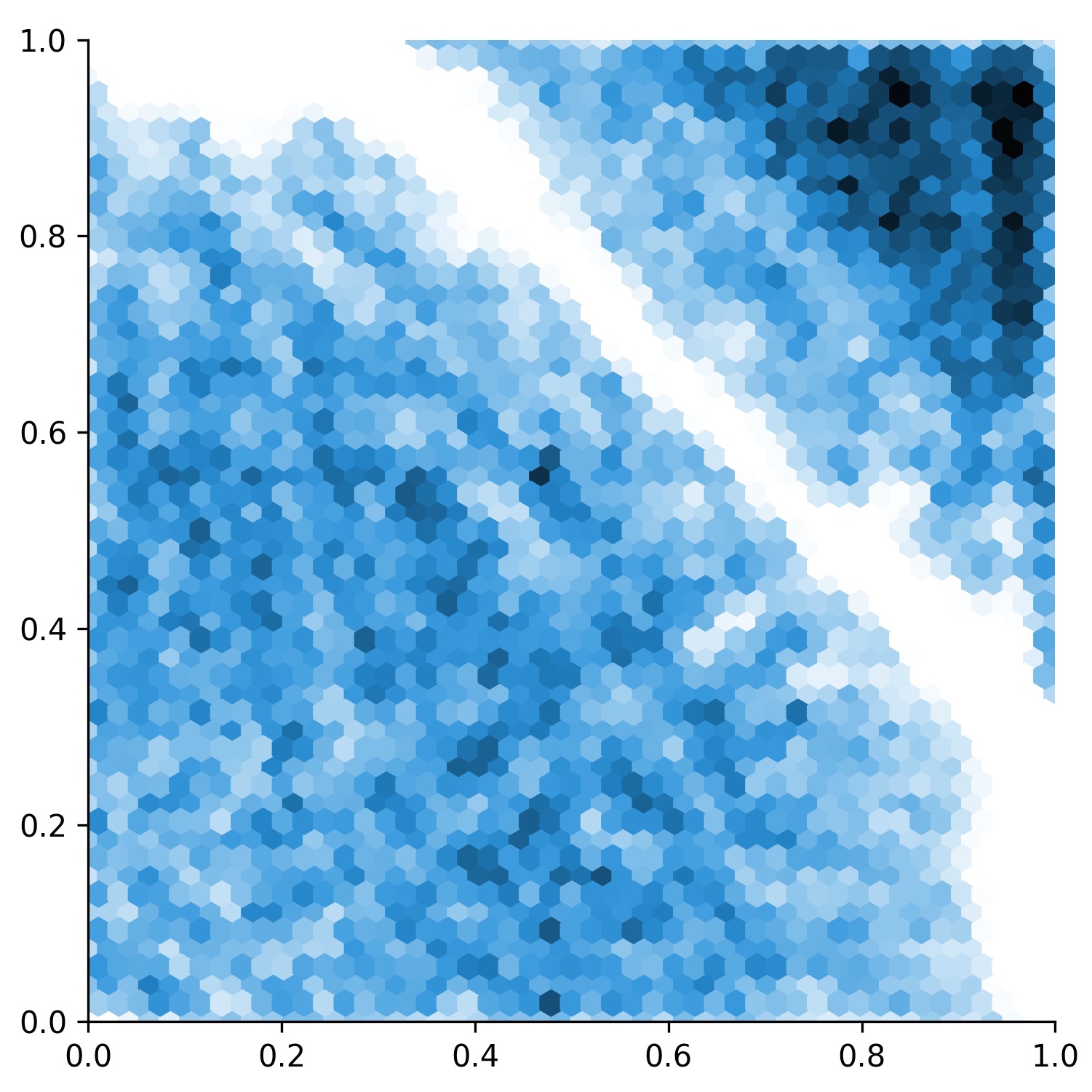}}
\put(0,0){\includegraphics[width=0.3\textwidth]{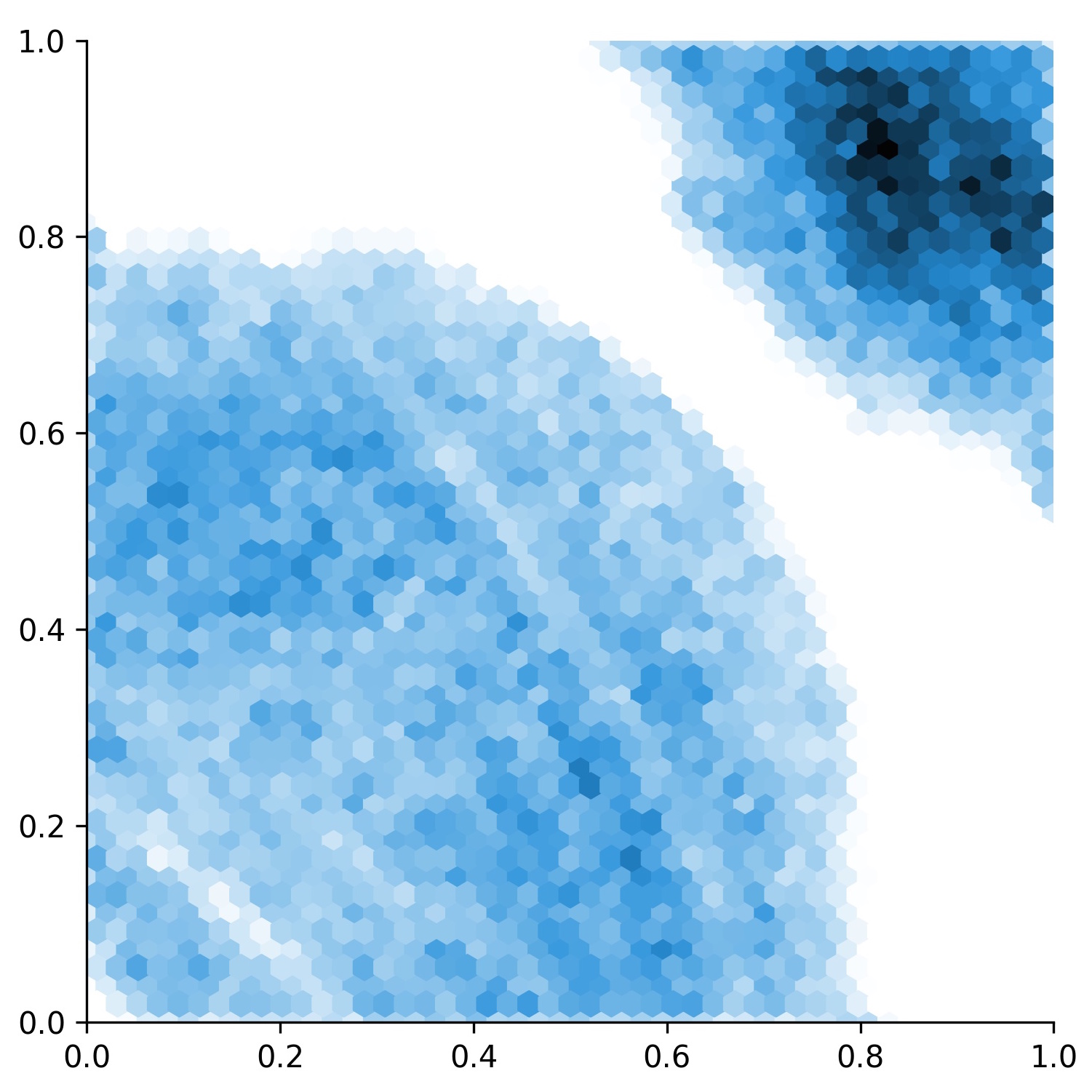}}
\put(1,0){\includegraphics[width=0.3\textwidth]{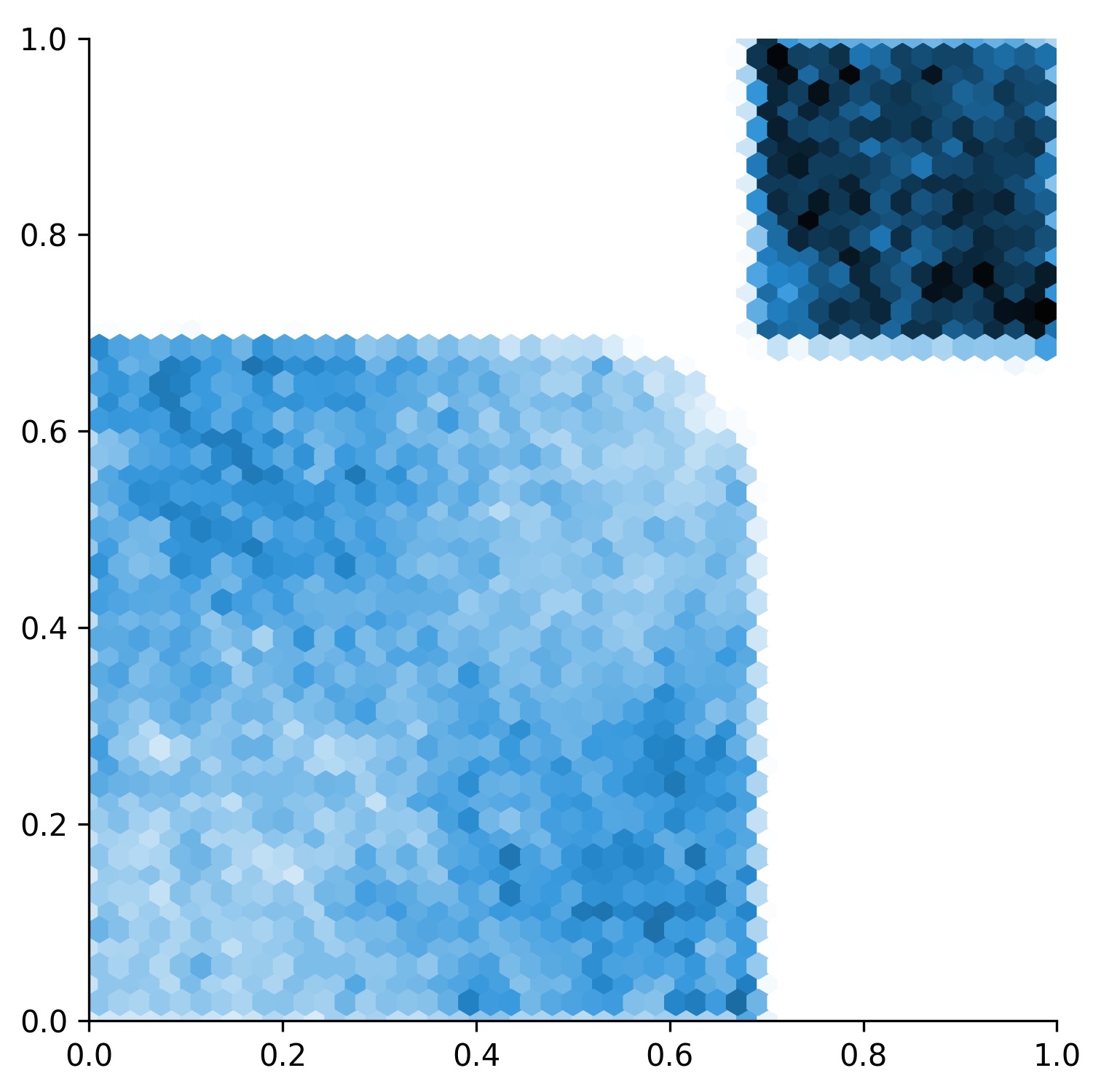}}
\put(2,0){\includegraphics[width=0.3\textwidth]{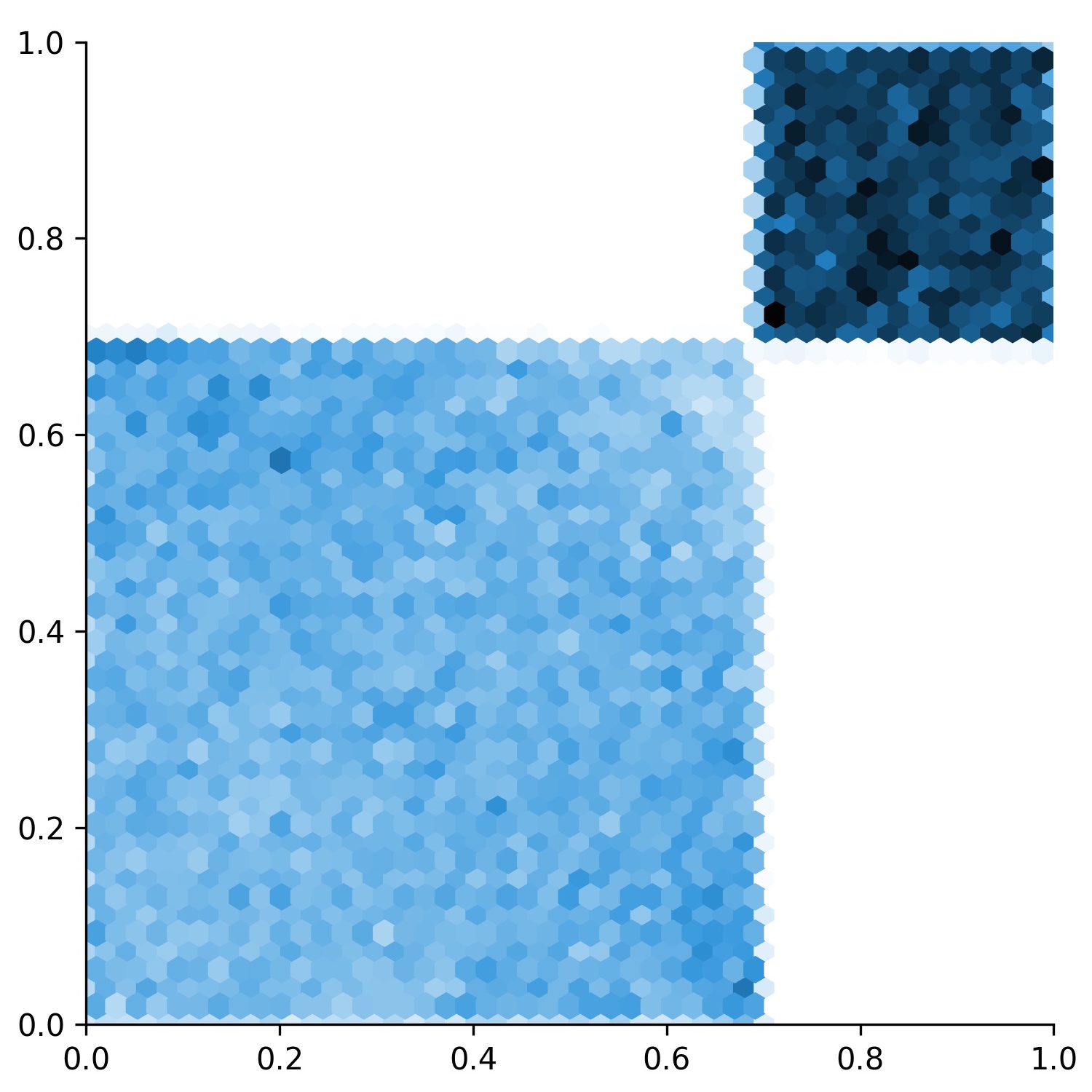}}
\put(0.33,1.9){$\rho = 0.00$}
\put(1.33,1.9){$\rho = 0.04$}
\put(2.33,1.9){$\rho = 0.08$}
\put(0.33,0.9){$\rho = 0.12$}
\put(1.33,0.9){$\rho = 0.16$}
\put(2.33,0.9){$\rho = 0.20$}
\end{picture}
\caption{Samples from the optimizer $\mu^\star$ of problem \eqref{eq:ExampleTwo} as obtained by the neural networks approach are shown in form of a heatplot for six different levels of ambiguity, i.e.~$\rho = 0, 0.04, 0.08, 0.12, 0.16, 0.2$. }
\label{fig:Example2Optimizer}
\end{center}
\end{figure}

We now want to illustrate a further merit of the neural networks approach, namely that we can sample from the numerical optimizer $\mu^\star$ of problem \eqref{eq:ExampleTwo}. By doing so, we obtain information about the structure of the worst case distribution. The samples are obtained by acceptance-rejection sampling from the density given by Proposition \ref{prop:penal} (b), where we replace true optimizers by numerical ones. Figure \ref{fig:Example2Optimizer} plots samples of this worst case distribution $\mu^\star$ for different values of $\rho$. 
To understand the intriguing nature of the results presented in Figure \ref{fig:Example2Optimizer}, we have to describe problem \eqref{eq:ExampleTwo} in some more detail. It should be clear that the comonotone coupling of the Uniforms $U$ and $V$ is maximizing $\text{AVaR}_\alpha(U+V)$ among all possible coupling of $U$ and $V$. However, one can find many different maximizing couplings. Notably, the optimizer shown for $\rho = 0.2$ corresponds to the one which has the lowest relative entropy with respect to the independent coupling among the maximizers of $\text{AVaR}_\alpha(U+V)$. On the other hand, the middle panel for $\rho = 0.16$ motivated us to derive a coupling which - among maximizers of $\text{AVaR}_\alpha(U+V)$ - we conjecture to have the lowest Wasserstein distance to the independent coupling. This coupling is used to derive the lower bound for problem \eqref{eq:ExampleTwo} in Appendix \ref{ssec:ProofEx2}. Some features of the others couplings, e.g.~for $\rho = 0.08$ and $\rho=0.12$ came as a surprise to us: For example, the curved lines as boundary for the support are unusual in an $L_1$-Wasserstein problem.

\subsection{Variance of three normally distributed random variables with distance penalization}
\label{ssec:VarThreeNormals}

We now leave the domain of uniformly distributed, univariate marginals and replace the distance constraint by a distance penalty. 
We analyze the following problem 
\begin{align}
\chi(f_4) :=& \sup_{\binom{X}{Y} \sim \mu \in \Pi(\bar{\mu}_{12},\bar{\mu}_3)} \mathbb{V}\text{ar}( X_1 + X_2 + Y) - \frac{1}{r} d_{\tilde{c}}(\bar{\mu},\mu)^r \notag \\
=& \sup_{\mu \in \Pi(\bar{\mu}_{12},\bar{\mu}_3)} \int_{\mathbb{R}^3} \left( (x_1 + x_2 + y)^2 - m^2 \right) \mu(dx_1,dx_2,dy) -  \frac{1}{r} d_{\tilde{c}}(\bar{\mu},\mu)^r, \label{eq:ProblemVarThreeNormals}
\end{align}
where the cost function ${\tilde{c}}(x,y) = 2 ||x-y||_1$.\footnote{One can think of the factor $2$ occurring in the cost function similarly to a particular choice of $\rho$ for the radius of the Wasserstein ball.} We specify the reference distribution function as follows
$$ \bar{\mu} = \mathcal{N} \left( \left( \begin{array}{c}
0 \\ 0 \\ 0
\end{array} \right), \left( \begin{array}{ccc}
1 & 0.8 & 0\\ 0.8 & 1 & 0 \\ 0 & 0 & 1
\end{array} \right) \right).$$
In this examples, there are two novelties that are explained in the following.

Firstly, the fact that we set $\mu \in \Pi(\bar{\mu}_{12},\bar{\mu}_3)$, means we are fixing not only the univariate marginal distributions, which are standard normal, but also the dependence structure between the first and the second margin $X_1$ and $X_2$. In this case, we assume that $X_1$ and $X_2$ are jointly normal with correlation 0.8. We use $\bar{\mu}_{12}$ to denote the fixed, \emph{bivariate} margin. As a consequence, the model ambiguity concerns solely the dependence structure between the third margin $Y$ and the other two margins $X_1$ and $X_2$.

Secondly, rather than a distance constraint $d_{\tilde{c}}(\bar{\mu},\mu) \leq \rho$, 
we now use a distance penalty to account for the described model ambiguity: we set $\varphi(x) = \frac{1}{r} x^r$ in Theorem \ref{thm:main}. The parameter $r$ accounts for the degree of penalization and hence is \emph{not} comparable to the radius $\rho$ of the Wasserstein balls described above. Instead, for $r \rightarrow \infty$ the penalization becomes closer and closer to the case where we impose the constraint $d_{\tilde{c}}(\bar{\mu}, \mu) \leq 1$.

These two specifications aim to demonstrate the value of the generality of Theorem \ref{thm:main} with respect to both the choice of polish spaces and the modelling of ambiguity.

Even though Section \ref{subsec:penalization} focuses on the Wasserstein ball constraint, the solution method based on penalization and neural networks is trivially adapted to problems like \eqref{eq:ProblemVarThreeNormals}.
We state the resulting numerical solution of problem \eqref{eq:ProblemVarThreeNormals} for different values of $r$ in Table \ref{tab:Example4}. In order to make these results more concrete, we sampled $20000$ values from the respective worst case distribution $\mu^\star$ and report the corresponding empirical covariance matrix $\hat{\Sigma}_{\mu^\star}$. Notably the covariance matrix does not completely characterize $\mu^\star$, since $\mu^\star$ does not have to be a joint normal distribution.

\begin{table}[t]
	\begin{center}
		\begin{tabular}{c||c|c|c|c}
			&$\chi(f_4)$ & $\int_{\mathbb{R}^3} \left( x_1 + x_2 + y \right)^2 d\mu^\star$ & $d_{\tilde{c}}(\bar{\mu},\mu^\star)$ & $\hat{\Sigma}_{\mu^\star}$ \\
			\hline \hline 
			\begin{footnotesize}
				$ \begin{array}{c}
				\text{No} \\ \text{penalization}
				\end{array} $
			\end{footnotesize} & 4.6 & 4.6 & 0 & \begin{footnotesize}
			$\left( \begin{array}{ccc}
			1 & 0.8 & 0\\ 0.8 & 1 & 0 \\ 0 & 0 & 1
			\end{array} \right)$
		\end{footnotesize} \\
		\hline 
		$r=1$ & 6.16 & 8.08 & 1.92 & \begin{footnotesize}
			$\left( \begin{array}{ccc}
			0.998 & 0.801 & 0.847\\ 0.801 & 1.008 & 0.853 \\ 0.847 & 0.853 & 1.011
			\end{array} \right)$
		\end{footnotesize} \\
		\hline 
		$r=2$ & 6.50 & 7.60 & 1.48 & \begin{footnotesize}
			$\left( \begin{array}{ccc}
			0.989 & 0.806 & 0.737\\ 0.806 & 0.989 & 0.736 \\ 0.737 & 0.736 & 0.997
			\end{array} \right)$
		\end{footnotesize} \\
		\hline
		$r=3$ & 6.57 & 7.29 & 1.29 & \begin{footnotesize}
			$\left( \begin{array}{ccc}
			0.975 & 0.795 & 0.675\\ 0.795 & 0.991 & 0.682 \\ 0.675 & 0.682 & 0.980
			\end{array} \right)$
		\end{footnotesize} \\
		\hline 
		$r=4$ & 6.65 & 7.19 & 1.21 & \begin{footnotesize}
			$\left( \begin{array}{ccc}
			0.976 & 0.792 & 0.652\\ 0.792 & 0.970 & 0.654 \\ 0.652 & 0.654 & 0.986
			\end{array} \right)$
		\end{footnotesize} \\
		\hline 
		$r=\infty$ & 6.76 & 6.76 & 1.00 & \begin{footnotesize}
			$\left( \begin{array}{ccc}
			0.991 & 0.803 & 0.554\\ 0.803 & 0.998 & 0.551 \\ 0.554 & 0.551 & 0.993
			\end{array} \right)$
		\end{footnotesize}
	\end{tabular}
	\caption{Comparison of the numerical solutions $\chi(f_4)$ of problem \eqref{eq:ProblemExampleOne}, computed based on penalization and neural networks, for different values of $r$. We define the worst case distribution $\mu^\star \in \Pi(\bar{\mu}_{12},\bar{\mu}_3)$ such that $\chi(f_4) = \int_{\mathbb{R}^3} \left( x_1 + x_2 + y \right)^2 \mu^\star(dx_1,dx_2,dy) -  \frac{1}{r} d_{\tilde{c}}(\bar{\mu},\mu^\star)^r$ and report also the empirical covariance matrix $\hat{\Sigma}_{\mu^\star}$ computed from $N=20000$ samples of $\mu^\star$. The case $r=\infty$ corresponds to the constraint $d_{\tilde{c}}(\bar{\mu}, \mu) \leq 1$.}
	\label{tab:Example4}
\end{center}
\end{table}

\section{DNB case study: Aggregation of six given risks}
\label{sec:DNB}
\citeA{aas2014bounds} provide a very illustrative case study of the risk aggregation at the DNB, Norway's largest bank. We want to make use of this example to showcase the applicability of the novel framework presented in this paper. 

The DNB is exposed to six different types of risks: credit, market, asset, operational, business and insurance risk. Let the random variables $L_1,\dots,L_6$ represent the marginal risk exposures for these six risks. 
Per definition, risk aggregation is not concerned with the computation of the distribution of the marginal risks. Hence, we take the corresponding marginal distribution functions $F_1,\dots,F_6$ as given.
In this particular case, $F_1, F_2$ and $F_3$ are empirical cdfs originating from given samples, while $L_4$, $L_5$ and $L_6$ are assumed to be log-normally distributed with given parameters, see Table \ref{tab:givenInfo}.

\begin{table}
\begin{center}
\begin{tabular}{c|ccc}
&  Description & Type & Parameters/Other details \\
\hline
\hline
\multirow{2}{*}{$F_1$} & \multirow{2}{*}{cdf of credit risk $L_1$} & \multirow{2}{*}{empirical cdf} & given by 2.5 Million samples; \\
 & & & standard deviation $\bar{\sigma}_1 =  644.602$ \\
\hline
\multirow{2}{*}{$F_2$} & \multirow{2}{*}{cdf of market risk $L_2$} & \multirow{2}{*}{empirical cdf} & given by 2.5 Million samples; \\
 & & & standard deviation $\bar{\sigma}_2 = 5562.362$ \\
 \hline
\multirow{2}{*}{$F_3$} & \multirow{2}{*}{cdf of asset risk $L_3$} & \multirow{2}{*}{empirical cdf} & given by 2.5 Million samples; \\
 & & & standard deviation $\bar{\sigma}_3 = 1112.402$ \\
 \hline
\multirow{2}{*}{$F_4$} & \multirow{2}{*}{cdf of operational risk $L_4$} & \multirow{2}{*}{lognormal cdf} & mean $\bar{m}_4 = 840.735$; \\
 & & & standard deviation $\bar{\sigma}_4 = 694.613$ \\
 \hline
\multirow{2}{*}{$F_5$ }& \multirow{2}{*}{cdf of business risk $L_5$} & \multirow{2}{*}{lognormal cdf} &  mean $\bar{m}_5 = 743.345$; \\
 & & & standard deviation $\bar{\sigma}_5 = 465.064$ \\
 \hline
\multirow{2}{*}{$F_6$} & \multirow{2}{*}{cdf of insurance risk $L_6$} & \multirow{2}{*}{lognormal cdf} &   mean $\bar{m}_6 = 438.978$; \\
 & & & standard deviation $\bar{\sigma}_6 = 111.011$ \\
\hline
\hline
\multirow{2}{*}{$C_0$} & reference copula  & \multirow{2}{*}{student-t copula} & with 6 degrees of freedom \\
& linking $L_1,\dots,L_6$ & &  and correlation matrix $\Sigma_0$ \\
\end{tabular}
\caption{Overview of the information concerning the reference distribution in the DNB case study. The correlation matrix $\Sigma_0$ is given in Appendix \ref{ssec:CorrelationMatrix}. $F_i$ denotes the cumulative distribution function of the marginal probability measure $\bar{\mu}_i$ for $i=1,\dots,6$.}
\label{tab:givenInfo}
\end{center}
\end{table}

For the purpose of risk management, the DNB needs to determine the capital to be reserved. 
According to the \citeA{basel2013statistical}, this capital requirement should be computed by the Average Value at Risk (AVaR) of the sum of these six losses.\footnote{\citeA{aas2014bounds} focus on the Value at Risk (VaR) rather than the AVaR. Since the Basel Committee on Banking Supervision recently shifted the quantitative risk metrics system from VaR to Expected Shortfall ~\cite<see>{chang2016choosing}, which is equivalent to the AVaR, we consider the AVaR in our study.} The AVaR of the sum of these six losses at a specific confidence level $\alpha$ is defined as
\begin{align} \label{eq:VaR_DNBexample}
\text{AVaR}_\alpha\left( L_6^+ \right) = \min_{\tau \in \mathbb{R}} \left\lbrace \tau + \frac{1}{1-\alpha} \mathbb{E} \left[ \max\right(L_6^+ -\tau,0\left) \right] \right\rbrace,
\end{align} 
where $L_6^+ := \sum_{i=1}^6 L_i$.
To evaluate expression \eqref{eq:VaR_DNBexample}, the joint distribution of $L_1,\dots,L_6$ is needed. As the  marginal distributions of $L_1,\dots,L_6$ are known, the DNB relies on the concept of copulas to model the dependence structure between these risks. From the above description, it is clear that joint observations of the $L_1,\dots,L_6$ are not available. Hence, standard techniques to determine the copula, e.g., by fitting a copula family and the corresponding parameters to a multivariate data set, do not apply. A panel of experts at the DNB therefore chooses a specific \emph{reference copula} $C_0$, in this case a student-t copula with six degrees of freedom and a particular correlation matrix. Such an approach is common in practice and referred to as \emph{expert opinion}.

From an academic point of view, this method for risk aggregation is not very satisfying due to the fact that the experts' choice of a \emph{reference dependence structure} between the different risk types might be very inaccurate. Hence, we say that there is \emph{model ambiguity} with respect to the dependence structure. It should be emphasized that a misspecification of this \emph{reference copula} chosen by expert opinion can have a significant impact on the aggregated risk and therefore on the required capital. Table \ref{tab:AVaR0.95} supports this statement by comparing the AVaR implied by the reference copula $C_0$ to the AVaR implied by other dependence structures: Without any information regarding the dependence structure between the six risk, the lower (resp.~upper) bound for the AVaR with confidence level $\alpha=0.95$ is 24165.52 (resp. 36410.12) million Norwegian kroner. Similar bounds are studied in \citeA{aas2014bounds}. As we pointed out in the literature review in Section \ref{ssec:literature}, these bounds have been criticized in the literature since they are too far apart for practical purposes. We therefore apply the results derived in this paper to compute bounds for the AVaR which depend on the level $\rho$ of distrust concerning the reference copula $C_0$. Alternatively, the parameter $\rho$ can be understood as the level of ambiguity with respect to the reference distribution $\bar{\mu}$.
\vspace{5mm}

\begin{table}
\begin{center}
\begin{tabular}{cccc}
$\inf_{C \in \mathcal{C}}  \text{AVaR}_\alpha^C (L_6^+)$ & $\text{AVaR}_\alpha^\Pi (L_6^+)$ & $\text{AVaR}_\alpha^{C_0} (L_6^+)$ & $\sup_{C \in \mathcal{C}}  \text{AVaR}_\alpha^C (L_6^+)$ \\
\hline
\hline
24165.52 & 26980.64 & 30498.94 & 36410.12
\end{tabular}
\caption{ Note that we set $\alpha = 0.95$. We use the rearrangement algorithm (see Aas and Puccetti, 2014)  to approximate $\inf_{C \in \mathcal{C}}  \text{AVaR}_\alpha^C (L_6^+)$, while $\sup_{C \in \mathcal{C}}  \text{AVaR}_\alpha^C (L_6^+) = \sum_{i=1}^6  \text{AVaR}_\alpha(L_i)$. The two remaining entries are computed by averaging over 50 simulation runs where 10 millions sample points are drawn in each run. Note that $\Pi$ denotes the independence copula. Thus, $\text{AVaR}_\alpha^\Pi (L_6^+)$ corresponds to the AVaR of the sum of the six losses given that they are independent. }
\label{tab:AVaR0.95}
\end{center}
\end{table}

We define the probability measure $\bar{\mu}$ of the reference distribution by the following joint cumulative distribution function
$$ \bar{F}(x) = C_0 (F_1(x_1),F_2(x_2),\dots,F_6(x_6),$$
for all $x \in \R^6$. Hence, the cdfs of the marginals $\bar{\mu}_i$ are given by $F_i(\cdot)$ for $i=1,2,\dots,6$.
The problem of interest can be formulated as follows:
 \begin{align}
  \underline{\Phi}_4^{C_0}(\alpha,\rho) &:= \inf_{\substack{L_6^+ \sim \mu \in \Pi(\bar{\mu}_1,\dots,\bar{\mu}_6),\\ d_c(\bar{\mu},\mu)\leq \rho}} \text{AVaR}_\alpha\left( L_6^+ \right), \label{eq:DNBbestAVaR} \\
 \overline{\Phi}_4^{C_0}(\alpha,\rho) &:= \sup_{\substack{L_6^+ \sim \mu \in \Pi(\bar{\mu}_1,\dots,\bar{\mu}_6),\\ d_c(\bar{\mu},\mu)\leq \rho}} \text{AVaR}_\alpha\left( L_6^+ \right). \label{eq:DNBworstAVaR}
\end{align}
The cost function $c$ defining the transportation distance $d_c$ in problem \eqref{eq:DNBbestAVaR} and \eqref{eq:DNBworstAVaR} is set to
\begin{align}
 c(x,y) = \sum_{i=6}^d \frac{\vert x_i - y_i \vert}{\bar{\sigma}_i}, \label{eq:DNBcost}
 \end{align}
where $\bar{\sigma}_i$ denotes the standard deviation of $\bar{\mu}_i$ and is given in Table \ref{tab:givenInfo}. The rational behind this definition of $c$ is that we want to model the ambiguity such that it concerns solely the dependence structure of the reference distribution. Definition \eqref{eq:DNBcost} is a simple way to achieve this.\footnote{It should be mentioned that \citeA{gao2017data} promote the definition $c(x,y) = \sum_{i=1}^6 \vert F_i(x_i) - F_i(y_i) \vert$, which implies that the transportation distance $d_c$ is defined directly on the level of copulas. Even if this approach is arguably more intuitive, we stick to definition \eqref{eq:DNBcost} mainly for the sake of computational efficiency.}

\begin{figure}
\begin{center}
\includegraphics[width=0.95\textwidth]{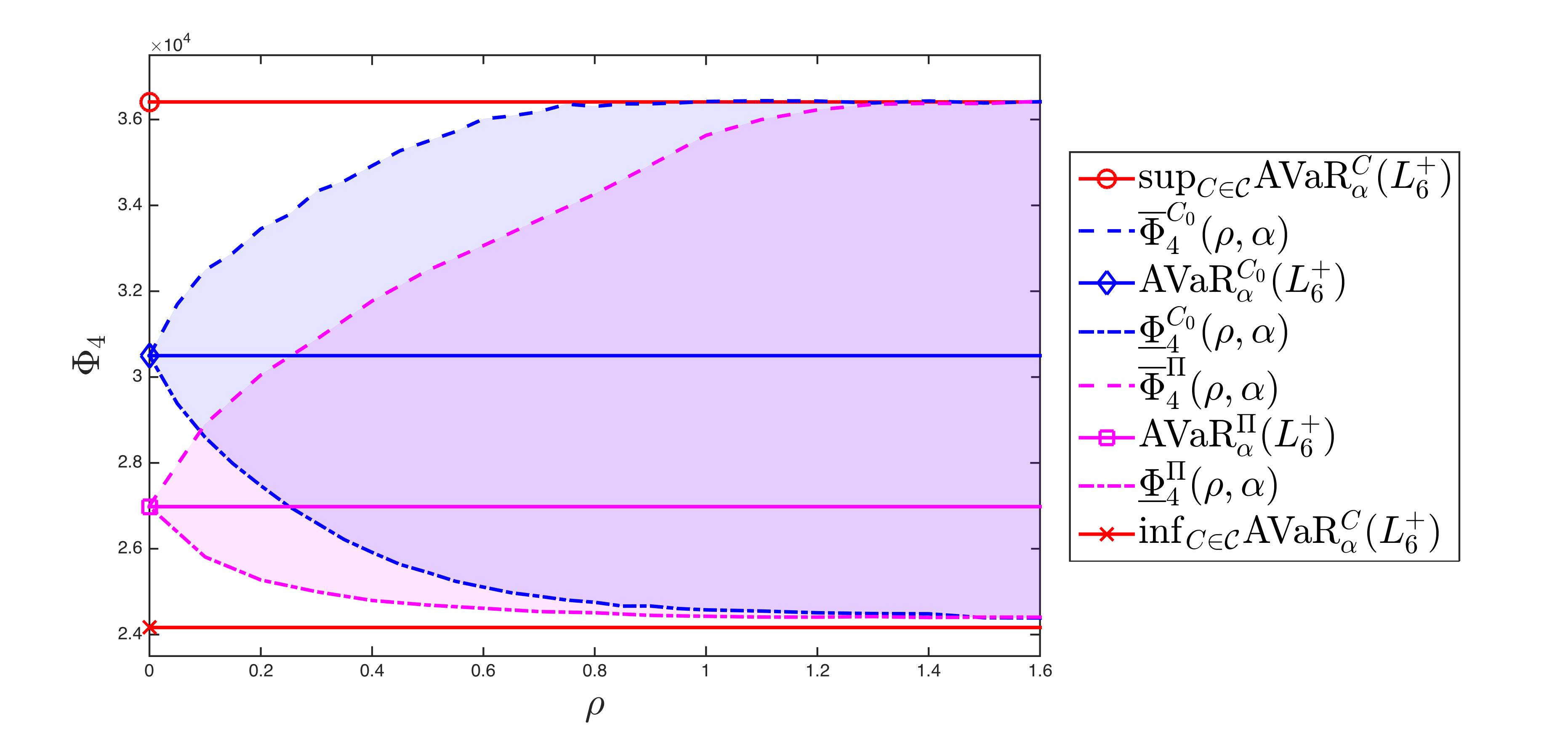}
\caption{We consider two distinct reference dependence structures, the student-t copula $C_0$ defined in Table \ref{tab:givenInfo} and the independence copula $\Pi$. 
The corresponding robust solutions $\underline{\Phi}_4^{C_0}(\alpha,\rho)$ and $\overline{\Phi}_4^{C_0}(\alpha,\rho)$, defined in \eqref{eq:DNBbestAVaR}, and \eqref{eq:DNBworstAVaR} resp. $\underline{\Phi}_4^{\Pi}(\alpha,\rho)$ and $\overline{\Phi}_4^{\Pi}(\alpha,\rho)$, defined analogously, are plotted as a function of the level of ambiguity $\rho$. We compare these results, which were computed relying on the concept presented in this paper, to the known values of AVaR$_\alpha(L_6^+)$ given in Table \ref{tab:givenInfo}. Note that we fix $\alpha = 0.95$.}
\label{fig:DNBplot}
\end{center}
\end{figure}

Figure \ref{fig:DNBplot} shows the numerical solutions of problems  \eqref{eq:DNBbestAVaR} and \eqref{eq:DNBworstAVaR}, which are computed relying on penalization and neural networks,  as a function of $\rho$ and for $\alpha = 0.95$. As a comparison, the same problem is also solved with respect to the independence coupling $\Pi$ rather than the reference copula $C_0$ described in Table \ref{tab:givenInfo}. The shaded regions outline the possible levels of risk for a given level of ambiguity $\rho$ and the two reference structures. On one hand, the evolution of the risk levels in $\rho$, combined with the given optimizers of problems \eqref{eq:DNBbestAVaR} and \eqref{eq:DNBworstAVaR} can be used as an informative tool to better understand the risk the DNB is exposed to. On the other hand, if a certain level of ambiguity is justified in practice, the bank can assign their capital based on the corresponding worst-case value. If for example $\rho=0.1$ is decided on, the bank would have to assign 32490 capital compared to 30499 as dictated by the reference structure $C_0$.

Analytically, one striking feature of the numerical solution with respect to $C_0$ is worth pointing out: The absolute upper bound is attained already for $\rho \approx 0.8$, while the distance from the reference measure to the comonotone joint distribution can be calculated to be around $1.7$. This underlines the fact that even though the comonotone distribution is a maximizer of the worst case AVaR, there are several more,
and they may be significantly more plausible structurally than the comonotone one.

\vspace{5mm}
In conclusion, this paper introduces a flexible framework to aggregate different risks while accounting for ambiguity with respect to the chosen dependence structure between these risks. The proposed numerical method allows us to perform this task without making restrictive assumptions about either the particular form of the aggregation functional, or the considered distributions, or the specific way to account for the model ambiguity.

\appendix
\section{Appendix}

	\subsection{A basic inequality between difference in $p$-th moment and difference in $p$-norm}
	\label{subsec:Inequality}
	The following result is used in Remark \ref{rem:errors}. The statement and proof are taken from \citeA{475146}.
	\begin{lemma} \label{lem:Inequality}
		Let $p\in \mathbb{N}$ and $X,Y \in L^p$, then $$ \| X^p - Y^p \|_1 \leq \| X - Y \|_p 
		\sum\limits_{k=0}^{p-1}\Vert X\Vert_p^k\Vert Y\Vert_p^{p-1-k}.
		$$
	\end{lemma}
	\begin{proof}
		For $p=1$ the inequality obviously holds. Let $p\geq 2$, then
		\begin{align}
		\Vert X^p-Y^p\Vert_1
		&=\left\Vert(X-Y)\sum\limits_{k=0}^{p-1}X^k Y^{p-1-k}\right\Vert_1
		\leq \Vert X-Y\Vert_p\sum\limits_{k=0}^{p-1}\left\Vert |X|^{kq} |Y|^{(p-1-k)q}\right\Vert_1^{1/q}. \label{eq:InequFirst}
		\end{align}
		It holds for all $k\in \{0, ..., p-1\}$  that
		\begin{align}
		\left\Vert |X|^{kq} |Y|^{(p-1-k)q}\right\Vert_1\leq\Vert X\Vert_p^{kq}\Vert Y\Vert_p^{p-kq}. \label{eq:InequToShow}
		\end{align}
		For $k\in \{0, p-1\}$ \eqref{eq:InequToShow} is immediate. For $0<k<p-1$ \eqref{eq:InequToShow} follows by H\"{o}lder's inequality applied with $r = p/kq$. Putting \eqref{eq:InequFirst} and \eqref{eq:InequToShow} together, we obtain
		\begin{align*}
		\Vert X^p-Y^p\Vert_1
		&\leq \Vert X-Y\Vert_p\sum\limits_{k=0}^{p-1}\left(\Vert X\Vert_p^{kq}\Vert Y\Vert_p^{p-kq}\right)^{1/q}\leq \Vert X-Y\Vert_p\sum\limits_{k=0}^{p-1}\Vert X\Vert_p^{k}\Vert Y\Vert_p^{p-1-k}
		\end{align*}
		and thus the claim.
	\end{proof}

	\subsection{Convergence analysis for Example \ref{ssec:MaxUniforms}}
	\label{subsec:convergenceanalysis}
	This section is meant to demonstrate how to assess the quality of the obtained numerical solution. We consider the case $\rho = 0.25$ in Figure \ref{fig:Example1}. We use the sampling measure $\theta^{\text{prod}}$ and compare the following three parameter settings:
	\begin{enumerate}[(i)]
	\item $\gamma = 100$, batch size = $1024$, $N_0 = 15000$ and $N_{\text{fine}} = 5000$.
	\item $\gamma = 2500$, batch size = $1024$, $N_0 = 15000$ and $N_{\text{fine}} = 5000$.
	\item $\gamma = 2500$, batch size = $16$, $N_0 = 7500$ and $N_{\text{fine}} = 2500$.
	\end{enumerate}
	Figure \ref{fig:Setting(i)and(ii)} examines the contrast between setting (i) and (ii). As can be seen, in both settings the algorithm  appears to converge in a stable way. In setting (i), there is however an apparent difference between the dual value $\phi_{\theta, \gamma}(f_1)$ and the primal value $\int f_1 d\mu^\star$, which is computed using the worst case distribution $\mu^\star$. We can therefore conclude that the penalization in setting (i) is insufficient, i.e. the penalization parameter $\gamma$ is chosen too low. This is clearly not the case for setting (ii). Figure \ref{fig:Setting(iii)} shows that both a small batch size and a small number of iterations lead to bad numerical behavior.

	\begin{figure}[t]
		\begin{minipage}[b]{0.5\textwidth} 
			\includegraphics[width=\textwidth,height=0.7\textwidth]{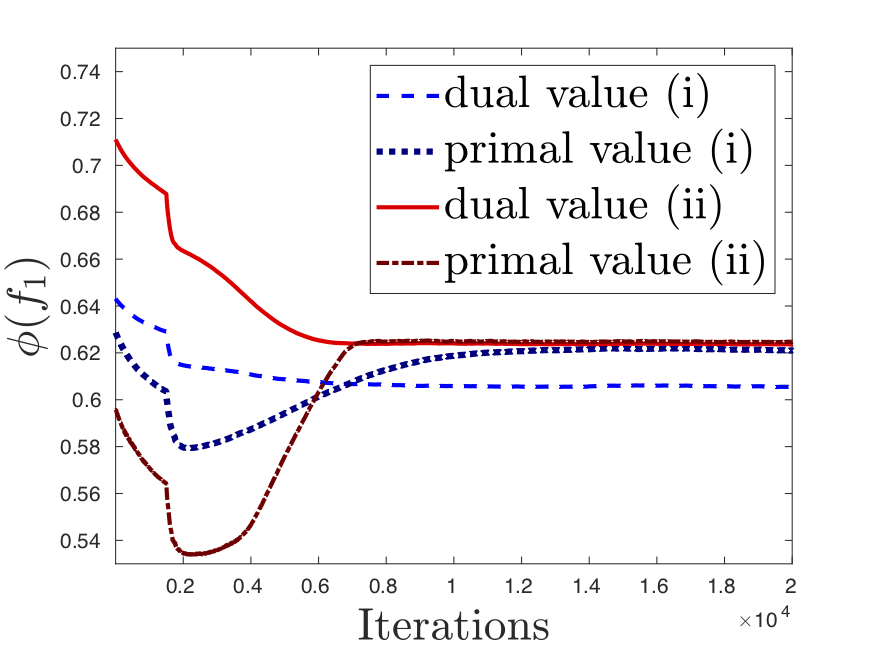}
		\end{minipage}\hspace{-1mm}
		\hfill
		\begin{minipage}[b]{0.5\textwidth}
			\includegraphics[width=\textwidth,height=0.7\textwidth]{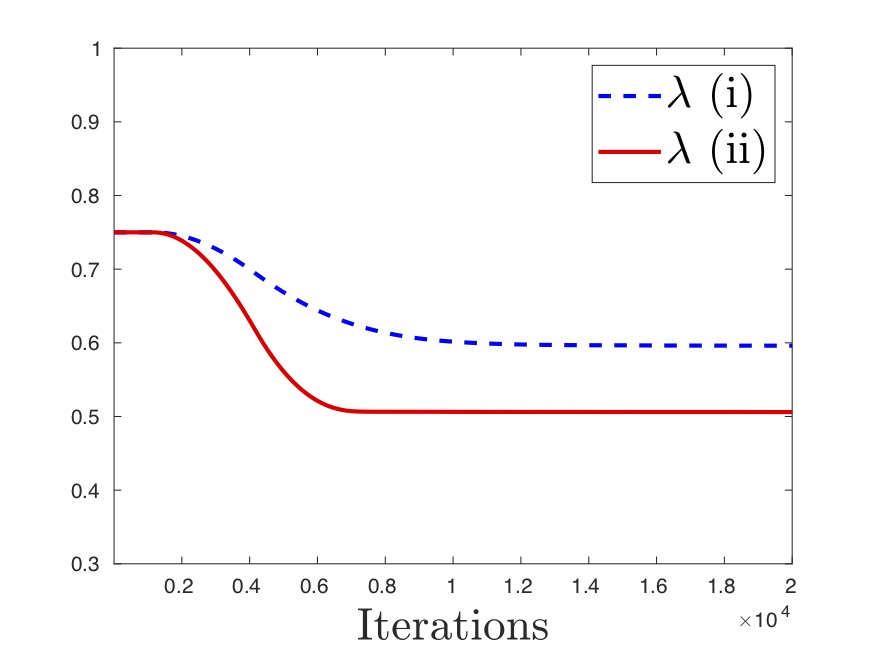}
		\end{minipage}\hspace{-1mm}
		\hfill
		\vspace{3mm}
		\begin{minipage}[b]{0.5\textwidth}  
			\includegraphics[width=\textwidth,height=0.7\textwidth]{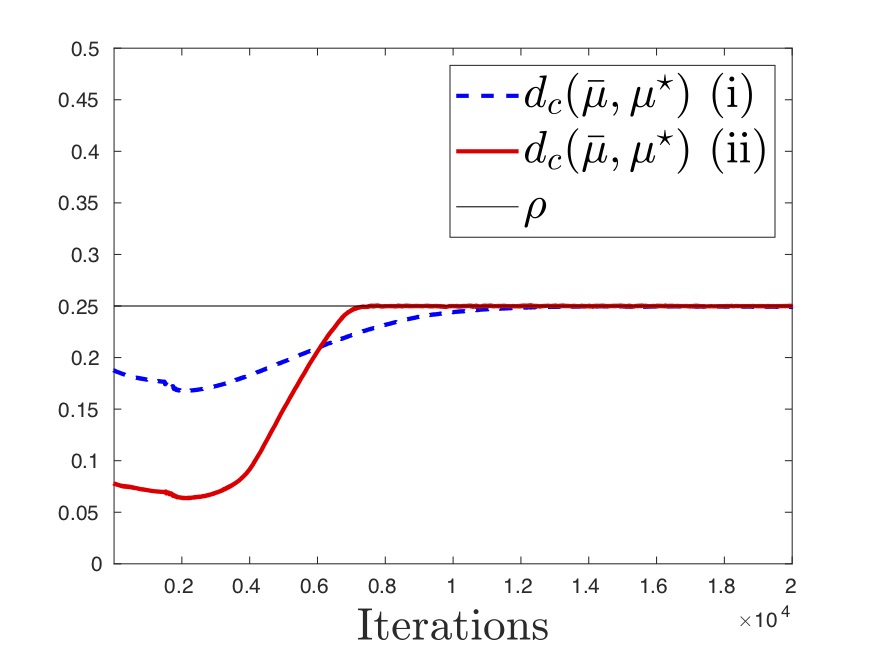}
		\end{minipage}
		\hspace{-3mm}
		\begin{minipage}[b]{0.5\textwidth}  
			\includegraphics[width=\textwidth,height=0.7\textwidth]{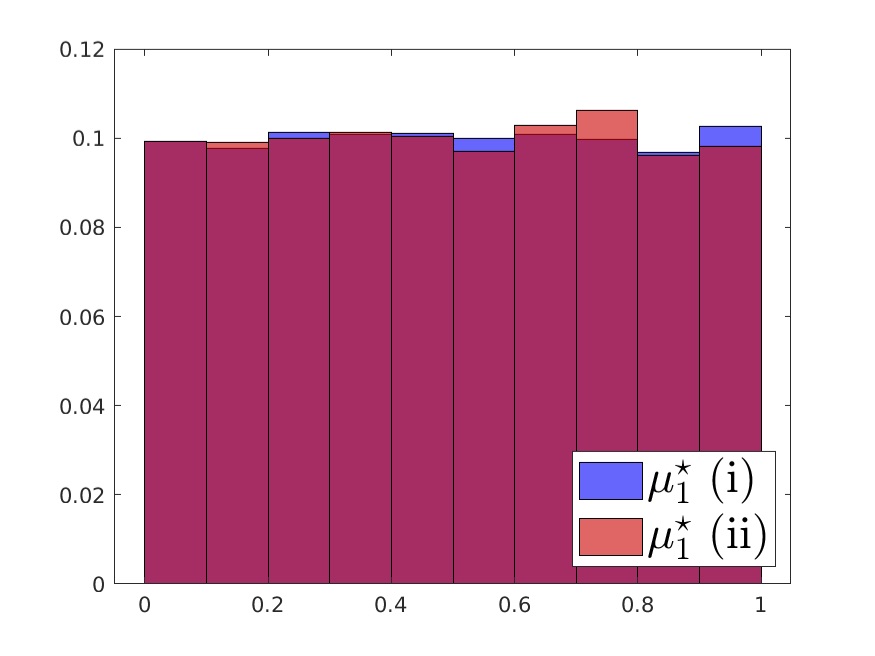}
		\end{minipage}
		\caption{Comparison of the parameter setting (i) with $\gamma = 100$ and (ii) with $\gamma = 2500$, while batch size = $1024$, $N_0 = 15000$ and $N_{\text{fine}} = 5000$ in both settings. The upper left panel shows the dual value $\phi_{\theta, \gamma}(f_1)$ as well as the primal value $\int f_1 d\mu^\star$. The upper right resp.~lower left panel illustrates the convergence of $\lambda$ resp.~$d_c(\bar{\mu},\mu^\star)$. The lower right panel plots 5000 samples from the first marginal $\mu_1^\star$ of the worst case distribution $\mu^\star$. Note that this histogram is also representative for the second marginal $\mu_2^\star$. The computation time is 205 seconds in both cases. }
		\label{fig:Setting(i)and(ii)}
	\end{figure}
	
	\begin{figure}[t]
		\begin{minipage}[b]{0.5\textwidth} 
			\includegraphics[width=\textwidth,height=0.7\textwidth]{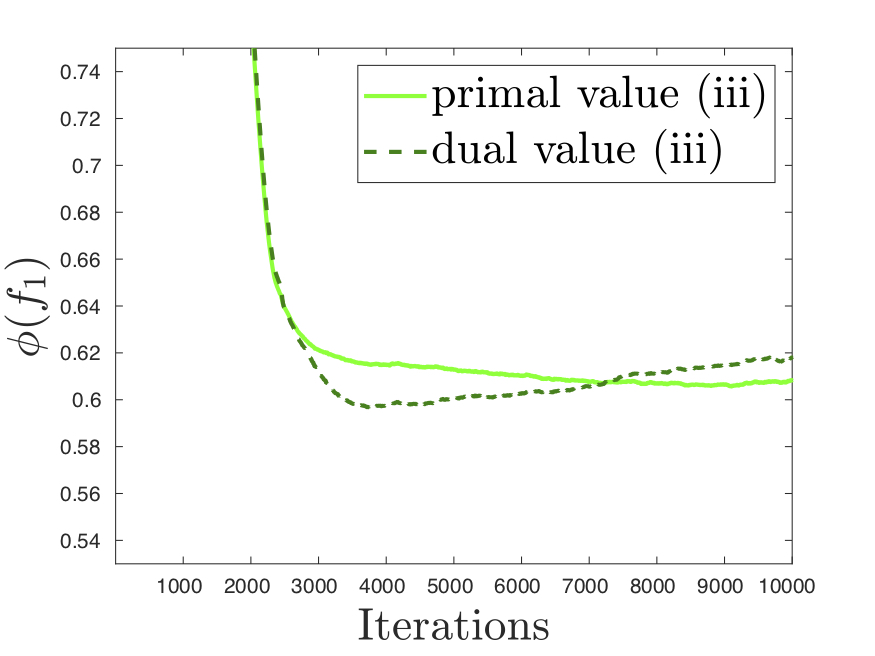}
		\end{minipage}\hspace{-1mm}
		\hfill
		\begin{minipage}[b]{0.5\textwidth}
			\includegraphics[width=\textwidth,height=0.7\textwidth]{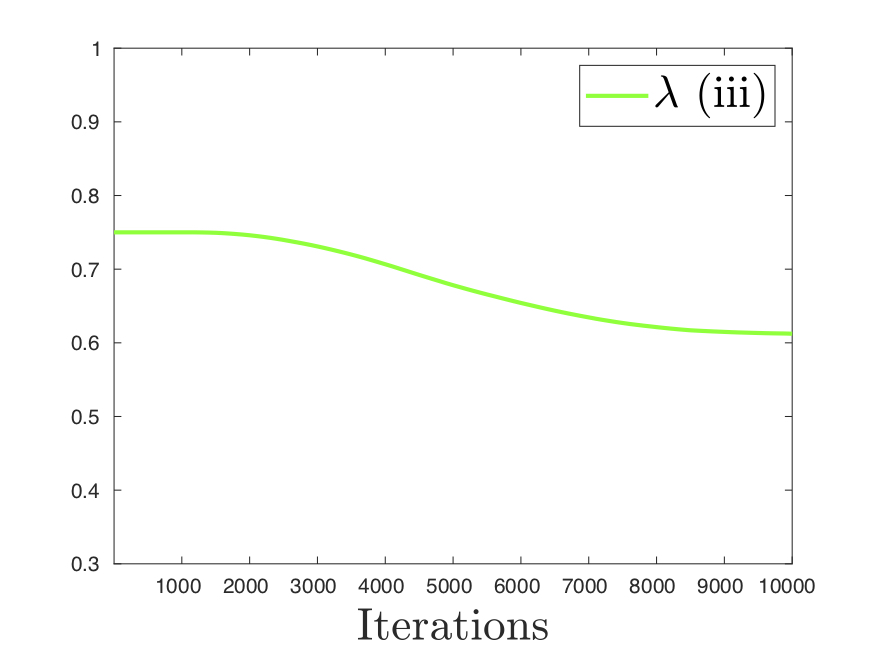}
		\end{minipage}\hspace{-1mm}
		\hfill
		\vspace{3mm}
		\begin{minipage}[b]{0.5\textwidth}  
			\includegraphics[width=\textwidth,height=0.7\textwidth]{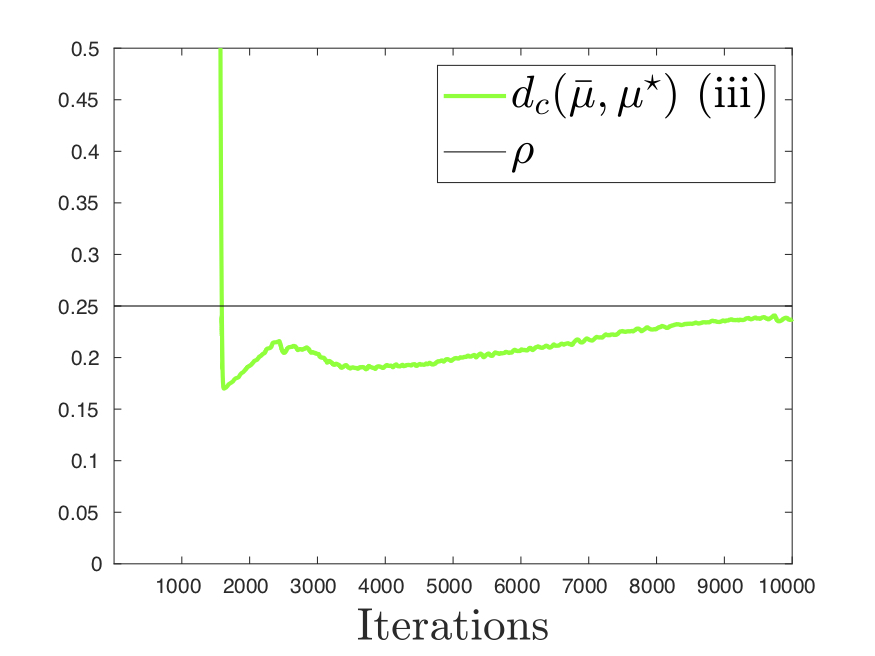}
		\end{minipage}
		\hspace{-3mm}
		\begin{minipage}[b]{0.5\textwidth}  
			\includegraphics[width=\textwidth,height=0.7\textwidth]{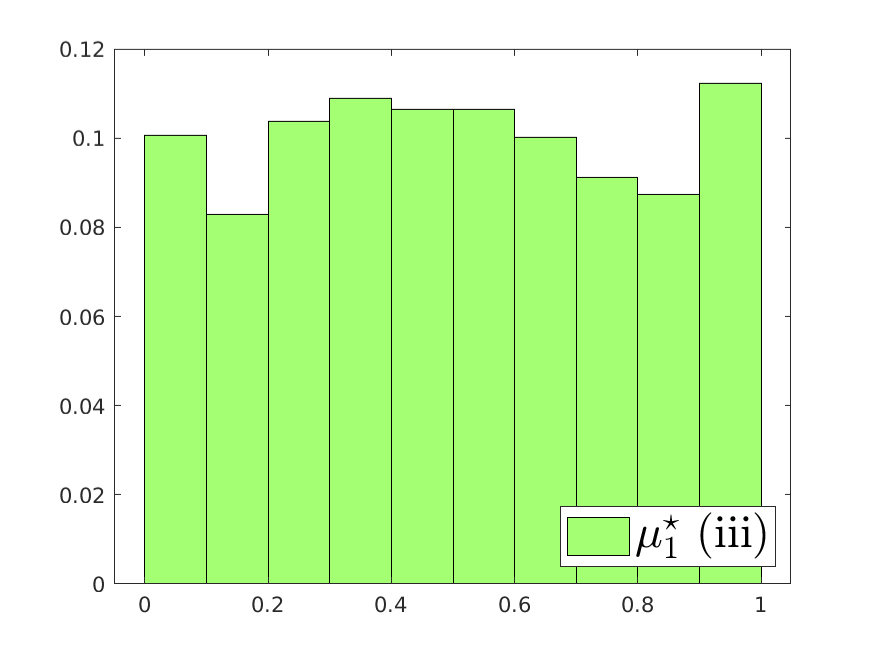}
		\end{minipage}
		\caption{Convergence analysis of the parameter setting (iii) where $\gamma = 2500$, batch size = $16$, $N_0 = 7500$ and $N_{\text{fine}} = 2500$. The upper left panel shows the dual value $\phi_{\theta, \gamma}(f_1)$ as well as the primal value $\int f_1 d\mu^\star$. The upper right resp.~lower left panel illustrates the convergence of $\lambda$ resp.~$d_c(\bar{\mu},\mu^\star)$. The lower right panel plots 5000 samples from the first marginal $\mu_1^\star$ of the worst case distribution $\mu^\star$. Note that this histogram is also representative for the second marginal $\mu_2^\star$. The computation time is 45 seconds.}		
		\label{fig:Setting(iii)}
	\end{figure}

\subsection{Proof for Section \ref{ssec:MaxUniforms} }
\label{ssec:ProofEx1}
We want to derive the analytic solution of problem \eqref{eq:ProblemExampleOne}. To do so, the concept of copulas turns out to be rather useful. We refer to \citeA{nelsen2007introduction} for an introduction to this topic. Let $\mathcal{C}$ denote the set of all copulas and let the comonotonic copula be denoted by $M(u_1,u_2) = \min(u_1,u_2)$, for all $u_1,u_2 \in [0,1]$. Using this notation, we can rewrite problem \eqref{eq:ProblemExampleOne} and show the following:
$$
\phi_1(f) = \sup_{\substack{C \in \mathcal{C},\\ d_c(M,C) \leq \rho}} \int_{[0,1]^2} \max(u_1,u_2) dC(u_1,u_2) = \frac{1+\min(\rho,0.5)}{2}.
$$
\begin{proof}
First, we derive an upper bound for $d_c(M,C)$, where $C \in \mathcal{C}$. Since $M$ lives on the main diagonal of the unit square, the vertical (or horizontal) projection of the mass of an arbitrary copula $C\in \mathcal{C}$ onto $M$ is feasible transportation plan with costs $\int_{[0,1]^2} \vert u_1 - u_2 \vert dC(u)$.\footnote{Note that $\vert u_1 - u_2 \vert$ is the distance (and thereby the cost of transportation) of any point-mass $C(u_1,u_2)$ to the main diagonal $M(u_1,u_2)$.} The latter expression appears in the definition of a concordance measure called \emph{Spearman's footrule} and it known to be maximized by the countermonotonic copula $W(u_1,u_2) := \max(u_1+u_2-1,0)$ for all $u_1,u_2 \in [0,1]$ cite<see>{liebscher2014copula}.
Hence, we obtain that $$d_c(M,C) \leq \int_{[0,1]^2} \vert u_1 - u_2 \vert dC(u) \leq \int_{[0,1]^2} \vert u_1 - u_2 \vert dW(u) = \int_0^1 |2u_1 -1| du_1 = 0.5.$$
Second, we show that this upper bound is attend for $d_c(M,W)$. The Kantorovich Rubinstein duality yields that
\begin{align*}
d_c(M,W) &= \sup_{\vert h(u)-h(v) \vert \leq c(u,v)}  \int_{[0,1]^2} h(u) dM(u) - \int_{[0,1]^2} h(v) dW(v) \\
&\geq \int_{[0,1]^2} u_1+u_2 dM(u) - \int_{[0,1]^2} v_1 + v_2 dW(v) = 1 - 0.5 = 0.5,
\end{align*} 
where we simply set $h(u) = u_1 +u_2$ to obtain the inequality. Since $d_c(M,C) \leq 0.5$ for all $C \in \mathcal{C}$, we have that $d_c(M,W)=0.5$.

Combining these two observations yields for $\rho > 0.5$ that \begin{align*}
\phi(f_1) &= \sup_{C \in \mathcal{C}} \int_{[0,1]^2} \max(u_1,u_2) dC(u_1,u_2) =  \int_{[0,1]^2} \max(u_1,u_2) dW(u_1,u_2) = \frac{3}{4}.
\end{align*}
It follows that we can assume $\rho \leq 0.5$ for the remainder of the proof.

Let us define the copula $R_\alpha$ as follows:
\begin{align*}
R_\alpha(u_1,u_2) = \begin{cases} W(u_1,u_2) &\mbox{if } \frac{1-\alpha}{2} \leq u_1,u_2 \leq \frac{1+\alpha}{2} \\
M(u_1,u_2) &\mbox{else} \end{cases},
\end{align*}
for $\alpha\in [0,1]$.
Using the same projection-argument as in the beginning of the proof, it follows that
\begin{align*} d_c(M,R_\alpha) &\leq \int_{[0,1]^2} \vert u_1 - u_2 \vert dR_\alpha(u) = \int_{(1-\alpha)/2}^{(1+\alpha)/2} |2u_1 -1| du_1 = \alpha^2 /2.
\end{align*} Thus, $d_c(M,R_{\sqrt{2\rho}}) \leq \rho$, which implies
$$ \phi(f_1) \geq \int_{[0,1]^2} \max(u_1,u_2) dR_{\sqrt{2\rho}}(u_1,u_2) = \frac{1+\rho}{2}.$$ 

By Corollary \ref{coro:DualityWassersteinball}, we have that
\begin{align*}
\phi(f_1) = \inf_{\lambda \geq 0, h_i \in C([0,1])} &\left\lbrace \lambda \rho + \sum_{i=1}^2 \int_0^1 h_i(u_i) du_i \right. \\
& + \left.\int_{[0,1]^2} \sup_{v \in [0,1]^2} \left[ \max(v_1,v_2) - \sum_{i=1}^2 h_i(v_i) - \lambda \sum_{i=1}^2 \vert u_i -v_i \vert \right] dM(u) \right\rbrace. \notag
\end{align*}
Plugging in the value $\lambda =0.5$ and setting $h_1(u) = h_2(u) =u/2$, yields $\phi_1(f) \leq \frac{\rho}{2} + \frac{1}{2} + 0$.
\end{proof}

\subsection{Proof for Section  \ref{ssec:AVaRUniforms} }
\label{ssec:ProofEx2}
We now derive the analytic bounds for $\Phi_2$, i.e.~the solution of problem \eqref{eq:ExampleTwo}, which are plotted in Figure \ref{fig:Example2}. 

Let us start by proving the following upper bound
\begin{align}
\label{eq:Phi2UpperBound}
\Phi_2 \leq \min\left(1+\alpha,2 - \frac{2}{3} \sqrt{2-2\alpha}+ \frac{\rho}{2(1-\alpha)}\right),\end{align}
where $\Phi_2$ is defined in \eqref{eq:ExampleTwo}.
\begin{proof}
Due to Corollary \ref{coro:DualityWassersteinball},
\begin{align} \label{eq:Example2Dual} 
\Phi_2= &\inf_{\tau,\lambda \geq 0, h_i \in C([0,1])} \left\lbrace \lambda \rho + \sum_{i=1}^2 \int_0^1 h_i(u_i) du_i \right. \\
& \hspace*{-6mm}+ \left.\int_{[0,1]^2} \sup_{v \in [0,1]^2} \left[ \tau + \frac{1}{1-\alpha}\max(v_1+v_2-\tau,0) - \sum_{i=1}^2 h_i(v_i) - \lambda \sum_{i=1}^2 \vert u_i -v_i \vert \right] d(u_1, u_2) \right\rbrace. \notag
\end{align}
The following choice of optimizers in equation \eqref{eq:Example2Dual} yields the upper bound for $\Phi_2$ given in \eqref{eq:Phi2UpperBound}:
\begin{align*}
\lambda = \frac{1}{2(1-\alpha)}, \qquad \tau = \tau^\star := 2 - \sqrt{2-2\alpha} \quad \text{and} \quad h_i(v) = \frac{1}{1-\alpha}\left(v-\frac{\alpha \tau^\star}{2}\right) \text{ for } i=1,2.
\end{align*}
\end{proof}

We now derive the following lower bound
\begin{align}\label{eq:Example2LowerBound}
 \Phi_2 \geq 
\min \left( 1+\alpha,2 - \frac{2}{3} \sqrt{2-2\alpha} +\frac{2 (-3 + 2 \sqrt{2-2\alpha} + 3\alpha) \rho}{3 (2- \alpha) (1-\alpha) \alpha}\right),\end{align}
where $\Phi_2$ is defined in \eqref{eq:ExampleTwo}.

\begin{proof}
It is straight forward to see that $\Phi_2$ is concave in the radius $\rho$ of the considered Wasserstein ball around $\bar{\mu}$. This is due to the fact that we defined the ground metric $c(\cdot,\cdot)$ of the transportation distance $d_c$ by the $L_1$-metric, i.e.~$c(x,y) = ||x-y||_1$.  Hence, to establish the lower bound \eqref{eq:Example2LowerBound}, we only need to show that for $\rho^\star = \alpha (1-\alpha) (1-\alpha/2)$ it holds that $\Phi_2 \geq 1 + \alpha$.

Therefore, we define the probability measure $\mu_\alpha$ by the following bivariate copula
$$
C_\alpha(u_1,u_2) = \begin{cases}
u_1 u_2 &\mbox{if } u \in [0,\alpha/2]^2 \cup [\alpha/2,\alpha]^2 \\
\frac{2-\alpha}{\alpha} u_1 u_2 &\mbox{if } u \in \left( [0,\alpha/2] \times [ \alpha/2,\alpha] \right) \cup \left( [ \alpha/2,\alpha]\times[0,\alpha/2] \right) \\
\frac{1}{1-\alpha} u_1 u_2 &\mbox{if } u \in [\alpha,1]^2 \\
\min(u_1,u_2) &\mbox{else}
\end{cases}.
$$
Tedious calculations show that $d_c(\bar{\mu},\mu_\alpha) \leq \alpha (1-\alpha) (1-\alpha/2) = \rho^\star$, where $\bar{\mu}$ is the bivariate probability measure with independent, standard uniformly distributed marginals defined in problem \eqref{eq:ExampleTwo}. Moreover, for $\binom{V}{U} \sim \mu_\alpha$ it holds that AVaR$_\alpha(U+V) = 1+\alpha$.
\end{proof}

\subsection{Correlation Matrix}
\label{ssec:CorrelationMatrix}
The purpose of this subsection is to give the correlation matrix $\Sigma_0$. Recall that $\Sigma_0$ defines the student-t copula $C_0$ with six degrees of freedom used as a reference dependence structure in the case study by \citeA{aas2014bounds}, which we consider in Section \ref{sec:DNB}. As this matrix is not given in the paper by \citeA{aas2014bounds}, we simply choose the following arbitrary correlation matrix
\begin{align*}
\Sigma_0 = 
  \left( {\begin{array}{cccccc}
    1    &   0.36 &   0.35  &  0.44  &  0.45  &  0.30 \\
    0.36 &   1    &   0.37  &  0.36  &  0.41  &  0.43 \\
    0.35 &   0.37 &   1     &  0.44  &  0.32  &  0.42 \\
    0.44 &   0.36 &   0.44  &  1     &  0.41  &  0.29 \\
    0.45 &   0.41 &   0.32  &  0.41  &  1     &  0.28 \\ 
    0.30 &   0.43 &   0.42  &  0.29  &  0.28  &  1
  \end{array} } \right).
\end{align*}

\section*{Acknowledgments}
The authors confirm that the data supporting the findings of this study are available within the article by \citeA{aas2014bounds}. 

The authors thank Daniel Bartl, Ludovic Tangpi, Ruodu Wang and the participants of the numerous conference and seminars, where the authors presented this paper, for helpful comments as well as interesting discussions.
Moreover, Mathias Pohl thanks Philipp Schmocker for his help and acknowledges support by the Austrian Science Fund (FWF) under the grant P28661 and Stephan Eckstein sincerely thanks Jan Ob\l\'{o}j for his hospitality. 
Finally, we thank the two referees for their helpful comments.

\bibliographystyle{apacite}
\bibliography{references}

\end{document}